\documentclass{article}

\usepackage{esvect}
\usepackage{hhline}

\usepackage{amsfonts,amsmath,amsthm,amssymb}
\usepackage{paralist}
\usepackage{multirow}
\usepackage{comment}
\usepackage{xspace}
\usepackage[noend,ruled]{algorithm2e}
\usepackage{colortbl}
\usepackage{tikz}
\usetikzlibrary{decorations,arrows,petri,topaths,backgrounds,shapes,positioning,fit,calc,decorations.pathreplacing,patterns,intersections}
\usepackage[noend,ruled]{algorithm2e} 
\usepackage{multirow}

\usepackage[numbers,sectionbib]{natbib}
\usepackage{subcaption}

\usepackage{authblk}

\usepackage{esvect}
\usepackage{mathtools}

\usepackage{hhline}
\usepackage{booktabs}
\usepackage{xcolor}

\newlength{\additionaltextwidth}
\usepackage[colorlinks,citecolor=green!60!black,linkcolor=red!60!black,pagebackref]{hyperref}
\newcommand{\mytitle}{Prices Matter for the Parameterized Complexity~of Shift~Bribery}
\title{\mytitle}
\author[1]{Robert Bredereck}
\author[1]{Jiehua Chen}
\author[2]{Piotr Faliszewski}
\author[1]{Andr\'e Nichterlein}
\author[1]{Rolf~Niedermeier}

\affil[1]{Institut f\"ur Softwaretechnik und Theoretische Informatik,
  TU Berlin, Berlin, Germany
  \texttt{\{robert.bredereck, jiehua.chen, andre.nichterlein, rolf.niedermeier\}@tu-berlin.de}}
\affil[2]{AGH University of Science and Technology, Krakow, Poland
  \texttt{faliszew@agh.edu.pl}}
\date{}

\newcommand{\probShiftBBorda}{\textsc{Borda Shift Bribery}\xspace}
\newcommand{\probShiftBMaximin}{\textsc{Maximin Shift Bribery}\xspace}

\newcommand{\probShiftBCopelandAbbrv}{\textsc{Copeland Shift Bribery}\xspace}
\newcommand{\probShiftBCopeland}{\probShiftBCopelandAbbrv}
\newcommand{\probClique}{\textsc{Clique}\xspace}
\newcommand{\probSC}{\textsc{Set Cover}\xspace}
\newcommand{\probMCClique}{\textsc{Multicolored Clique}\xspace}

\newcommand{\clique}{\ensuremath{\mathcal{Q}}}
\newcommand{\setFamily}{\ensuremath{\mathcal{S}}}
\newcommand{\setUniverse}{\ensuremath{\mathcal{U}}}
\newcommand{\setcover}{\ensuremath{\mathcal{Q}}}

\newcommand{\Bribery}{\textsc{Bribery}\xspace}
\newcommand{\DollarBribery}{\textsc{\${}Bribery}\xspace}
\newcommand{\ShiftBribery}{\textsc{Shift Bribery}\xspace}
\newcommand{\ShiftBriberyO}{\textsc{Shift Bribery(O)}\xspace}

\newcommand{\decprob}[3]{
  \begin{center}%
    \begin{minipage}{0.9\linewidth}%
      \textsc{#1}\\[0.2ex]
      \textbf{Input:} #2\\[0.2ex]
      \textbf{Question:} #3
    \end{minipage}%
  \end{center} }

\newcommand{\np}{{\mathrm{NP}}}
\newcommand{\classFPT}{{\mathsf{FPT}}}
\newcommand{\classXP}{{\mathsf{XP}}}
\newcommand{\wone}{{\mathsf{W[1]}}}
\newcommand{\wtwo}{{\mathsf{W[2]}}}
\newcommand{\w}{{\mathsf{W}}}

\newcommand{\classFPTAS}{\textsf{FPT-AS}\xspace}

\newcommand{\naturals}{{{\mathbb{N}}}}

\newcommand{\calR}{{{{\mathcal{R}}}}}

\newcommand{\pref}{\succ}
\newcommand{\score}{{{{\mathrm{score}}}}}
\newcommand{\s}[2]{\score_{#1}(#2)}

\newcommand{\shift}{{{{\mathrm{shift}}}}}
\newcommand{\costshift}{{{\Pi\mathrm{\hbox{-}shift}}}}
\newcommand{\costprimeshift}{{{\Pi'\mathrm{\hbox{-}shift}}}}
\newcommand{\sortedshift}{{{\mathrm{s\hbox{-}shift}}}}
\newcommand{\costsortedshift}{{{\Pi\mathrm{\hbox{-}s\hbox{-}shift}}}}

\newcommand{\OPT}{{{{\mathrm{OPT}}}}}

\newcommand{\all}{{{{\mathrm{all}}}}}
\newcommand{\unit}{{{{\mathrm{unit}}}}}
\newcommand{\sort}{{{{\mathrm{sort}}}}}
\newcommand{\convex}{{{{\mathrm{convex}}}}}

\newcommand{\probILPshort}{\textsc{ILP Feasibility}\xspace}

\newcommand{\edges}{\ensuremath{\mathrm{edges}}}
\newcommand{\structure}{\ensuremath{\mathrm{struct}}}

\newcommand{\sets}{\ensuremath{\mathrm{sets}}}

\newcommand{\point}{{{{\mathrm{point\text{-}pair}}}}}

\newcommand{\seq}[1]{\ensuremath{\langle #1\rangle}}
\newcommand{\revseq}[1]{\ensuremath{\overleftarrow{\langle #1\rangle}}}
\newcommand{\halfseqone}{\mathrm{half\text{-}seq1}}
\newcommand{\halfseqtwo}{\mathrm{half\text{-}seq2}}

\newcommand{\myparagraph}[1]{\paragraph{#1.}}

\usepackage{cleveref}

\newtheorem{definition}{Definition}
\newtheorem{example}{Example}
\newtheorem{proposition}{Proposition}
\newtheorem{theorem}{Theorem}

\newtheorem{corollary}{Corollary}

\crefname{subsection}{Section}{Sections}
\crefname{section}{Section}{Sections}

\crefname{table}{Table}{Tables}
\crefname{figure}{Figure}{Figures}
\crefname{algorithm}{Algorithm}{Algorithms}
\crefname{theorem}{Theorem}{Theorems}
\crefname{definition}{Definition}{Definitions}
\crefname{corollary}{Corollary}{Corollaries}
\crefname{proposition}{Proposition}{Propositions}
\crefname{obs}{Observation}{Observations}
\crefname{lemma}{Lemma}{Lemmas}
\crefname{example}{Example}{Examples}
\crefname{reduction}{Reduction}{Reductions}
\crefname{algocf}{Algorithm}{Algorithms}

\DeclareMathOperator{\col}{col}

\usepackage{ifthen}

\newboolean{ShowPaperOrganization}
\setboolean{ShowPaperOrganization}{true}

\newboolean{ShowCliqueSClongDefinitions}
\setboolean{ShowCliqueSClongDefinitions}{false}

\newboolean{ShowShiftBriberyExample}
\setboolean{ShowShiftBriberyExample}{false}

\newboolean{ShowPrfMaximinFPTnumShifts}
\setboolean{ShowPrfMaximinFPTnumShifts}{true}

\newboolean{ShowPrfFPTnumVotersAllNothingPrices}
\setboolean{ShowPrfFPTnumVotersAllNothingPrices}{true}

\newboolean{ShowPrfCopelandWoneNumVoters}
\setboolean{ShowPrfCopelandWoneNumVoters}{true}
\newboolean{ShowPrfSketchCopelandWoneNumVoters}
\setboolean{ShowPrfSketchCopelandWoneNumVoters}{true}

\newboolean{ShowFormalprfFptAsCandidates}
\setboolean{ShowFormalprfFptAsCandidates}{true}

\begin{document}

\maketitle
\thispagestyle{plain}
\setcounter{footnote}{0}

\begin{abstract}
 In the \ShiftBribery problem, we are given an election (based on
  preference orders), a preferred candidate~$p$, and a budget. The goal is to
  ensure~$p$'s victory by shifting~$p$ higher in some voters' preference orders.  However,
  each such shift request comes at a price (depending on the voter and
  on the extent of the shift) and we must not exceed the given budget. 
  We study the parameterized computational complexity
  of \ShiftBribery with respect to a number of parameters (pertaining
  to the nature of the solution sought and the size of the election)
  and several classes of price functions. 
  When we parameterize \ShiftBribery 
  by the number of affected voters, then for each of our voting rules
  (Borda, Maximin, Copeland) the problem is $\wtwo$-hard.
  If, instead,
  we parameterize by the number of positions by which $p$ is shifted in total,
  then the problem is fixed-parameter tractable for Borda and Maximin,
  and is $\wone$-hard for Copeland. If we parameterize by the budget, 
  then the results depend on the price function class. We
  also show that \ShiftBribery tends to be tractable when parameterized
  by the number of voters, but that the results for the number of
  candidates are more enigmatic.
\end{abstract}
\section{Introduction}

%\ShiftBribery models direct campaigning efforts in elections which may occur in various fields, such as elections in politics, aggregating product rankings on the Web\footnote{www.idealo.co.uk or www.idealo.de}, or determining the winner of music charts\footnote{www.aircheck.net.au}, to name but a few. 
%In such elections, the winner determination procedure seems to be immune to direct interference. For instance, the web company idealo claims that their results cannot be bribed.
%Thus, in order to let a specific candidate win the election,
%it is quite tempting to bribe the different voters to ``promote'' this candidate in the voter's opinions.

Rank aggregation and election-winner determination are of key
importance in various economical and political settings. For instance,
there are % the German website idealo.de offers
product rankings based on comparing their prices, their features, and
different tests~(performed by various institutions such as
foundations, journals, etc.); universities are judged based on
multiple different criteria (e.g., the number of students per faculty
member, availability of particular facilities, the number of
Nobel prize winners employed etc.); sport competitions involve
multiple rankings (for example, a Formula~1 season consists of about
twenty races, each resulting in a ranking of the drivers); and
political elections require members of the society to express
preferences regarding the participating candidates.  In each of these
cases the provided rankings are aggregated into the final one, often
of significant importance (for example, customers decide on their
purchases based on product rankings, students pick the 
best-ranked universities, 
the Formula~1 world champion is the driver who comes out
first in the aggregated ranking, and the most appreciated candidate
becomes the country's president).  A sophisticated way of dealing with
rankings based on multiple different criteria is to compute a
consensus ranking using preference-based rank aggregation
methods.\footnote{For example, the German website idealo.de aggregates
  different product tests by first translating the test results into a
  unified rating system and then taking the ``average'' of all the
  ratings.  Various university rankings are prepared in a similar
  way.  It would be very interesting, however, to utilize the rankings
  themselves, instead of the ratings, for the aggregation.  
  Moreover, Formula~1 racing (and numerous similar competitions) use
  pure ranking information (e.g., Formula~1 uses a very slightly
  modified variant of the Borda election rule).}  In order to affect the
outcome of the rank aggregation one has to influence the component
rankings obtained from different sources~(different product tests,
different judgment criteria, different races, different voters).
Clearly, the cost of influencing may differ from source to source and,
indeed, can sometimes be quite high. Nonetheless, the effect of
improved position in the final ranking can be very beneficial.
% It is emphasized by idealo.de that they cannot be bribed.

In this work, we study the 
% It is, however, of general interest how hard in terms of
computational complexity of affecting the outcome of the rank
aggregation by ``bribing'' specific agents to change their
rankings. Moreover, replacing ``bribery'' with ``product
development,'' or ``university expansion,'' or ``training,'' or
``political campaigning'' we see that our work is relevant to all the
settings mentioned above; the particular entities (companies offering
their products, universities, drivers, politicians) can find out how
much effort they need to invest in order to achieve a better position
in the aggregated ranking (or maybe even become the winner).  A
natural and simple model in this context, using the formalisms of
voting theory, is \ShiftBribery as introduced by Elkind et
al.~\citep{elk-fal-sli:c:swap-bribery}.  We extend their studies in
terms of charting the border of computational worst-case
tractability, herein putting particular emphasis on the voter-specific
``shifting prices'' (how expensive it is to shift a candidate by
$x$~positions ``up''). % within each voter (e.g., product tester).

%In this paper, we study the parameterized complexity of the \ShiftBribery problem that is first introduced by Elkind et al.~\citep{elk-fal-sli:c:swap-bribery}.
%Formally speaking, 
Informally (see Section~\ref{sec:shift-bribery} for a formal definition),
\ShiftBribery is the following decision 
problem:
\decprob{\ShiftBribery}
{An election, that is, a set of
candidates and a set of voters, each with a linear preference order
over the candidate set,
some preferred candidate~$p$, and some budget.
}
{Can we make $p$~win by bribing voters to shift $p$
higher in their preference orders
by ``paying'' no more than the given budget?}

%In \ShiftBribery, we are given an election, that is, a set of
%candidates and a set of voters, each with a linear preference order
%over the candidate set.\footnote{We assume that we have the knowledge
%  of the voters' preference orders (for example, from preelection
%  polls). Further, in our example settings often the full rankings are
%  known. For example, a driver preparing for a new Formula~1 season
%  has full knowledge of the results from the previous one.}
%% We study the parameterized complexity of the \ShiftBribery problem,
%% introduced by Elkind et al.~\citep{elk-fal-sli:c:swap-bribery}
%% to model
%% direct campaigning efforts in elections. %Specifically, 
%% In 
%% \ShiftBribery, we are given an election, that is, a set of candidates
%% and a collection of voters, each with a linear preference order over
%% the candidate set.
%Our goal is to ensure that our preferred candidate $p$ wins. 
We assume that we have the knowledge of the voters' preference orders (for example, from preelection
polls). Further, in our example settings often the full rankings are
known. For example, a driver preparing for a new Formula~1 season
has full knowledge of the results from the previous one.
Our \ShiftBribery problem models the situation where we
approach each of the voters, one-on-one, and try
to convince\footnote{What ``to convince'' means can vary a lot
  depending on the application scenario. On the evil side we have
  bribery, but it can also mean things such as product development,
  hiring more faculty members, training on a particular racing
  circuit, or explaining the details of one's political
  platform. Clearly, different ranking providers may appreciate
  different efforts, which is modeled by the individual price
  functions.}  him or her to rank $p$ higher. Naturally, the effect
(the number of positions by which $p$ is shifted in each voter's
preference order) depends on the voter's character and situation, and
on the amount of effort we invest into convincing the voter.  This
``effort'' could, for example, mean the amount of time spent, the cost
of implementing a particular change, or, in the bribery view of the
problem, the payment to the voter.  Thus, the computational complexity
of the problem depends on the voting rule used in the election, on
various election parameters such as the numbers of candidates and
voters, and on the type of price functions describing the efforts
needed to shift $p$ up by a given number of positions in the voters'
preference orders. Our goal is to unravel the nature of these
dependencies.

\myparagraph{Related Work}
The computational complexity of bribery in elections was first studied by
Faliszewski et al.~\cite{fal-hem-hem:j:bribery}. They considered the
$\Bribery$ problem, where one asks if it is possible to ensure that a
given candidate is an election winner by changing at most a given
number of votes. Its priced variant, \DollarBribery, is the same
except that  each voter has a possibly different price for
which we can change his or her vote. These problems were studied for
various election rules, including
Borda~\cite{fal-hem-hem:j:bribery,bre-fal-hem-sch-sch:c:approximating-elections},
Maximin~\cite{fal-hem-hem:j:multimode}, and 
Copeland~\cite{fal-hem-hem-rot:j:llull} (see \cref{sec:prelim} for
exact definitions of these rules).  
Recently, \citet{GerMacMagXiaYi2015} studied the bribery problem for linear ranking systems.
Notably,
the destructive variant of the \Bribery problem (known under the name
\textsc{Margin of Victory}), where the goal is to ensure that a
despised candidate does not win (and which was studied, e.g., by
Magrino et al.~\cite{mag-riv-she-wag:c:stv-bribery} and
Xia~\cite{xia:margin-of-victory}) has a surprisingly positive
motivation---it can be used to detect fraud in elections.

The above problems, however, do not take into account that the price
of bribing a voter may depend on what vote we wish the ``bribed'' voter to cast.
For example, a voter might be perfectly happy to swap the two least
preferred candidates but not the two most preferred ones.  To model
such situations, Elkind et al.~\cite{elk-fal-sli:c:swap-bribery}
introduced the \textsc{Swap Bribery} problem. They assumed that each
voter has a swap-bribery price function which gives the cost of
swapping each two candidates (provided they are adjacent in the
voter's preference order; one can perform a series of swaps to
transform the voter's preference order in an arbitrary way). They
found that \textsc{Swap Bribery} is both $\np$-hard and hard to
approximate for most well-known voting rules (essentially, because
the \textsc{Possible Winner}
problem~\cite{kon-lan:c:incomplete-prefs,con-xia:j:possible-necessary-winners,bet-dor:j:possible-winner-dichotomy,bau-rot:j:pos-win-dich},
which is $\np$-hard for almost all natural voting rules, is a special
case of \textsc{Swap Bribery} with each swap costing either zero or
infinity). 
Motivated by this, Dorn and
Schlotter~\cite{dor-sch:j:parameterized-swap-bribery} considered the
parameterized complexity of \textsc{Swap Bribery} for the
case of~$k$-Approval (where each voter gives a point to his or her top $k$ candidates).
In addition, Elkind et al.~\cite{elk-fal-sli:c:swap-bribery} also considered
\ShiftBribery, a variant of \textsc{Swap Bribery}
%. Intuitively, \ShiftBribery is a variant of \textsc{Swap  Bribery} 
where all the swaps have to involve the preferred candidate
$p$. %Elkind et al.~\cite{elk-fal-sli:c:swap-bribery} 
They have shown that
\ShiftBribery remains $\np$-hard for Borda, Maximin, and Copeland
but that there is a $2$-approximation algorithm for Borda and a
polynomial-time algorithm for the~$k$-Approval voting rule.
\ShiftBribery was further studied by Elkind and
Faliszewski~\cite{elk-fal:c:shift-bribery}, who viewed it as a
political campaign management problem (and whose view we adopt in this
paper), and who gave a $2$-approximation algorithm for all scoring
rules (generalizing the result for Borda) and other approximation
algorithms for Maximin and Copeland. Then, Schlotter et
al.~\cite{SchFalElk15} have shown that \ShiftBribery
is polynomial-time solvable for the case of Bucklin and Fallback
voting rules. 

Based on the idea of modeling campaign management as bribery
problems, other researchers introduced several variants of bribery
problems. For example, Schlotter et
al.~\cite{SchFalElk15} introduced
\textsc{Support Bribery} and Baumeister et
al.~\cite{bau-fal-lan-rot:c:lazy-voters} introduced \textsc{Extension
  Bribery} (both problems model the setting where voters cast
partial votes that rank some of their top candidates only, and the
briber extends these votes; they differ in that the former assumes
that the voters know their full preference orders but do not report
them completely and the latter assumes that the voters have no
preferences regarding the not-reported candidates).

There is quite a lot of research on various other algorithmic aspects of
elections. Reviewing this literature is beyond the scope of this
paper, but we point the readers to some of the recent
surveys~\cite{bet-bre-che-nie:b:fpt-voting,che-end-lan-mau:c:polsci-intro,fal-hem-hem:j:cacm-survey,fal-pro:j:manipulation,bra-con-end:b:comsoc}
and to recent textbooks~\cite{bra-con-end-lan-pro:b:comsoc-handbook,Rothe15}.

%Elkind et al.~\citep{elk-fal-sli:c:swap-bribery} have shown that
%\ShiftBribery is polynomial-time solvable for Plurality and
%$k$-Approval, but that it is $\np$-complete for the Borda,
%Copeland$^\alpha$, and Maximin voting rules (see \cref{sec:prelim} for
%exact definitions of these rules). In addition, they provided a
%$2$-approxima\-tion algorithm for \ShiftBribery under Borda.  Elkind
%and Faliszewski~\citep{elk-fal:c:shift-bribery} extended this last
%result to all scoring protocols and found approximation algorithms for
%\ShiftBribery under Copeland$^\alpha$ and Maximin. Schlotter et
%al.~\citep{elk-fal-sch:c:fallback-shift} complemented these results by
%showing that \ShiftBribery is easy for the Bucklin and Fallback rules.
%(We provide more detailed account of related literature in \cref{sec:related-work}.)
%

\myparagraph{Our Contributions}

For the Borda, the Maximin, and the Copeland rules, Elkind
et al.~\citep{elk-fal:c:shift-bribery,elk-fal-sli:c:swap-bribery} have
shown that \ShiftBribery has high worst-case complexity, but that one
can deal with it using polynomial-time approximation
algorithms. To better understand where the intractability of
\ShiftBribery really lies in different special cases, we use another
approach of dealing with computationally hard problems, namely
parameterized complexity analysis and, more specifically, the notion
of fixed-parameter tractability and the correspondingly developed exact algorithms.
For instance, almost tied elections are tempting targets for \ShiftBribery. 
An exact algorithm which is efficient for this special case may be more attractive than a general approximation algorithm.
%For instance, 
% The idea here is that an efficient election campaign is most important
% when the elections are nearly tied, in which case approximate
% solutions may not be sufficient. However,
In close-to-tied elections
it might suffice, for example, to contact only a few voters or,
perhaps, to shift the preferred candidate by only a few positions in
total.
Similarly, it is important to solve the problem exactly if one has
only a small budget at one's disposal.
% (an approximate algorithm might overspend by too large a factor).
This is captured by using various problem parameterizations and performing 
a parameterized complexity analysis.

Furthermore, it is natural to expect that the computational complexity
of \ShiftBribery depends on the nature of the voters' price functions
and, indeed, there is some evidence for this fact: For example, if we
assume that shifting~$p$ by each single position in each voter's
preference order has a fixed unit price or, at the very least, if
functions describing the prices are \emph{convex}, then one can verify
that the $2$-approximation algorithm of Elkind and
Faliszewski~\citep{elk-fal:c:shift-bribery} boils down to a greedy
procedure that picks the cheapest available single-position shifts
until it ensures the designated candidate's victory (such an
implementation would be much faster than the expensive
dynamic programming algorithm that they use, but would guarantee a
$2$-approximate solution for convex price functions only).
% speed up the $2$-approximation algorithm of Elkind
% and Faliszewski~\citep{elk-fal:c:shift-bribery} by a polynomial
% factor (in this case, one can replace an expensive dynamic programming
% algorithm with 
%
% \todo[inline]{proof or reference?!}, replacing an expensive dynamic
% programming algorithm by a simple greedy
% one. %\todo[inline]{Do we have a proof for this result?}
On the contrary, the hardness proofs of Elkind et
al.~\citep{elk-fal-sli:c:swap-bribery} all use a very specific form of
price functions which we call \emph{all-or-nothing} prices, where if
one decides to shift the preferred candidate~$p$ in some vote, then the cost of this shift
is independent of how far one shifts~$p$. In effect, one might as well shift $p$ to
the top of the vote.  See \cref{sec:shift-bribery} for the
definitions of the different price functions that we study.
%  These say
% that if we decide to shift $p$ even by one position in some preference order, then
% we might as well shift $p$ to the top of this order because it would
% cost the same.

We combine the above two sets of observations and we study the
parameterized complexity of \ShiftBribery for Borda, Maximin, and
Copeland$^\alpha$, for parameters describing the number of affected
voters, the number of unit shifts, the budget, the number of
candidates, and the number of voters, under price functions that are
either all-or-nothing, sortable, arbitrary, convex, or have a unit
price for each single shift.  The three voting rules that we select
are popular in different kinds of elections apart from political
ones. For instance, Borda is used by the X.Org Foundation to elect its
board of directors, its modified variant is used for the Formula~1
World Championship (and numerous other competitions including, e.g.,
ski-jumping, and song contests). A slightly modified version of
Copeland is used to elect the Board of Trustees for the Wikimedia
Foundation.%, etc.

We summarize our results in \cref{tab:results}, and we %will
discuss them throughout the paper. In short, it turns out that indeed
both the particular parameters used and the nature of the price
functions have strong impact on the computational complexity of
\ShiftBribery.  Three key technical contributions of our work are:
\newcommand{\localheight}{1.7ex}
\newcommand{\budget}{B}
\newcommand{\dsfollow}{$\color{darkgray}^{\spadesuit}$}
\newcommand{\bfntfollow}{$\color{darkgray}^{\Diamond}$}
\newcommand{\bfntffollow}{$\color{darkgray}^{\bigstar}$}
\newcommand{\tfont}[1]{\textsf{#1}}
\newcommand{\citefont}[1]{\small#1}
\begin{table*}[t!]
\centering
\def \pricewidth{8ex}
\def \pricewidthlong{15ex}
\newcommand{\lolocalheight}{1.6ex}
\resizebox{\textwidth}{!}{
  \begin{tabular}{@{\hspace{-.4ex}}l @{\hspace{.5ex}}l  l @{\hspace{1ex}}l @{\hspace{1ex}}l @{\hspace{1.5ex}}l @{\hspace{1ex}}l@{}}
  \toprule
  && \multicolumn{5}{c}{\normalsize \ShiftBribery}\\\cline{3-7}
  \\[-\localheight]
  \multicolumn{2}{l}{\textbf{parameter}}  & unit & convex  & arbitrary  & sortable & all-or-nothing \\
  \quad\; $\calR$ &  & prices & prices  & prices  & prices & prices\\
\midrule

\multicolumn{2}{l}{\textbf{\#shifts ($t$)}} &&&&&\\
 \quad\; \tfont{B}/\tfont{M} &[\citefont{Thm~\ref{thm:Borda-fpt-num_shifts} \& \ref{thm:Maximin-fpt-num_shifts}}]
 & $\classFPT$ & $\classFPT$ & $\classFPT$ & $\classFPT$ & $\classFPT$ \\ %  (Thm.~\ref{sb-thm:Borda-fpt-num_shifts} \& Thm.~\ref{sb-thm:Maximin-fpt-num_shifts})}\\%       

 \quad\; \tfont{C} & [\citefont{Thm~\ref{thm:Copeland-w[1]-c-num_shifts}}]
 & $\wone$-h& $\wone$-h& $\wone$-h& $\wone$-h& $\wone$-h\\ % (Thm.~\ref{sb-cor:Copeland-w[1]-c-num_shifts})} \\ %& $\wone$-com. & $\wone$-com. & $\wone$-com. & $\wone$-com.\\
 \\[-\lolocalheight]

\multicolumn{2}{l}{\textbf{\#affected voters ($n_a$)}}  \\
\quad\; \tfont{B}/\tfont{M}/\tfont{C} & [\citefont{Thm~\ref{thm:w[2]-h-num_voters_affected}}]
& $\wtwo$-h   & $\wtwo$-h  & $\wtwo$-h   & $\wtwo$-h   & $\wtwo$-h  \\ %(Thm.~\ref{sb-thm:w[2]-h-num_voters_affected})}} \\
&& & \\[-\lolocalheight]

\multicolumn{2}{l}{\textbf{budget ($\budget$)}} \\
\quad\; \tfont{B}/\tfont{M} &[\citefont{Cor~\ref{cor:Borda+Maximin-budget}}]
&  $\classFPT$  & $\classFPT$  
&  $\wtwo$-h &  $\wtwo$-h &  $\wtwo$-h \\ 

\quad\; \tfont{C} & [\citefont{Cor~\ref{cor:Copeland-budget}}]&  $\wone$-h &  $\wone$-h % (Cor.~\ref{sb-cor:Copeland-budget})}
& $\wtwo$-h & $\wtwo$-h & $\wtwo$-h \\ \\[-\lolocalheight]

\multicolumn{2}{l}{\textbf{\#voters ($n$)}} \\
\quad\; \tfont{B}/\tfont{M} & [\citefont{Pro~\ref{prop:fpt-num_voters-0/1-prices}}] 
& {\color{darkgray}$\wone$-h}\bfntffollow & {\color{darkgray}$\wone$-h}\bfntffollow& {\color{darkgray}$\wone$-h}\bfntffollow & {\color{darkgray}$\wone$-h}\bfntffollow  &  $\classFPT$\\

\quad\; \tfont{C} & [\citefont{Thm~\ref{thm:W[1]h-num_voters-copeland}  \& Pro~\ref{prop:fpt-num_voters-0/1-prices}}]  & $\wone$-h & $\wone$-h & $\wone$-h & $\wone$-h  &  $\classFPT$\\ 
\\[-\localheight]

\multicolumn{2}{l}{\textbf{\#candidates~($m$)}}    \\
\quad\; \tfont{B}/\tfont{M}/\tfont{C} & [\citefont{Thm~\ref{thm:xp-num_cands-0/1-prices}}] &  {\color{darkgray}$\classFPT$}\dsfollow & $\classXP$ & $\classXP$ % (Thm.~\ref{sb-thm:xp-num_cands-0/1-prices})
& {\color{darkgray}$\classFPT$}\bfntfollow& {\color{darkgray}$\classFPT$}\bfntfollow\\ \\[-\lolocalheight]

\midrule
&&\multicolumn{5}{c}{\normalfont  \ShiftBriberyO}\\\cline{3-7}
  \\[-\localheight]

\multicolumn{2}{l}{\textbf{\#voters ($n$)}} \\
\quad\; \tfont{B}/\tfont{M}/\tfont{C} & [\citefont{Thm~\ref{thm:W[1]h-num_voters-copeland}}] &  \multicolumn{5}{l}{$\classFPTAS$ for all considered price function families} \\ 
\\[-\localheight]

\multicolumn{2}{l}{\textbf{\#candidates~($m$)}}    \\
\quad\; \tfont{B}/\tfont{M}/\tfont{C} & [\citefont{Thm~\ref{thm:fpt-as-num-candidates}}] &  
\multicolumn{5}{l}{$\classFPTAS$ for sortable prices} \\ 

\bottomrule                                 
\end{tabular}
}
\caption[The parameterized complexity of $\calR$ \ShiftBribery and  $\calR$ \ShiftBriberyO for
  Borda (\tfont{B}), Maximin~(\tfont{M}), and Copeland{$^\alpha$}]{\label{tab:results}
  The parameterized complexity of $\calR$ \ShiftBribery and  $\calR$ \ShiftBriberyO for
  Borda (\tfont{B}), Maximin~(\tfont{M}), and Copeland{$^\alpha$}~(\tfont{C}) (for each
  rational number $\alpha$, $0\le \alpha \le 1$).  
  Note that \ShiftBriberyO is the optimization variant of \ShiftBribery, which seeks to minimize the budget spent.
  ``$\wone$-h'' (resp. ``$\wtwo$-h'') stands for $\wone$-hard (resp. $\wtwo$-hard).
  Definitions for $\classFPT$, $\wone$, $\wtwo$, $\classXP$
  are provided in \cref{subsec:parameterized-complexity}. 
  The $\wone$-hardness results marked with~``\bfntffollow'' follow from the work of \citet{BreFalNieTal2016b}. 
  The $\classFPT$ result marked with~``\dsfollow'' follows from the work of \citet{dor-sch:j:parameterized-swap-bribery}. 
  The $\classFPT$ results marked with~``\bfntfollow'' follow from the work of \citet{BreFalNieTal2016a}.
  % one of the co-authors 
  % of our technical report~\cite{BreCheFalNicNiearxiv2015}. %of the conference version~\cite{BreCheFalNicNiearxiv2015} of this work.
  }
\end{table*}
\begin{enumerate}
\item novel \emph{$\classFPT$ approximation schemes} exploiting the parameters 
``number of voters'' and ``number of candidates'' 
(such schemes are rare in the literature and of
significant practical interest),
\item a surprising W[1]-hardness result for the parameter ``number of
  voters'' when using Copeland$^\alpha$ voting (contrasting
  fixed-parameter tractability results when using Borda and Maximin),
  and
% \item a \emph{partial kernelization} (polynomial-time data reduction) result 
% for the parameter ``budget'' (maximum allowed total cost of shifting).
\item a \emph{partial kernelization} %~\cite{BGKN11}
(polynomial-time data reduction) result
for the parameter ``number of unit shifts''.
\end{enumerate}

\myparagraph{Article Outline}
The paper is organized as follows. In \cref{sec:prelim} we present
preliminary notions regarding elections and parameterized complexity
theory. In \cref{sec:shift-bribery} we formally define the
\ShiftBribery problem, together with its parameterizations and
definitions of price functions.  Our results are in \cref{sec:pbscm}
(parameterization by solution cost) and~\cref{sec:pbesm}
(parameterization by solution election size).  We discuss an outlook of our
research in \cref{sec:conclusions}.  In the appendix we provide those
proofs that were omitted from the main part for the sake of
readability.
%We discuss related work in
%\cref{sec:related-work} and we conclude in \cref{sec:conclusions}.

%% PF: commented out for the full version
% We omit most of the proofs due to space constraints; they can be found ...
%Due to space constraints, most proofs are omitted.
%deferred to an appendix which can be found in the attached supplemental material.

\section{Preliminaries}\label{sec:prelim}

Below we provide a brief overview of the notions regarding elections
and parameterized complexity theory. 
%For each integer $n$, we set $[n]
%= \{1, \ldots, n\}$.

%
% PF: We only need this in one proof? ---delete and use
%     explicit notation in that proof?
%%% For each vector $\vv{x} = (x_1, \ldots, x_n)$,
%%% each $i$, $1 \leq i \leq n$, and each $y$, by $(x_{-i},y)$ we mean
%%% vector $\vv{x}$ whose $i$'th element was replaced by $y$.\smallskip

\subsection{Elections and Voting Rules}
We use the standard, ordinal model of elections where an election $E =
(C,V)$ consists of a set $C = \{c_1, \ldots, c_m\}$ of candidates and
a set $V = \{v_1, \ldots, v_n\}$ of voters. Each voter~$v$
provides a preference order over $C$, i.\,e., a linear order
ranking the candidates from the most preferred one (position~$1$) to the least preferred one (position~$m$).
For example, if $C = \{c_1, c_2, c_3\}$, then writing $v
\colon c_1 \pref c_2 \pref c_3$ indicates that voter~$v$ likes~$c_1$
best, then~$c_2$, and then~$c_3$. 
For two candidates $c_i, c_j$ and
voter~$v$ we write $v \colon c_i \pref c_j$ to indicate that~$v$
prefers $c_i$ to~$c_j$. Further, we % sometimes
write $\pref_v$ to
denote voter $v$'s preference order. % We often use multiset notation
% to speak of a list~$V$. For example, we write $v \in V$ to
% indicate that some voter $v$ is included in the list~$V$ and
% write $\left|V\right|$ to denote the number of voters in $V$.
Given an election~$E$, for each two candidates $c$ and~$d$, we define
$N_E(c,d)$ to be the number of voters in~$E$ who prefer $c$ over~$d$.

% PF: We only needed pos to define scoring protocols and in one
%     other place; we can skip scoring protocols and rephrase that
%     other place to save space
%
% Let $E = (C,V)$ be some election. For each voter $v$ and candidate
% $c$, we define $\pos(c,v)$ to be the position of $c$ in $v$'s
% preference order (e.g., if $c$ is $v$'s most preferred candidate then
% $\pos(c,v) = 1$ and if $c$ is $v$'s least preferred candidate then
% $\pos(c,v) = \left|C\right|$). 
% For each two candidates $c$ and $d$, we define
% $N_E(c,d)$ to be the number of voters in $E$ who prefer $c$ over $d$.

A voting rule $\calR$ is a function that given an election~$E$ outputs
a non-empty set $\emptyset \neq \calR(E) \subseteq C$ of the tied
winners of the election.  Note that we use the nonunique-winner model,
where each of the tied winners is considered a winner and we disregard
tie-breaking.  This is a standard assumption in papers on bribery in
elections~\cite{fal-hem-hem:j:bribery,elk-fal-sli:c:swap-bribery,elk-fal:c:shift-bribery,SchFalElk15,dor-sch:j:parameterized-swap-bribery,bau-fal-lan-rot:c:lazy-voters}. However,
we mention that various types of tie-breaking can sometimes affect the
complexity of election-related problems quite
significantly~\cite{obr-elk:c:random-ties-matter,obr-elk-haz:c:ties-matter,fal-hem-sch:c:copeland-ties-matter,con-rog-xia:c:mle}.
Furthermore, we implicitly assume that all voting rules that we
consider are anonymous, i.e., their outcomes depend only on the
numbers of voters with particular preference orders and not on
the identities of the particular voters that cast them.
We consider
% (positional) scoring rules (with a particular focus on Borda), 
Borda, Maximin,
and the Copeland$^\alpha$ family of rules. 
%Each of 
These rules assign points to every candidate and pick as winners those
who get most; we write $\s{E}{c}$ to denote the number of points
candidate~$c\in C$ receives in election $E$---the particular voting rule
used to compute the score will always be clear from the context.
%The candidates with the highest score
%are the winners.

Let $E\coloneqq(C, V)$ be an election.
Under \emph{Borda}, each candidate~$c\in C$ receives from each voter~$v \in V$ 
as
many points as there are candidates that~$v$~ranks lower than~$c$.  
Formally, the Borda score of candidate~$c$ is $\s{E}{c} \coloneqq \sum_{d \in  C\setminus\{c\}}N_E(c,d)$. 
Similarly, the \emph{Maximin} score of a candidate~$c$ is the number of voters who prefer $c$ to his or her ``strongest competitor.'' Formally, for Maximin we have $\s{E}{c} \coloneqq \min_{d \in C\setminus\{c\}}N_E(c,d)$.
% 
% For an election with $m=\left| C\right|$ candidates, 
% a scoring rule is defined through
% scoring vector $(\alpha_1, \ldots, \alpha_m)$ of
% nonincreasing integers. Each candidate $c$ receives $\sum_{v \in
%   V}\alpha_{\pos(v,c)}$ points and the candidate(s) with most points
% win. Typical (families of) scoring rules include \emph{Plurality} (with
% scoring vectors of the form $(1,0,\ldots,0)$), \emph{$k$-Approval} (with
% scoring vectors such that $\alpha_i = 1$ for $i \leq k$ and $\alpha_i
% = 0$ for $i > k$), and \emph{Borda} (with scoring vectors of the form
% $(m-1, m-2, \ldots, 0)$).
% 
% 
% Under \emph{Maximin}, the score of candidate $c$ in election $E$ is
% defined as $\min_{d \in C\setminus\{c\}}N_E(c,d)$. Intuitively, the Maximin score
% of a candidate $c$ is the number of voters who prefer $c$ to his or
% her ``strongest competitor.''
% 

\subsection{Conventions for Describing Preference Orders}\label{subsec:conventions}
Under \emph{Copeland$^\alpha$} with $\alpha$ being a rational number,
$\alpha \in [0,1]$,
we organize a \emph{head-to-head} contest
between each two candidates $c$ and~$d$; if $c$ wins (i.e., if
$N_E(c,d) > N_E(d,c)$) then $c$ receives one point, if $d$ wins then
$d$ receives one point, and if they tie (i.e., $N_E(c,d) = N_E(d,c)$)
then they both receive $\alpha$ points.  
Formally, %More precisely, for each
%rational number~$\alpha \in [0, 1]$,
%  for an election $E = (C,V)$, 
Copeland$^\alpha$ assigns to each candidate $c \in C$ score:
\[\left|\{ d \in C\setminus\{c\} \colon N_E(c,d)\text{\textgreater}N_E(d,c)\}\right| +
\alpha \left|\{ d \in C\setminus\{c\} \colon N_E(c,d)\text{=}N_E(d,c)\}\right|.\]
Typical values of $\alpha$ are $0$, $1/2$, and
$1$, but there are cases where other values are used.  All our results
that regard \ShiftBribery for Copeland$^\alpha$ hold for each rational value of~$\alpha$.
For brevity's sake, we write ``Copeland'' instead of
``Copeland$^\alpha$ for arbitrary rational number~$\alpha$''
throughout this paper.

Given a subset~$A$ of candidates, unless specified otherwise, we write
$\seq{A}$ to denote an arbitrary but fixed preference order over $A$.
We write $\revseq{A}$ to denote the reverse of $\seq{A}$.  If $A$ is
empty, then $\seq{A}$ means an empty preference order.

Let $A$ and $B$ be two disjoint subsets of candidates which do not
contain our preferred candidate~$p$ (i.e., $p \notin A\cup B$).
Let $d,g,z\notin A\cup B$ be three distinct candidates. %that also do
%not belong to $A$ or $B$. 
We write
$\point(p, A, d, g, B, z)$ to denote the following two preference
orders:
\begin{align*}
  p \pref \seq{A} \pref d \pref g \pref \seq{B} \pref z \text{\quad and \quad}
  z \pref \revseq{B} \pref d \pref g \pref \revseq{A} \pref p.
\end{align*}
With these two preference orders, under the Borda rule we achieve the
effect that candidate~$d$ gains one point more than every other
candidate in $A\cup B \cup \{p,z\}$ and two points more than candidate
$g$.  Observe that $p$ is ranked first in the first preference order
and in order to let $d$'s score decrease (by shifting $p$ forward),
one has to shift $p$ by $|A|+1$ positions in the second vote.  This is
crucial for some of our reductions.

Consider a set $A = \{a_1, \ldots, a_{2x+1}\}$ of candidates. For
each $i$, $1 \leq i \leq 2x$, set $A_i \coloneqq \{a_{(i+1)\mod (2x+1)}, \ldots,
{a_{(i+x)\mod (2x+1)}}\}$.  For each candidate~$a_i$, we write
$\halfseqone(A,a_i)$ and $\halfseqtwo(A,a_i)$ to denote
the following two preference orders, respectively:
%%
%% PF: We should spell out these votes explicitly to avoid confusion
\begin{align*}
  a_i \pref \seq{A_i} \pref \seq{A\setminus A_i} \text{\quad and \quad}
  \revseq{A\setminus A_i}\pref a_i \pref \revseq{A_i}\text{.}
\end{align*}
%where $A_i=\{a_{(i+1)\mod (2x+1)}, \ldots, {a_{(i+x)\mod (2x+1)}}\}$
%and $\seq{A_i} = a_{(i+1)\mod (2x+1)} \pref \ldots \pref a_{(i+x)\mod
%  (2x+1)}$.  
%%
Under these two preference orders, we achieve the effect that
candidate~$a_i$ wins head-to-head contests against exactly half of the
other candidates (namely, those in $A_i$) and ties with all the
remaining ones. All the other pairs of candidates tie in their
head-to-head contests.

\subsection{Parameterized Complexity}\label{subsec:parameterized-complexity}
We point the reader to standard text books for general information on
parameterized complexity and algorithms~\citep{cyg-fom-kow-lok-mar-pil-pil-sau:b:fpt,DF13,flu-gro:b:parameterized-complexity,nie:b:invitation-fpt}.
There are also general accounts on the applications of parameterized complexity analysis to computational social choice~\cite{bet-bre-che-nie:b:fpt-voting,BCFGNW14}.

To speak of parameterized complexity of a given problem, we declare a
part of the input as a parameter (here we consider numerical
parameters only, e.g., for \ShiftBribery it could be the number of
candidates or the budget; see the next section). We say that a problem
parameterized by $k$ is \emph{fixed-parameter tractable}, that is, is
in $\classFPT$, if there exists an algorithm that given input $I$ with
parameter~$k$ gives a correct solution in time $f(k) \cdot
|I|^{O(1)}$, where $f(k)$ is an arbitrary computable function of $k$,
and $|I|$ is the length of the encoding of $I$.  To describe the
running times of our algorithms, we often use the $O^*(\cdot)$
notation. It is a variant of the standard $O(\cdot)$ notation where
polynomial factors are omitted.  For example, if the algorithm's
running time is $f(k)\cdot |I|^{O(1)}$, where $f$ is superpolynomial,
then we would say that it is $O^*(f(k))$.

Parameterized complexity theory also provides a hierarchy of hardness
classes, starting with~$\wone$, such that 
$\classFPT \subseteq \wone \subseteq \wtwo \subseteq
\ldots\subseteq \classXP$. 
%Two complete problems for classes $\wone$ and $\wtwo$ are given in the appendix.
For our purposes, it suffices to define the $\w$-classes through
their complete problems and an appropriate reducibility notion.

\begin{definition}
  Let $A$ and $B$ be two decision problems. We say that $A$ \emph{reduces} 
  to~$B$ (in a parameterized way) if there are two functions $g_1$ and $g_2$
  such that, given instance $I_A$ of~$A$ with parameter $k_A$, it holds
  that:
  \begin{enumerate}[i)]
  \item $I_A$ is a yes-instance of $A$ if and only if $g_1(I_A)$ is
  a yes-instance of $B$ with parameter~$g_2(k_A)$, and 
  \item $g_1$ is computable in $\classFPT$ time (for parameter~$k_A$).
  \end{enumerate}
\end{definition}

$\wone$ is the class of problems that reduce (in a parameterized way) to \probClique, and
$\wtwo$ is the class of problems that reduce to \probSC---in both
problems we take solution size~$k$ as the parameter. In effect, these
two problems are complete for these two classes. 
%
% \probClique is complete for $\wone$ while \probSC
% is complete for $\wtwo$.

\decprob{\probClique} {An undirected graph $G = (V(G), E(G))$ and a non-negative integer $k\ge 0$.}
{Does $G$ contain a clique of size~$k$, that is, a subset~$\clique\subseteq V(G)$ of $k$ vertices such that for each two distinct vertices~$u,w\in \clique$ one has~$\{u,w\}\in E(G)$?}

Note that even
though we use letters $V$ and~$E$ to denote, respectively, voter sets
and elections, when we speak of a graph~$G$, we write $V(G)$ to denote
its set of vertices and $E(G)$ to denote its set of edges.

\decprob{\probSC} {A family $\setFamily = (S_1, \ldots, S_m)$ of sets over a
  universe~$\setUniverse =\{u_1, \ldots, u_n\}$ of elements and a non-negative integer $k\ge 0$.}  
{Is there a size-$k$ \emph{set cover}, that is, a collection~$\setcover$ of $k$~sets in $\setFamily$ whose union is $\setUniverse$?}

For problems where an $\classFPT$ algorithm and a hardness proof are
elusive, one can at least try to show an $\classXP$ algorithm 
whose running time is polynomial if one treats the parameter as a constant.
%that runs in polynomial time, provided that one treats the parameter
%as a constant. 
As opposed to $\classFPT$, the degree of the polynomial
describing the ``$\classXP$'' running time can depend on the parameter.

% PF: We define an FPT-AS directly for Shift-Bribery, for the case
%     of the optimization variant, so I think we should not provide
%     the definition here (especially since the above part of the 
%     section talks about decision problems)
%
% Finally, an \emph{$\classFPT$-approximation scheme}~($\classFPT$-AS)~\cite{Mar08}
% with parameterization~$k$ is an algorithm whose inputs are an instance~$I$
% and an $\epsilon >0$, and it produces a $(1+\epsilon)$-approximate solution in
% $f(\epsilon,k)\cdot |I|^{O(1)}$ time for some computable function~$f$.

Below we provide a brief overview of the notions regarding elections
and parameterized complexity theory. 
%For each integer $n$, we set $[n]
%= \{1, \ldots, n\}$.

%
% PF: We only need this in one proof? ---delete and use
%     explicit notation in that proof?
%%% For each vector $\vv{x} = (x_1, \ldots, x_n)$,
%%% each $i$, $1 \leq i \leq n$, and each $y$, by $(x_{-i},y)$ we mean
%%% vector $\vv{x}$ whose $i$'th element was replaced by $y$.\smallskip

\subsection{Elections and Voting Rules}
We use the standard, ordinal model of elections where an election $E =
(C,V)$ consists of a set $C = \{c_1, \ldots, c_m\}$ of candidates and
a set $V = \{v_1, \ldots, v_n\}$ of voters. Each voter~$v$
provides a preference order over $C$, i.\,e., a linear order
ranking the candidates from the most preferred one (position~$1$) to the least preferred one (position~$m$).
For example, if $C = \{c_1, c_2, c_3\}$, then writing $v
\colon c_1 \pref c_2 \pref c_3$ indicates that voter~$v$ likes~$c_1$
best, then~$c_2$, and then~$c_3$. 
For two candidates $c_i, c_j$ and
voter~$v$ we write $v \colon c_i \pref c_j$ to indicate that $v$
prefers $c_i$ to~$c_j$. Further, we % sometimes
write $\pref_v$ to
denote voter $v$'s preference order. % We often use multiset notation
% to speak of a list~$V$. For example, we write $v \in V$ to
% indicate that some voter $v$ is included in the list~$V$ and
% write $\left|V\right|$ to denote the number of voters in $V$.
Given an election~$E$, for each two candidates $c$ and~$d$, we define
$N_E(c,d)$ to be the number of voters in~$E$ who prefer $c$ over~$d$.

% PF: We only needed pos to define scoring protocols and in one
%     other place; we can skip scoring protocols and rephrase that
%     other place to save space
%
% Let $E = (C,V)$ be some election. For each voter $v$ and candidate
% $c$, we define $\pos(c,v)$ to be the position of $c$ in $v$'s
% preference order (e.g., if $c$ is $v$'s most preferred candidate then
% $\pos(c,v) = 1$ and if $c$ is $v$'s least preferred candidate then
% $\pos(c,v) = \left|C\right|$). 
% For each two candidates $c$ and $d$, we define
% $N_E(c,d)$ to be the number of voters in $E$ who prefer $c$ over $d$.

A voting rule $\calR$ is a function that given an election~$E$ outputs
a non-empty set $\emptyset \neq \calR(E) \subseteq C$ of the tied
winners of the election.  Note that we use the nonunique-winner model,
where each of the tied winners is considered a winner and we disregard
tie-breaking.  This is a standard assumption in papers on bribery in
elections~\cite{fal-hem-hem:j:bribery,elk-fal-sli:c:swap-bribery,elk-fal:c:shift-bribery,elk-fal-sch:c:fallback-shift,dor-sch:j:parameterized-swap-bribery,bau-fal-lan-rot:c:lazy-voters}. However,
we mention that various types of tie-breaking can sometimes affect the
complexity of election-related problems quite
significantly~\cite{obr-elk:c:random-ties-matter,obr-elk-haz:c:ties-matter,fal-hem-sch:c:copeland-ties-matter,con-rog-xia:c:mle}.
Furthermore, we implicitly assume that all voting rules that we
consider are anonymous, i.e., their outcomes depend only on the
numbers of voters with particular preference orders and not on
the identities of the particular voters that cast them.
We consider
% (positional) scoring rules (with a particular focus on Borda), 
Borda, Maximin,
and the Copeland$^\alpha$ family of rules. 
%Each of 
These rules assign points to every candidate and pick as winners those
who get most; we write $\s{E}{c}$ to denote the number of points
candidate~$c\in C$ receives in election $E$---the particular voting rule
used to compute the score will always be clear from the context.
%The candidates with the highest score
%are the winners.

Let $E:=(C, V)$ be an election.
Under \emph{Borda}, each candidate~$c\in C$ receives from each voter~$v \in V$ 
as
many points as there are candidates that $v$~ranks lower than~$c$.  
Formally, the Borda score of candidate~$c$ is $\s{E}{c} := \sum_{d \in  C\setminus\{c\}}N_E(c,d)$. 
Similarly, the \emph{Maximin} score of a candidate~$c$ is the number of voters who prefer $c$ to his or her ``strongest competitor.'' Formally, for Maximin we have $\s{E}{c} := \min_{d \in C\setminus\{c\}}N_E(c,d)$.
% 
% For an election with $m=\left| C\right|$ candidates, 
% a scoring rule is defined through
% scoring vector $(\alpha_1, \ldots, \alpha_m)$ of
% nonincreasing integers. Each candidate $c$ receives $\sum_{v \in
%   V}\alpha_{\pos(v,c)}$ points and the candidate(s) with most points
% win. Typical (families of) scoring rules include \emph{Plurality} (with
% scoring vectors of the form $(1,0,\ldots,0)$), \emph{$k$-Approval} (with
% scoring vectors such that $\alpha_i = 1$ for $i \leq k$ and $\alpha_i
% = 0$ for $i > k$), and \emph{Borda} (with scoring vectors of the form
% $(m-1, m-2, \ldots, 0)$).
% 
% 
% Under \emph{Maximin}, the score of candidate $c$ in election $E$ is
% defined as $\min_{d \in C\setminus\{c\}}N_E(c,d)$. Intuitively, the Maximin score
% of a candidate $c$ is the number of voters who prefer $c$ to his or
% her ``strongest competitor.''
% 

Under \emph{Copeland$^\alpha$} with $\alpha$ being a rational number,
$\alpha \in [0,1]$,
we organize a \emph{head-to-head} contest
between each two candidates $c$ and~$d$; if $c$ wins (i.e., if
$N_E(c,d) > N_E(d,c)$) then $c$ receives one point, if $d$ wins then
$d$ receives one point, and if they tie (i.e., $N_E(c,d) = N_E(d,c)$)
then they both receive $\alpha$ points.  
Formally, %More precisely, for each
%rational number~$\alpha \in [0, 1]$,
%  for an election $E = (C,V)$, 
Copeland$^\alpha$ assigns to each candidate $c \in C$ score:
\[\left|\{ d \in C\setminus\{c\} \colon N_E(c,d)\text{\textgreater}N_E(d,c)\}\right| +
\alpha \left|\{ d \in C\setminus\{c\} \colon N_E(c,d)\text{=}N_E(d,c)\}\right|.\]
Typical values of $\alpha$ are $0$, $1/2$, and
$1$, but there are cases where other values are used.  All our results
that regard \ShiftBribery for Copeland$^\alpha$ hold for each rational value of~$\alpha$.
For brevity's sake, we write ``Copeland'' instead of
``Copeland$^\alpha$ for arbitrary rational number~$\alpha$''
throughout this paper.

\section{Shift Bribery}\label{sec:shift-bribery}

Given a voting rule $\calR$, in $\calR$ \ShiftBribery the goal is to
ensure that a certain candidate~$p$ (the \emph{preferred} candidate)
is an $\calR$-winner of the election. To achieve this effect, we can
shift $p$ forward in some of the voters' preference orders.  Each
shift may have a different price, depending on the voter and on the
length of the shift.  The problem was defined by Elkind et
al.~\citep{elk-fal-sli:c:swap-bribery}. Here we follow the notation of
Elkind and Faliszewski~\citep{elk-fal:c:shift-bribery}; see the
introduction for other related work.

\subsection{The Problem}
Let $E = (C,V)$ be some election, where $C = \{p,c_1, \ldots, c_m\}$
and $V = \{ v_1, \ldots, v_n\}$; $p$ is the preferred candidate.  A
\ShiftBribery price function~$\pi_{i}$ for voter $v_i \in V$, $\pi_i
\colon \naturals \rightarrow \naturals$, gives the price of shifting
$p$ forward in $v_i$'s preference order a given number of
positions. We require that $\pi_i(0) = 0$ and that $\pi_i(\ell) \leq
\pi_i(\ell+1)$ for each $\ell \in \naturals$. We also assume that if
$p$ is ranked on a position~$r$ in the preference order, then
$\pi_i(\ell) = \pi_i(\ell-1)$ whenever $\ell \geq r$.
% PF: removed the use of \pos
%
%, and that $\pi_i(k) =
%\pi_i(k-1)$ whenever $k \geq \pos(v_i,p)$.  
In other words, it costs
nothing to keep a voter's preference order as is, it never costs less
to shift $p$ farther, and we cannot shift $p$ beyond the top position
in the preference order.
\ifShowShiftBriberyExample %%% begin %%%
\begin{example}
Let $v_i$ be a voter with preference order $v_i \colon c_1 \pref c_2
\pref p \pref c_3$ and let $\pi_i$ be $v_i$'s price
function.
Then, by paying $\pi_i(1)$ we can change $v_i$'s preference
order to $c_1 \pref p \pref c_2 \pref c_3$, and by paying $\pi_i(2)$
we can change it to $p \pref c_1 \pref c_2 \pref c_3$.
\end{example}
\else
For instance, let $v_i$ be a voter with preference order $v_i \colon c_1 \pref
c_2 \pref p \pref c_3$ and let $\pi_i$ be $v_i$'s \ShiftBribery price function.
Then,  by paying $\pi_i(1)$ we can change $v_i$'s preference order to $c_1 \pref p \pref c_2 \pref c_3$, and
by paying $\pi_i(2)$ we can change it to $p \pref c_1 \pref c_2 \pref c_3$.
\fi %%% end %%%

It is clear that we need at most $|C|-1$ values to completely describe
each \ShiftBribery price function. We write $\Pi = (\pi_1, \ldots,
\pi_n)$ to denote the list of \ShiftBribery price functions for
the voters in $V$.

A \emph{shift action}~$\vv{s}$ is a vector $(s_1, \ldots, s_n)$ of
natural numbers, describing how far $p$ should be shifted in each of
the $n$~voters' preference orders.  We define $\shift(E,\vv{s})$ to be
the election $E' = (C,V')$ identical to $E$ except that $p$ has been
shifted forward in each voter~$v_i$'s preference order by~$s_i$
positions.  If that would mean moving $p$ beyond the top position in
some preference order, we shift $p$ up to the top position only.  We
write $\Pi(\vv{s}) = \sum_{i=1}^{n}\pi_i(s_i)$ to denote the price of
a given shift action.  A shift action $s$ is \emph{successful} if $p$
is a winner in the election $\shift(E,\vv{s})$. The term \emph{unit
  shift} refers to shifting~$p$ by one position in one preference
order.

Given the above notation, the decision variant of
$\calR$ \ShiftBribery is defined as follows.

\decprob{$\calR$ \ShiftBribery} 
{An election $E = (C,V)$ with voter set~$V = \{v_1, \ldots, v_n\}$, a list
  $\Pi = (\pi_1, \ldots, \pi_n)$ of \ShiftBribery price functions for
  $V$, a candidate $p \in C$, and an integer $B \in \naturals$.}
{Is there a shift action $\vv{s} = (s_1, \ldots, s_n)$ such that
  $\Pi(\vv{s}) \leq B$ and $p$ is an $\calR$-winner in $\shift(E,\vv{s})$?}

The optimization variant is defined analogously, but we do not include
the budget $B$ in the input and we ask for a shift action~$\vv{s}$
that ensures $p$'s victory while minimizing $\Pi(\vv{s})$.  For an
input instance~$I$ of the optimization variant of $\calR$~\ShiftBribery, we
write $\OPT(I)$ to denote the cost of the cheapest successful shift
action for $I$ (and we omit $I$ if it is clear from the context).  We
sometimes also write ``$\calR$ \ShiftBribery'' when we refer to the
optimization variant (and this is clear from the context).

% PF: For the AAMAS submission we do not need this definition
% 
% Let $\beta$ be a number, $\beta \geq 1$.  A $\beta$-approximation
% algorithm for $\calR$ \ShiftBribery is an algorithm that given an
% instance $I = (C,V,\Pi,p)$ returns a successful shift action $\vv{s}$
% such that $\Pi(\vv{s}) \leq \beta\OPT$.  
%

Typically, approximation algorithms are asked to run in polynomial time.
It is, however, not always possible to approximate the optimal solution to
an arbitrary given factor
in polynomial time.
Relaxing the polynomial running time to $\classFPT$ time, 
we can obtain the following notion (also see \citet{Mar2008} for more information).

%%%%%%%%%%%%%%% 
\begin{definition}[$\classFPT$-approximation scheme ($\classFPTAS$)]
\label{def:fpt-as}
An $\classFPT$-approximation scheme~($\classFPTAS$) %~\cite{?}
with parameter~$k$ for $\calR$ \ShiftBribery is an algorithm that,
given an instance~$I = (C,V,\Pi,p)$ and a number $\varepsilon > 0$,
returns a successful shift action $\vv{s}$ such that $\Pi(\vv{s}) \leq
(1+\varepsilon)\cdot \OPT(I)$.  This algorithm must run in $f(k,\varepsilon)\cdot |I|^{O(1)}$~time, 
where $f$ is a computable function depending on $k$ and $\varepsilon$.
\end{definition}
%%%%%%%%%%%%%%% 

When describing the running time of an approximation scheme, 
we treat $\varepsilon$ as a fixed constant.
Thus, a polynomial-time approximation scheme may, for example,
include exponential dependence on $1/\varepsilon$.

%  (when describing the running time of an
% approximation scheme, we treat $\varepsilon$ as a fixed constant so,
% for example, a polynomial-time approximation scheme may, for example,
% include exponential dependence on $1/\varepsilon$).

\subsection{Parameters for Shift Bribery}%\quad
So far, \ShiftBribery has not been studied from the parameterized
complexity theory
point of view. Dorn and
Schlotter~\citep{dor-sch:j:parameterized-swap-bribery} and Schlotter
et al.~\citep{SchFalElk15} have, however, provided
parameterized complexity results for \textsc{Swap Bribery} and for
\textsc{Support Bribery}. 
%two problems that are very close in spirit
%(see \cref{sec:related-work}).
% for several voting rules. 
%We base our parameterization of \ShiftBribery on their ideas, but we also study some other interesting parameters.

We consider two families of parameters, those referring to the properties of the successful shift action that we
seek and those describing the input election.  Regarding the first
group, we study: 
\begin{enumerate}
\item the total number~$t$ of unit shifts in the solution
(\#shifts), 
\item the total number~$n_a$ of voters whose preference orders are changed
(\#voters-affected), and 
\item the budget~$B$.
\end{enumerate} 
Regarding the second
group, we consider:
\begin{enumerate}
\item the number~$m$ of candidates
(\#candidates) and 
\item the number~$n$ of voters (\#voters).
\end{enumerate}
As we can test every possible value of a parameter in increasing order, 
we assume that the values of these parameters are passed explicitly as
part of the input.  Note that pairs of parameters such as $m$ and~$n$
clearly are incomparable, whereas $n$, as well as \#shifts, clearly
upper-bound \#voters-affected, making \#voters-affected
``stronger''~\cite{KN12} than the two other parameters. Similarly, $B$
upper-bounds both \#shifts and \#voters-affected, provided that each
shift has price at least one.

\subsection{Price Functions} 
% PF: rephrased because the discussion is in the introduction now
%We have stated in the introduction that the price functions used in
One of the conclusions from our work is that the price functions used
in \ShiftBribery instances may strongly affect the complexity of the
problem.  In this section we present the families of price functions
that we focus on.

% Up to now, researchers studying \ShiftBribery did not put emphasis on the nature of the price functions
% used in their proofs. However, we believe that depending on the price
% functions, the complexity of \ShiftBribery can change dramatically and
% our results support this view. We focus on the following families of
% price functions. 

\paragraph{All-or-nothing prices} 
A \ShiftBribery price function~$\pi$ is \emph{all-or-nothing} if there is a
value $c$ such that $\pi(0) = 0$ and for each ${\ell} > 0$,
$\pi({\ell}) = c$ (this value $c$ can be different for each voter).
Interestingly, the $\np$-hardness proofs of Elkind et
al.~\citep{elk-fal-sli:c:swap-bribery} use exactly
% for \textsc{$\mathcal{R}$ Shift Bribery} (given by Elkind et
% al.~\cite{elk-fal-sli:c:swap-bribery}), in essence, used 
this family of price functions (without referring to them directly,
though).  All-or-nothing price functions are a special case of what we
would intuitively call \emph{concave} functions.

\paragraph{Convex prices}
We are also interested in \emph{convex} price functions;
$\pi$ is convex if for each ${\ell}$, $\pi({\ell}+1) - \pi({\ell})
\leq \pi({\ell}+2) - \pi({\ell}+1)$ (provided that it is possible to
shift the preferred candidate by up to ${\ell}+2$ positions in the
given preference order).  

\paragraph{Unit prices}
To capture the setting where each unit shift has
the same cost for each voter, we define the price functions by setting $\pi({\ell}) \coloneqq
{\ell}$ for each~${\ell}$ such that $p$~can be shifted by
${\ell}$~positions. 
Unit prices are an extreme example of convex price functions.
%(the name ``unit prices'' refers to the fact that each unit shift has
%the same cost for each voter).% (and %, as it turns out,
% one of the easiest families of
% \ShiftBribery price functions).

\paragraph{Sortable prices} 
Finally, we consider sortable price functions.  A list $\Pi=(\pi_1,
\pi_2,\ldots)$ of price functions is called \emph{sortable} if for
each two voters $v_i, v_j \in V$ with the same preference order (that
is, for each two voters $v_i$ and $v_j$ such that $\mathord{\prec_i} =
\mathord{\prec_j}$) it holds that $\forall{1 \le \ell \le m-2} \colon
\pi_i(\ell)>\pi_j(\ell) \implies \pi_i(\ell+1)>\pi_j(\ell+1)$.
Informally, this means that one can sort each set~$V'$ of voters
having the same preference order so that the prices of shifting the
preferred candidate by each $\ell$ positions are nondecreasing along
the corresponding sorted order of~$V'$.  
Thus, for a given number of shifts, bribing the voters according to the sorted order is always optimal.
Many natural price function families are sortable.  
For example, a list of exponential functions
of the form $\pi_i(\ell)=a_i^\ell$ (where each voter $v_i$ may have an
individual base~$a_i$) is sortable. A list of polynomials of the form
$\pi_i(\ell)=a_i \cdot \ell^b$ (where the exponent~$b$ is the same for
the voters having the same preference order but each voter~$v_i$ may
have an individual coefficient~$a_i$) is sortable as well.

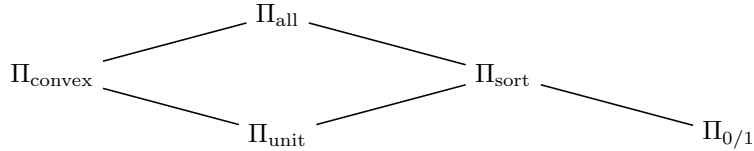
\begin{figure}[t]
\centering
\begin{tikzpicture}
  \node (all) at (0,1.6) {$\Pi_\all$};
  \node (convex) at (-3,0.8) {$\Pi_\convex$};
  \node (sort) at (3,0.8) {$\Pi_\sort$};
  \node (unit) at (0,0) {$\Pi_\unit$};
  \node (allornothing) at (6,0) {$\Pi_{0/1}$};
  \draw[semithick] (all) -- (convex) -- (unit);
  \draw[semithick] (all) -- (sort) -- (unit);
  \draw[semithick] (sort) -- (allornothing);
\end{tikzpicture}
\caption{The Hasse diagram of the inclusion relationship among the price function families (for a given election) that we study.}
\label{fig:inclusion-price-functions}
\end{figure}

\medskip Given an election $E=(C, V)$, we write $\Pi_\all$ to mean the
set of all possible price functions for this election, $\Pi_\convex$
to mean the set of the lists of convex price functions, $\Pi_\unit$ to
mean the set of the lists of unit price functions, $\Pi_{0/1}$ to mean
the set of the lists of all-or-nothing price functions, and
$\Pi_\sort$ to mean the set of all sortable lists of price functions.
We observe the following straightforward relations between these sets
(also see the Hasse diagram in \cref{fig:inclusion-price-functions}).

\begin{proposition}\label{prop:inclusion-price-functions}
  For each given election, the following relations hold between the
  families of price functions:
\begin{enumerate}
 \item  $\Pi_\unit \subset \Pi_\convex \subset \Pi_\all$,
 \item  $\Pi_{0/1} \subset \Pi_\sort \subset \Pi_\all$, 
 \item  $\Pi_\unit \subset \Pi_\sort$.
\end{enumerate}
\end{proposition}

\section{Parameterizations by Solution Cost Measures}\label{sec:pbscm}

In this section we present our results for parameters measuring the
solution cost, i.e., for the number of unit shifts, 
for the number of voters affected by at least one shift,
and for the budget. It turns out that parameterization
by the number of unit shifts tends to lower complexity ($\classFPT$
algorithms for Borda and Maximin and $\wone$-hardness for Copeland)
than parameterization by the number of affected voters ($\wtwo$-hardness). The case of parameterization by the budget lies in between, and
the complexity depends on each particular price function family.

\subsection{Unit Shifts}\label{subsec:unit shifts}
\probShiftBBorda and \probShiftBMaximin parameterized by the number~$t$ of unit
shifts in the solution are in $\classFPT$ for arbitrary price functions.
% For Borda and Maximin, \ShiftBribery parameterized by the number of unit
% shifts in the solution is in $\classFPT$ for each price-function family that we
% study.
The reason is that in these cases it is easy to isolate a small number
of candidates on which one needs to focus. More precisely, we can shrink the number of candidates as well as the number of voters to be bounded by 
functions in $t$ (in effect, achieving a partial kernelization~\cite{BGKN11}).

\newcommand{\thmbordafptnumshifts}{\probShiftBBorda parameterized by the number~$t$ of unit shifts is in $\classFPT$ for arbitrary price functions. The running time of the algorithm is $O^*((2^{t}\cdot (t+1)\cdot t)^{t})$.}

\newcommand{\criticalcandidatesetdef}[1]{
  \ensuremath{C(#1) \coloneqq \{c \in C\setminus\{p\} \mid \s{E}{c} >
  \s{E}{p}+#1\}}
}

\newcommand{\votersetbound}[1]{
  \ensuremath{2^{#1}\cdot (#1+1)\cdot {#1}}
}

\begin{theorem}\label{thm:Borda-fpt-num_shifts}
   \thmbordafptnumshifts
 \end{theorem}

\begin{proof}
  Let $I = (E=(C,V),\Pi,p,B)$ be an input instance of \probShiftBBorda for
  which we seek a shift action that uses at most $t$ unit shifts. 
  Our algorithm iterates over all possible number~$t'$ with $0 \leq
  t' \leq t$, which we interpret as the exact number of shifts used in
  the solution.
% 	and then, in each iteration, executes the following algorithm.

  Under the Borda rule, applying a shift action that uses $t'$~unit
  shifts increases the score of $p$ to exactly $\s{E}{p}+t'$. This
  means that irrespective of which shift action with $t'$~unit shifts
  we use, if some candidate $c$ has score at most $\s{E}{p}+t'$ in
  election $E$, then after applying the shift action, $p$~will
  certainly have score at least as high as $c$. On the contrary, if
  in election~$E$ some candidate $c$ has score greater than
  $\s{E}{p}+t'$, then a successful shift action must ensure that
  $c$ loses at least $\s{E}{c}-(\s{E}{p}+t')$~points.  Since
  every unit shift decreases the score of exactly one candidate, it follows that if
  there is a successful shift action that uses exactly $t'$ shifts,
  then the set $\criticalcandidatesetdef{t'}$ has at most $t'$ elements. Our algorithm first
  computes the set~$C(t')$. If $|C(t')| > t'$, then the algorithm skips the current iteration and continues with $t'\leftarrow t'+1$.

  By the above argument, we can focus on a small group of candidates, namely on~$C(t')$.
  Obviously, this set has at most $t'$ elements and can be found in linear time.
  Now if we only need to brute-force search all possible shift actions among 
  a small group of voters, whose size is upper-bounded by a function in $t'$,
  then we obtain our desired $\classFPT$ result.
  Indeed, the remainder of this proof is devoted to show that 
  only a small subset of voters (whose size is upper-bounded by a function in $t'$) are relevant to search for a successful shift action
  and this subset can be found in polynomial time.
  
  We first define a subset of ``relevant'' voters and then show why it is enough to focus on them.
  For each candidate subset $C'$ of $C(t')$ and for each number $j$, 
  $0 \leq j \leq t'$, 
  we compute a subset~$V^j_{C'}$ of $t'$~voters~$v$ such that:
  \begin{enumerate}[i)]
  \item if we shift $p$ by $j$ positions in $v$'s preference order,
    then $p$ passes each candidate from $C'$ and $j-|C'|$ other candidates from $C\setminus C(t')$, and
  \item $\sum_{v \in V^j_{C'}}\pi_v(j)$ is minimal. 
  \end{enumerate}
  If there are several subsets of voters that satisfy these conditions, then we
  pick one arbitrarily.
  We set $V(t') \coloneqq \bigcup_{\genfrac{}{}{0pt}{}{C' \subseteq C(t')}{ 0 \leq j \leq t'}}V^j_{C'}$.
 
  If there is a shift action $\vv{s}$ that uses exactly $t'$ unit
  shifts and has price at most $B$, then there is a shift action
  $\vv{s'}$ that uses exactly $t'$ shifts, has price at most $B$, and
  affects only the voters in $V(t')$. To see this, assume that there is a
  voter $v_i \not\in V(t')$ and a number $j$ such that $\vv{s}$ shifts
  $p$ by $j$ positions in $v_i$. Let $C'$ be the set of candidates
  from $C(t')$ that $p$ passes if shifted $j$ positions in $v_i$. By
  definition of $V_{C'}^{j}$ and by simple counting arguments, there is a
  voter $v_k \in V^j_{C'}$ for which $\vv{s}$ does not shift $p$.
  Again, by definition of $V_{C'}^{j}$, $\pi_{v_k}(j) \leq \pi_{v_i}(j)$. Thus, if in
  $\vv{s}$ we replace the shift by $j$~positions in~$v_i$ by a shift
  of $j$~positions in~$v_k$, then we still get a successful shift action
  using exactly $t'$ shifts and with price no greater than that of $\vv{s}$. We can repeat this
  process until we obtain a successful shift action of cost at most~$B$ 
  that affects the voters in~$V(t')$ only.
  
  In effect, to find a successful shift action that uses at most $t'$
  unit shifts, it suffices to focus on
  the %candidates from $C'$ and on the
  voters from $V(t')$. The cardinality of $V(t')$ can be upper-bounded by
  $|V(t')| \leq \sum_{\genfrac{}{}{0pt}{}{C'\subseteq C(t')}{0 \le j \le t'}}t' \leq
  \votersetbound{t'}$. 
% 2^{t'}(t'+1)t'$.
  Since in the preference order of each voter in~$V(t')$ 
  we can shift $p$ by at most $t'$ positions 
  and we can do so for at most $t'$ voters, there are at most $|V(t')|^{t'}$ shift actions that we need to try. 
  Moreover, $V(t')$ can be found in polynomial time.
  We can consider them all in $\classFPT$ time $O^*((2^{t'}\cdot (t'+1)\cdot t')^{t'})$.
\end{proof}

Using the ideas of the previous proof, we can transform a given election into one that
% It turns out that it is possible to transform the underlying election to one that
includes only $f(t)$ candidates and $g(t)$ voters, where $f$ and $g$
are two computable functions.  In this way, we obtain a so-called
partial problem kernel~\cite{BGKN11,Kom15}.  It is a \emph{partial}
problem kernel only since the prices defined by the price
function~$\Pi$ and the budget are not necessarily upper-bounded by a
function of~$t$.

\newcommand{\kernelnumbercandidates}{
  \ensuremath{t^4\cdot 2^t}}

\newcommand{\kernelnumbervoters}{
  \ensuremath{t^3\cdot 2^t}}

\newcommand{\criticalcandidateset}{\ensuremath{C_{\mathrm{crit}}}}
\newcommand{\criticalvoterset}{\ensuremath{V_{\mathrm{crit}}}}

\newcommand{\lembordakernelnumshifts}{
  An instance of \probShiftBBorda parameterized by the number~$t$ of unit shifts can be reduced to an equivalent instance with the same budget, 
  and with $O(\kernelnumbercandidates)$~candidates and $O(\kernelnumbervoters)$~voters.}

\begin{theorem} \label{thm:Borda-partialkernel-num-shifts}
  \lembordakernelnumshifts
\end{theorem}

\begin{proof}
  By the proof of \cref{thm:Borda-fpt-num_shifts}, 
  we know that we only need to focus on a subset~$\criticalcandidateset$ of candidates whose size is upper-bounded by a function of the parameter~$t$.
  This set~$\criticalcandidateset$ is defined as follows:
  \[
  \criticalcandidateset\coloneqq\bigcup_{t'=0}^{t} {C(t')},
  \]
  where $\criticalcandidatesetdef{t'}$ if the number of candidates whose scores are greater than $\s{E}{p}+t'$ is at most $t'$; otherwise it will not be possible to decrease the score of every of those candidates by at least one within $t'$ unit shifts, thus, $C(t')\coloneqq\emptyset$. 

  If~$n \le t^3 \cdot 2^t$, then we simply set~$\criticalvoterset \coloneqq V$ to be the set of all voters.
  Otherwise, from the previous proof it follows that we only need to focus on a small group of voters from the following set
  \[
  \criticalvoterset\coloneqq\bigcup_{t'=0}^{t} {V(t')},
  \]
  where $V(t')$ is defined as in the proof of \cref{thm:Borda-fpt-num_shifts}.
  Briefly put, $V(t')$ consists of the voters such that if we shift $p$ in some of their preference orders by a total of $t'$~positions,
  then we may make $p$ a winner.
  Since we only need to compute~$\bigcup_{t'=0}^{t} {V(t')}$ when~$n \ge t^3 \cdot 2^t$, it follows that~$\criticalvoterset$ can be computed in polynomial time using the straightforward algorithm in the previous proof.

  Clearly, $|\criticalcandidateset|\le t\cdot (t-1)/2$ and $|\criticalvoterset| \le \min\{n, \sum_{t'=0}^{t} \votersetbound{t'}\} \le \min \{ n, t^3\cdot 2^t \}$.

  The remaining task now is to construct an election \emph{containing} all the
  candidates from $\criticalcandidateset$ and the voters from
  $\criticalvoterset$ along with their price functions such that 
  the election size is still upper-bounded by a function in $t$.
  
  % To this end, for each candidate~$c_i\in \criticalcandidateset$,
  % let $s_i$ be the score difference between $c_i$ and $p$ in the original
  % election.
  % Note that $s_i \le t$.
  We construct the candidate set~$C_{\textrm{new}}$ of the new election as follows:
  We introduce a dummy candidate~$d_i$ for every candidate~$c_i\in \criticalcandidateset$.
  This dummy candidate will be used to realize the original score difference between $p$ and every candidate from $\criticalcandidateset$.
  We denote the set of all these dummy candidates as $D$.

  For each voter~$v_i\in \criticalvoterset$,
  we need to ``replace'' the candidates in his or her original preference order
  that do not belong to $\criticalcandidateset$ 
  but that are still ``relevant'' for the shift action.
  Let $C_i$ be the set of candidates from $\criticalcandidateset$ that
  are ranked ahead of $p$ by voter~$v_i$,
  that is, $C_i\coloneqq\{c'\in \criticalcandidateset \mid c' \pref_i p\}$.
  We introduce a new set~$F_j$ of dummy candidates whose size equals the number
  of irrelevant candidates that rank ahead of $p$ by no more than $t$ positions.
  That is, 
  \[F_i \coloneqq \{c \in C\setminus \criticalcandidateset \mid c \text{ is
    ranked at most } t \text{ positions ahead of } p \text{
    by } v_i\}.\]
  % By the definition of the voters in $\criticalvoterset$,
  % we know that there are at most $t$~candidates in $F_j$. 
  Now, the new candidate set~$C_{\textrm{new}}$ is set to 
  \[
  C_{\textrm{new}} \coloneqq \criticalcandidateset\cup D \cup \bigcup_{v_i \in \criticalvoterset}{F_i}.
  \]
  Clearly, $|C_{\textrm{new}}| \le 2|\criticalcandidateset| + t \cdot |\criticalvoterset| \le t\cdot (t-1) + \min\{ t^4 \cdot 2^t, t \cdot n \}.$
  
  The new voter set~$V$ consists of two blocks:
  \begin{enumerate}
  \item The first block is constructed to maintain the score difference between $p$ and any
    candidate~$c_j \in \criticalcandidateset$ in the original election.
    To this end, let $s_j$ be this score difference.
    We introduce $s_j$~pairs of voters with the preference orders~$\point(p,
    \emptyset, c_j, d_j, C_{\textrm{new}}\setminus \{p, c_j, d_j\}, \cdot)$ (see \cref{subsec:conventions} for the definitions).
    Note that the two preference orders that are specified by $\point()$ will cause
    $c_j$ to have two points more than~$d_j$ 
    and one point more than any other candidate, including $p$.
    We set their price for the first unit shift to $B+1$.
  \item The second part is constructed to ``maintain'' the important part of the
    voters from $\criticalvoterset$.
    That is, 
    for each voter~$v_i\in \criticalvoterset$,
    we set his or her new preference order to be 
    \[
    v_i: \seq{C_i \cup F_i} \pref p \pref \seq{C_{\textrm{new}} \setminus (C_i \cup F_i \cup \{p\})},
    \]
	 where~$\seq{C_i \cup F_i}$ corresponds to the order of the candidates in~$v_i$'s original preference order and the price function remains unchanged. 
    We also add a dummy voter~$v'_i$ with the reverse preference order of $v_i$ and we set his or her price for shifting $p$ by the first position to $B+1$.     
  \end{enumerate}
  The size of the newly constructed voter set~$V_{\textrm{new}}$ is $(\sum_{c_j \in \criticalcandidateset}{2s_j}) + 2|\criticalvoterset| \le t\cdot (t-1)+t^3 \cdot 2^{t+1}$.

  The equivalence of the two instances can be verified as follows:
  The newly added dummy candidates, either from $D$ or from $F_i$ have at most the same score as $p$.
  Each candidate~$c_j\in \criticalcandidateset$ has exactly $s_j$~points more than~$p$.
  Since it will be too expensive to bribe the voters from the first block 
  or the dummy voters~$v'_i$~($B+1$ budget for one unit shift),
  the only possibility of shifting $p$ within the budget~$B$ is to bribe the voters who come from the original election and who have the ``same'' preference orders up to renaming of the candidates
  that do not belong to $\criticalcandidateset$.
  %Thus, the equivalence of the two instances follows.
\end{proof}

The fixed-parameter tractability result for \probShiftBMaximin follows by using a similar approach as in \cref{thm:Borda-fpt-num_shifts}.

\newcommand{\thmmaximinfptnumshifts}{\probShiftBMaximin parameterized by the number~$t$ of unit shifts is in $\classFPT$ for arbitrary price functions. The running time of the algorithm is $O^*((2^{t}\cdot (t+1)\cdot t)^{t})$.}

\begin{theorem}\label{thm:Maximin-fpt-num_shifts}
  \thmmaximinfptnumshifts
\end{theorem}

\begin{proof}
   Let $I = (E=(C,V),\Pi, p, B)$ be an input instance of
  \probShiftBMaximin.  Suppose that,
  altogether, we are allowed to shift $p$ by at most $t$ positions.
  Under Maximin, such a shift can make $p$ gain at most~$t$~points.
  Just similar to the case of Borda in \cref{thm:Borda-fpt-num_shifts}, 
  our algorithm iterates over all possible numbers~$t'$ with $0 \leq
  t' \leq t$, which we interpret as the exact number of points 
  that $p$ gains after a successful shift action with $t$ unit shifts and then,
  in each iteration, executes the following algorithm.
  % The first step of our algorithm is to guess a number $i$, $0 \leq i
  % \leq t$, of points that candidate~$p$ gains after a successful shift
  % action (assuming that there, indeed, is a successful shift action
  % that uses at most $t$ shifts and has price at most $B$).

  For each candidate~$c\in C\setminus \{p\}$, if $N_{E}(p, c) <
  \s{E}{p}+t'$, then we have to shift $p$ to pass $c$ in some
  preference orders.
% to increase the number of voters preferring $p$
%  over $c$.
%   Indeed, such candidates can not have more
%   than $\s{p}+j$ points since otherwise it is not possible to make $p$ winning
%   with $\s{p}+j$ points after $s$ unit shifts.
%   Maximin scores
%   before the bribery are already more than $\s{p}$ are
%   \emph{interesting}.
  Since in every unit shift~$p$ passes exactly one candidate, a
  counting argument shows that if there are more than $t-t'+1$
  candidates~$c$ with $N_{E}(p, c) < \s{E}{p}+t'$, then there is no
  successful shift action that uses at most $t'$ unit shifts. 
  Thus the set~$C({t'})=\{c\in C\setminus \{p\} \mid N_{E}(p,c) < \s{E}{p}+t'\}$ has at most $t-t'+1$ candidates. 
  Hence, if $|C(t')|>t-t'+1$, then the algorithm skips the current iteration and continues with $t'\leftarrow t'+1$.

  By the above argument, it suffices to focus on a small set of
  candidates: the set $C(t')$. Using the same reasoning as in \cref{thm:Borda-fpt-num_shifts} for Borda voting,
  it also suffices to focus on a subset~$V(t')$ of at most $2^t\cdot (t+1)\cdot t$ voters (we can compute this set in a way
  analogous to that for Borda voting).  Since we can
  use at most $t$ unit shifts and, thus, we can affect at most $t$
  voters, there are at most $O^*(|V(t')|^t) = O^*((2^t\cdot (t+1)\cdot t)^t)$ shift
  actions that we have to try. We can do so in $\classFPT$ time.
\end{proof}

\newcommand{\obsmaximinkernelnumshifts}
{
  An instance of \probShiftBMaximin parameterized by the number~$t$ of unit shifts can be reduced to an equivalent instance with the same budget, 
  and with $O(\kernelnumbercandidates)$~candidates and $O(\kernelnumbervoters)$~voters.
}

Using an analogous approach to construct an equivalent election as for \cref{thm:Borda-partialkernel-num-shifts},
we obtain the following.
\begin{corollary}\label{obs:Maximin-partialkernel-num-shifts}
  \obsmaximinkernelnumshifts
\end{corollary}

For Copeland, we do not get $\classFPT$ membership, but we show
$\wone$-hardness even for unit prices and for all-or-nothing prices,
which implies hardness for each of our price function families.

\newcommand{\thmcopelandwonenumbershifts}{%  For each rational number~$\alpha$, $0 \leq \alpha \leq 1$,
  \probShiftBCopeland parameterized by the number of unit
  shifts is $\wone$-hard for each price function family that we consider.
  % for all-or-nothing prices and it
  %is $\wone$-complete for unary-encoded all-or-nothing prices.
}

\begin{theorem}\label{thm:Copeland-w[1]-c-num_shifts}
  \thmcopelandwonenumbershifts
\end{theorem}

\begin{proof}  
  We first show the result for all-or-nothing prices which, by
  \cref{prop:inclusion-price-functions}, also covers the case of
  convex prices, sortable prices, and  arbitrary price functions.
  In the second part, we deal with unit prices.

\paragraph{All-or-nothing prices}
 For the sake of completeness, fix some rational number~$\alpha$, $0 \leq \alpha \leq 1$, for the Copeland ``tie breaking'';
  we will see later that we will not use $\alpha$ since the number of voters in the constructed instance is odd.

%  We now give the $\wone$-hardness proof. Specifically,
  We give a reduction from \probClique parameterized by the clique size. 
  Let $G = (V(G),E(G))$ be our input graph and let $k$ be the clique size that we seek (we assume $k>1$).  Let $n =
  \left|V(G)\right|$ and $m = \left|E(G)\right|$.  Without loss of
  generality, assume that $n$ is odd and that $n>6$
  (both assumptions can be realized by adding several isolated
  vertices to the graph).
%  Finally, assume that the graph has at least $k(k-1)$ edges (one can
%  add $k(k-1)$ pairs of vertices, each connected by an edge).

  We construct an election $E = (C,V)$ as follows. 
  First, we set the available budget $B$ to be $k\cdot (k-1)/2$.
	We let $C = V(G)
  \cup \{p,d\} \cup D$, where $p$ is our preferred candidate, $d$ is the unique winner in $E$, and $D$ is a set of $2n$ dummy candidates.
  Note that we use $V(G)$ to denote both the vertex set and a subset of the candidates.
  However, given an element in $V(G)$ we will always make it clear whether we mean a vertex or a candidate.

  We form $V$ so that $\left|V\right|$ is odd and the following holds:
  \begin{enumerate}[i)]
  \item For each edge~$e = \{x,y\}$ in $E(G)$, 
    there are two voters $v_e$ 
    with preference order 
    \begin{align*}
       x \pref y \pref p \pref \seq{C\setminus \{x,y,p\}}  
% 		\text{\quad and \quad} \revseq{C\setminus \{x,y,p\}} \pref p \pref y \pref x. 
	 \end{align*}
	 and $v_{e}'$ with the reversed ordering of~$v_e$.
    They have all-or-nothing price functions $\pi_e, \pi_{e}'$ such that $\pi_e(0) = 0$ and $\pi_e(j) = 1$ for each $j > 0$ and $\pi_{e}'(j)= B+1$ for each $j \ge 0$. 
%     We refer to the subset of voters~$v_{e}, e\in E(G)$ as $V_\edges$ and the subset of voters~$v_{e}', e\in E(G)$ as $V'_{\edges}$.
    We set~$V_\edges \coloneqq \{ v_e \mid e \in E(G) \}$. % and~$V'_\edges \coloneqq \bigcup_{e \in E(G)} v'_e$

 \item There is a set~$V_\structure$ of polynomially many voters
   who implement the following results of head-to-head contests (this
   indeed can be done by applying the classic McGarvey's
   theorem~\cite{mcg:j:election-graph}): 
   Each candidate in $V(G)$ wins head-to-head contests against $p$ (by $2k-3$ voters' preference orders) and against $d$ (by one
   voter's preference order).  
   Candidate~$d$ wins head-to-head contests against all candidates in $D$ (by one voters' preference order).
   Candidate~$p$ wins head-to-head contests against $2n-k-1$ candidates in $D$
   and loses the remaining ones (by one preference order); 
   $p$ also wins the head-to-head contest against $d$ (by one preference order).
   Each candidate in $V(G) \cup D$ wins exactly half of the head-to-head contests against the other members of $V(G) \cup D$ (i.e., each member of $V(G) \cup D$ wins head-to-head contests against exactly  $\lfloor 3n/2\rfloor$ other members
   of $V(G) \cup D$; recall that $n$ is odd).
   
   A formal definition of the preference orders of voters in~$V_{\structure}$ reads as follows. 
   Let $D_1\subseteq D$ be a subset of $k+1$ (arbitrary) candidates.
  \begin{enumerate}[a)]
  \item There is one voter with preference order $\seq{V(G) \cup D} \pref p \pref d$.
  \item There are $(k-2)$ pairs of voters with preference orders
    \begin{align*}
      \seq{V(G)} \pref p \pref d \pref \seq{D}\text{\quad and \quad}
      \revseq{D} \pref d \pref \revseq{V(G)} \pref p.
    \end{align*}
  \item There are two voters with preference orders
    \begin{align*}
      d \pref \seq{D} \pref p \pref \seq{V(G)} \text{\quad and \quad}
      \revseq{V(G)} \pref p \pref d \pref \revseq{D}.
    \end{align*}
  \item There are two voters with preference orders
    \begin{align*}
      & p \pref \seq{D\setminus D_1} \pref \seq{D_1} \pref \seq{V(G)} \pref d \text{\quad and \quad} \\
      & d \pref \revseq{V(G)} \pref \revseq{D_1} \pref p \pref \revseq{D\setminus D_1}.
    \end{align*}
    \item %  For technical reasons, we rename the candidates in $V(G)$ to be
   % $a_1, \ldots, a_n$ and the candidates in $D$ to be $a_{n+1},
   % \ldots, a_{2n}$.
   % Now, 
   For each candidate~$x \in V(G)\cup D$,  
   % when we use the preference order construction from
   % $\halfseqone$ and $\halfseqtwo$ (see \cref{subsec:conventions} for
   % the definitions).   
   % There is one voter with preference
   %  order~$d \pref \revseq{D} \pref \revseq{V(G)}$. Note that 
   %    For each candidate~$a' \in V(G)\cup D \setminus \{a\}$,
   there are two voters with preference orders
   \begin{align*}
     & \halfseqone(V(G)\cup D, x) \pref p \pref d \text{\quad and \quad}\\
     & d \pref p \pref \halfseqtwo(V(G)\cup D, x).
   \end{align*}
   (See \cref{subsec:conventions} for the construction of
   preference orders $\halfseqone$ and $\halfseqtwo$.)
  \end{enumerate}
     
  Finally, each voter in $V_\structure$ has all-or-nothing price function $\pi$ such that $\pi(0) = 0$ and
   $\pi(j) = B+1$ for each $j > 0$.
  \end{enumerate}

  Note that due to the budget, we can only afford to bribe the voters in~$V_\edges$.

  Prior to any bribery, $p$ has $2n-k$ points, $d$ has $2n$ points,
  and each other candidate has at most $\lfloor3n/2\rfloor+2$
  points. This % By our assumptions, this 
  means that $d$ is the unique winner
  of $E$. We claim that there is a shift action $\vv{s}$ such that
  $\Pi(\vv{s}) \leq B$ and $p$ is a winner of $\shift(E,\vv{s})$ if and
  only if $G$ has a clique of order $k$.

  Assume that there is a successful shift action that ensures $p$'s victory. 
  Given our price functions, we can bribe up to
  $k(k-1)/2$ voters in $V_\edges$, in each case to shift $p$
  by two positions to the top. Further, it is clear that a successful
  shift action must ensure that $p$ obtains $k$ additional points (it
  is impossible to decrease the score of $d$). This means that there
  is some set $Q \subseteq V(G)$ with at least $k$ candidates such
  that $p$ passes each candidate from~$Q$ at least $k-1$ times (if $x
  \in Q$ and $p$ passes $x$ in $k-1$ preference orders, then the value $N_E(x,p) -
  N_E(p,x)$ changes from $2k-3$ to $-1$ and $p$ wins the head-to-head
  contest against $x$).
  Given 
%the voters in $V_\edges$, with budget $B = k(k-1)/2$ we
  our budget, we can altogether shift~$p$ by $t = k(k-1)$ positions (each
  two positions correspond to an edge in~$G$; $t$ is the value of our parameter). 
  Thus $Q$ contains
  exactly $k$ candidates, in each unit shift $p$ passes one of them,
  and $p$ passes each candidate from $Q$ exactly $k-1$ times. This is
  possible if and only if candidates in $Q$ correspond to a clique in
  $G$.

  On the contrary, if $Q$ is a set of $k$ vertices from $V(G)$ that
  form a clique, then shifting $p$ forward in $B$ preference orders from
  $V_\edges$ that correspond to the edges of the graph induced by $Q$ ensures $p$'s victory.

\paragraph{Unit prices}
  The above argument can be adapted to the case of unit prices. 
  The key trick is to insert sufficiently many ``filler'' candidates directly in front of $p$ in the preference orders of some voters 
  who are not intended to be affected. This simulates the effect of all-or-nothing prices.

  We again reduce from \probClique parameterized by the order of the clique. 
  The budget~$B$ is set to $3 k\cdot (k-1) / 2$.
  We use the same notation for the input instance $(G,k)$ of \probClique.
  Without loss of generality, we assume that $(n+m)$ is odd.
  We set the candidate set~$C$ to be $V(G)\cup \{p,d\}\cup D_1 \cup D_2 \cup F \cup R$ 
  where:
	\begin{enumerate}
		\item $p$ is our preferred candidate,
		\item $d$ will be the unique winner in the constructed election~$E$,
		\item $D_1$ is a set of $4(m+n)-(k\cdot (k-1)/2+k+1)$ dummy candidates, 
		\item $D_2$ is a set of $k\cdot (k-1)/2+k+1$ dummy candidates,
		\item $F\coloneqq F_1 \cup F_2$ consists of two sets of $B$ ``filler candidates'' each, and
		\item $R\coloneqq\{q_e \mid e \in E(G)\}$ is a set of $|E(G)|$ ``edge candidates''.
	\end{enumerate}

  The filler candidates are used to simulate the affect of prices higher than the budget~$B$
  by forming a ``wall'' of $B$ candidates in front of~$p$ against which~$p$ wins anyway.
  The edge candidates together with a budget blow-up of $k\cdot (k-1)/2$ in the score difference
  between~$d$ and~$p$ are used to enforce that one has to shift~$p$ to the top position
  whenever one bribes a voter.
  
  We form the voter set~$V$ as follows:
  \begin{enumerate}[i)]
  \item For each edge~$e=\{x,y\}$ in $E(G)$, add two voters~$v_e, v'_e$ with preference orders
    \begin{align*}
		v_{e}\colon & q_e \pref x \pref y \pref p \pref \seq{F} \pref \seq{C\setminus (\{x,y,p,q_e\}\cup F)} \text{\quad and }\\
		v'_{e}\colon & \revseq{C\setminus (\{x,y,p,q_e\}\cup F)} \pref \revseq{F} \pref p \pref y \pref x \pref q_e\text{,}
    \end{align*}
	 that is, $v_e'$ has the reversed preference order of~$v_e$.
    Let $V_{\edges}$ be the set of all voters~$v_e$.
  \item Add one voter with preference order
    \[
     \seq{R} \pref \seq{V(G)} \pref d \pref \seq{D_2} \pref \seq{F} \pref p \pref \seq{D_1}\text{.}
    \]
  \item Add two voters with preference orders (the boldfaced parts highlight the parts in the two voters that do \emph{not} occur reversed.)
    \begin{align*}
      \boldsymbol{p \pref d} \pref \seq{F}\pref \seq{C \setminus (\{p,d\} \cup F)} &\text{\quad and }\\
      \revseq{C \setminus (\{p,d\} \cup F)} \pref \revseq{F} \pref \boldsymbol{p \pref d} \text{.}&
    \end{align*}
  \item Add $k-2$ pairs of voters with preference orders
    \begin{align*}
    \boldsymbol{\seq{{V(G)}}\pref\seq{F}\pref p} \pref \seq{C \setminus (\{p\} \cup F \cup V(G))} &\text{\quad and }\\
    \revseq{C \setminus (\{p\} \cup F \cup V(G))} \pref \boldsymbol{\revseq{V(G)} \pref \revseq{F} \pref p}\text{.}&
  \end{align*}

  \item Add $k-1$ four tuples of voters with preference orders
    \begin{align*}
      \boldsymbol{p \pref \seq{D_1} \pref \seq{F_1}} \pref \seq{F_2} \pref \seq{C \setminus (\{p\} \cup F \cup D_1)} &\text{,}\\
      \revseq{C \setminus (\{p\} \cup F \cup D_1)} \pref \revseq{F_2} \pref \boldsymbol{p \pref \revseq{D_1} \pref \revseq{F_1}} &\text{,}\\
      \boldsymbol{p \pref \seq{D_1} \pref \seq{F_2}} \pref \seq{F_1} \pref \seq{C \setminus (\{p\} \cup F \cup D_1)} &\text{,\quad and }\\
      \revseq{C \setminus (\{p\} \cup F \cup D_1)} \pref \revseq{F_1} \pref \boldsymbol{p \pref \revseq{D_1} \pref \revseq{F_2}} &\text{.}
    \end{align*}
  \item Set $A=V(G)\cup D_1\cup D_2 \cup R$. 
    For each candidate~$x\in A$, 
    add one pair of voters with preference orders (see \cref{subsec:conventions} for the definitions of $\halfseqone$ and $\halfseqtwo$)
    \begin{align*}
      & \halfseqone(\seq{A},x) \pref \seq{F} \pref p \pref d\text{ \quad and }\\
      & d \pref p \pref \revseq{F} \pref \halfseqtwo(\seq{A},x)\text{.}
    \end{align*}
  \end{enumerate}
\newcommand{\mr}[2]{\multirow{#1}{*}{#2}}

% \begin{table}
%   \begin{tabular} {|c|lclcl|}
%     \hline
%   candidate & \multicolumn{5}{c|}{head-to-head contests won} \\
%   \hline
%   \hline
%   \mr{2}{$d$}   & $4(m+n)$ & +& $2B$ &&\\
%                 & (against $D_1 \cup D_2$)&& (against $F$) &&  \\
%   \hline
%   \mr{3}{$p$}   & $4(m+n)-$ & + & $2B$  & +  & $1$ \\
%                 & $(k(k-1)/2)+k+1)$ & & (against $F$) && (against $d$) \\
%                 &  (against $D_1$)  & & & &  \\
%   \hline
%   \mr{3}{$\in D_1$} &  $2(m+n)+\lfloor(n+m)/2\rfloor$  & + & $2B$  &&                 \\
%                 & (against half of && (against $F$) && \\
%                 & $V(G)\cup D_1 \cup D_2 \cup R$) &&  && \\
%   \hline
%   \mr{3}{$\in D_2$} & $2(m+n)+\lfloor(n+m)/2\rfloor$ &+& $2B$ & + & $1$   \\
%                 & (against half of &&(against $F$) && (against $p$)\\
%                 & $V(G)\cup D_1 \cup D_2 \cup R$) && && \\
%   \hline
%   \mr{3}{$\in V(G)$} & $2(m+n)+\lfloor(n+m)/2\rfloor$ &+& $2B$ &+& $2$    \\
%                 & (against half of &&(against $F$)&& (against\\
%                 & $V(G)\cup D_1 \cup D_2 \cup R$) &&&&$p$ and $d$)\\
%   \hline
%   $\in F$ &  $\le 2B$ (against $F$)     &&&&\\
%   \hline
%   \mr{3}{$\in R$} & $2(m+n)+\lfloor(n+m)/2\rfloor$ & + & $2B$ &+& $2$ \\
%                 & (against half of &&(against $F$) && (against \\
%                 & $V(G)\cup D_1 \cup D_2 \cup R$) &&&& $p$ and $d$)\\
%                 \hline
%   %\hline
%  \end{tabular}
\newcommand{\against}{}
 \begin{table}[t]
  \begin{tabular} {|l|l|l|}
    \hline
  \multicolumn{1}{|c|}{candidate} & \multicolumn{1}{c|}{head-to-head contests won} & \multicolumn{1}{c|}{against}\\
  \hline
  \hline
  \mr{2}{$d$}   &  $4(m+n)$ & (\against $D_1 \cup D_2$)\\
                & $+$ $ 2 B$&  (\against $F$)\\
  \hline
  \mr{3}{$p$}   &  $4(m+n)- (\frac{k\cdot (k-1)}{2}+k+1)$ & (\against $D_1$)\\
                &  $+$ $2B$  & (\against $F$) \\
                &  $+$ $1$ & (\against $d$)\\
  \hline
  \mr{2}{$\in D_1$} & $2(m+n)+\lfloor(n+m)/2\rfloor$  & (\against half of  $V(G)\cup D_1 \cup D_2 \cup R$)     \\
                & $+$ $2B$ & (\against $F$)  \\
  \hline
  \mr{3}{$\in D_2$} &  $2(m+n)+\lfloor(n+m)/2\rfloor$ & (\against half of $V(G)\cup D_1 \cup D_2 \cup R$) \\
                & $+$ $2B$ & (\against $F$) \\
                & $+$ $1$ &(\against $p$) \\
  \hline
  \mr{3}{$\in V(G)$} & $2(m+n)+\lfloor(n+m)/2\rfloor$ & (\against half of $V(G)\cup D_1 \cup D_2 \cup R$)  \\
                & $+$ $2B$ & (\against $F$)\\
                & $+$ $2$  & (\against $p$ and $d$)\\
  \hline
  $\in F$ &  $\le$ $ 2B$ & (\against $F$) \\
  \hline
  \mr{3}{$\in R$} & $2(m+n)+\lfloor(n+m)/2\rfloor$ & (\against half of $V(G)\cup D_1 \cup D_2 \cup R$)\\
                & $+$ $2B$ &(\against $F$) \\
                & $+$ $2$ & (\against  $p$ and $d$)\\
                \hline
  %\hline
 \end{tabular}

 \caption{Head-to-head contests won in the construction.
   Except for the candidate~$d$, every other candidate wins less than $4(m+n)+ 2B$
   head-to-head contests.}
 \label{tab:w1-budget-copeland-contests}
\end{table}

One can check that the score difference between $p$ and the unique winner~$d$ is $k+k\cdot (k-1)/2$, see \cref{tab:w1-budget-copeland-contests} for details.
% Using a similar reasoning as for the case of all-or-nothing prices, 
% we can show that the \probClique instance has a clique of order $k$ if and only if in the constructed election for unit prices, our preferred candidate~$p$ can be made a winner after $3/2 k\cdot (k-1)$ unit shifts.

Now, observe that with budget $B$, $p$ can only gain any additional points by
defeating candidates from $V(G)\cup R$ and by passing candidates in the
preference orders of the voters in $V_{\edges}$(the candidates from $F$ block all
other possibilities).
Since in the election~$E$, each vertex candidate~$c_i\in V(G)$ wins against
$p$ by $2k-3$ voters and each edge candidate~$q_e$ wins against $p$ by $1$ voter, 
the only possibility for $p$ to gain at least $k\cdot (k-1)/2$ additional points (within budget~$B$) is to 
win against $k\cdot (k-1)/2$ candidates from~$R$.
This means that for every successful shift action there is some set $E' \subseteq E(G), |E'|=k\cdot (k-1)/2$
such that $p$~was shifted to the top
position in the preference order of each the voter~$v_e$, $e \in E'$.
To achieve this, the whole budget must be used.
To gain the additional $k$ points (not from contests against candidates from $R$),
$p$ must have been shifted in front of $k$ vertex candidates $k-1$ times each.
This is only possible if $G[E']$ induces a clique of size~$k$.
\end{proof}

\subsection{Number of Affected Voters and Budget} 
\label{subsec:num affected voters}
For the number of affected voters, \ShiftBribery is $\wtwo$-hard for
each of Borda, Maximin, and Copeland$^\alpha$, and this is true for each family of
price functions that we consider: The result for all-or-nothing prices
follows almost directly from the $\np$-hardness proofs due to
Elkind et al.~\citep{elk-fal-sli:c:swap-bribery}.
Their reductions have to be adapted to work for \probSC (see below for the definition) 
rather than its restricted variant, \textsc{Exact Cover by 3-Sets}~\cite{gar-joh:b:int}, but this can be done quite easily. 
To obtain the result for unit prices, with some effort, it is still possible to carefully
modify their proofs, maintaining their main ideas.
Our proofs are included in~\ref{sec:omitted}.

\newcommand{\corBordawtwonumberaffectedvoters}{
  \probShiftBBorda parameterized by the number of affected voters is 
  $\wtwo$-hard for each price function family that we consider.
}

\newcommand{\corMaximinwtwonumberaffectedvoters}{
  \probShiftBMaximin parameterized by the number of affected voters is 
  $\wtwo$-hard for each price function family that we consider.
}

\newcommand{\corCopelandwtwonumberaffectedvoters}{
  \probShiftBCopeland parameterized by the number of affected voters is 
  $\wtwo$-hard for each price function family that we consider.
}

\newcommand{\corwtwonumberaffectedvoters}{
  \textsc{Borda, Maximin}, and \probShiftBCopeland parameterized by the number of affected voters are 
  $\wtwo$-hard for each price function family that we consider.
}

% \begin{corollary}\label{cor:Borda-w[2]-h-num_voters_affected}
%   \corBordawtwonumberaffectedvoters
% \end{corollary}
% 
% 
% \begin{corollary}\label{cor:Maximin-w[2]-h-num_voters_affected}
%   \corMaximinwtwonumberaffectedvoters
% \end{corollary}
% 
% 
% \begin{corollary}\label{cor:Copeland-w[2]-h-num_voters_affected}
%   \corCopelandwtwonumberaffectedvoters
% \end{corollary}

\begin{theorem}\label{thm:w[2]-h-num_voters_affected}
	\corwtwonumberaffectedvoters
\end{theorem}

% \subsection{Budget}
% \label{subsec:budget}
When we parameterize \ShiftBribery by the available budget, the
results fall between those from \cref{subsec:unit shifts,thm:w[2]-h-num_voters_affected}. 
In essence, the hardness proofs for all-or-nothing prices carry over from
the number of affected voters case to the budget case (and this
implies hardness for arbitrary prices and sortable prices), while the
results for the number of unit shifts carry over to the setting with
convex/unit prices. 

The hardness results translate because 
in the construction for all-or-nothing prices behind \cref{thm:w[2]-h-num_voters_affected} % \cref{cor:Borda-w[2]-h-num_voters_affected}, \cref{cor:Maximin-w[2]-h-num_voters_affected}, and \cref{cor:Copeland-w[2]-h-num_voters_affected},
the budget equals the parameter value
``solution size'' of the \probSC instance from which we reduce.
The $\classFPT$ results parameterized by the number of unit shifts
translate because for convex/unit prices the budget is an upper bound
on the number of possible unit shifts.
%  Doing so
% boils down to removing the padding candidates $F$ (they are necessary
% for the case of unit price functions, but are not necessary for the
% all-or-nothing case), setting the budgets to be $k$, and setting the
% price functions so that shifting $p$ to top costs $1$ for the voters
% that describe the input $\probSC$ instance (these are voters $v_1,
% v_3, v_5, \ldots, v_{2m+1}$ in the constructions from
% Corollary~\ref{cor:w[2]-h-num_voters_affected}); for all other voters
% we set the price of shifting $p$ to be $k+1$ (i.e., to exceed the
% budget). The $\classFPT$ results parametrized by the number of unit shifts
% translate because for convex/unit prices the budget is an upper bound
% on the number of possible unit shifts.
\newcommand{\corbordamaximinbudget}{
   \textsc{Borda} and \probShiftBMaximin parameterized by the budget~$B$ are $\wtwo$-hard for
  arbitrary, sortable, and all-or-nothing prices, and are in $\classFPT$
  for convex and unit prices with running time $O^*((2^{B}\cdot (B+1)\cdot B)^{B})$.
}

\begin{corollary}\label{cor:Borda+Maximin-budget}
 \corbordamaximinbudget
\end{corollary}

\newcommand{\corcopelandbudget}{%For each rational value~$\alpha$, $0 \leq \alpha \leq 1$,
  \probShiftBCopeland parameterized by the budget is
  $\wtwo$-hard for arbitrary, sortable, and all-or-nothing prices, and
  is $\wone$-hard for convex and unit prices.
}

For \probShiftBCopeland with unit prices, the budget equals the number of unit shifts in the reduction behind \cref{thm:Copeland-w[1]-c-num_shifts}. 
This implies the following corollary.

\begin{corollary}\label{cor:Copeland-budget}
  \corcopelandbudget
\end{corollary}

\section{Parameterizations by Election-Size Measures}\label{sec:pbesm}

In this section we consider \ShiftBribery parameterized by either the
number of candidates or the number of voters. Elections with few
candidates are natural in politics (for example, there is typically
only a handful of candidates in presidential elections) and elections
with few voters arise naturally in multiagent systems 

(for example, Dwork et al.~\citep{dwo-kum-nao-siv:c:rank-aggregation} suggested
election-based methods for aggregating results from several web search
engines; see also the work of Brandt et al.~\cite{bra-har-kar-see:c:few-voters} and of Chen et al.~\cite{che-fal-nie-tal:c:few-voters} for further
results and motivation regarding elections with few voters).  
For example, 
\citet{dwo-kum-nao-siv:c:rank-aggregation} suggested voting-based methods for aggregating results from several web search engines.
\citet{bet-guo-nie:j:dodgson-parametrized,bra-har-kar-see:c:few-voters,FelJanLokRosSau2010} 
 considered winner determination problems with few voters,
while \citet{che-fal-nie-tal:c:few-voters} considered voting control by adding alternatives with few voters.

\subsection{Number of Voters}

Let us now consider \ShiftBribery parameterized by the number of
voters. We have not found $\classFPT$ algorithms for our rules in this
setting, but we did find a general $\classFPT$-approximation scheme.
The idea of our algorithm is to use a scaling technique combined
with a brute-force search through the solution space. The scaling part
of the algorithm reduces the range of prices and then the brute-force
search finds a near-optimal solution. 
% The ideas underlying the proof
% are similar to those of Elkind and
% Faliszewski~\citep[Proposition~2]{elk-fal:c:shift-bribery}

\newcommand{\thmfptasnumvoters}{
  Let $\calR$ be a voting rule for which winner determination
  parameterized by the number~$n$ of voters is in $\classFPT$.
  There is a factor-$(1+\varepsilon)$ approximation algorithm solving $\calR$ \ShiftBribery
  in time $O^*(\lceil n/\varepsilon+1\rceil^n)$ times the cost of $\calR$'s winner
  determination.
}

\begin{theorem}\label{thm:fpt-as-num-voters}
 \thmfptasnumvoters
\end{theorem}

\begin{proof}

  Let $\calR$ be our voting rule and let $I = (C,V,\Pi,p)$ be an
  instance of the optimization version of $\calR$ \ShiftBribery. 
  Further, let
  $\varepsilon > 0$ be the desired approximation parameter. We will
  show how to compute a $(1+\varepsilon)$-approximate solution for
  $I$.

  We will need the following notation.  For a vector $\vv{b} = (b_1,
  \ldots, b_n)\in \mathbb{N}^{n}$, shift action $\costshift(\vv{b}) =
  (\costshift_1(b_1), \ldots, \costshift_n(b_n))$ is such that for
  each $i$, $1 \leq i \leq n$, value~$\costshift_i(b_i)$ is the
  largest number~$t_i$ such that $\pi_i(t_i) \leq b_i$. In other
  words, $\costshift(\vv{b})$ is the shift action that shifts $p$ in
  each voter $v_i$'s preference order as much as possible, without exceeding
  the per-voter budget limit~$b_i$.

%   Intuitively, the idea of the algorithm is as follows.
%   First, guess the maximum budget to spend on a single voter.
%   Then, the basic idea is to rescale and round the price of each shift
%   to be the smallest integer~$x$ between $0$ and~$\lceil n/\epsilon \rceil$
%   such that the old price is at most $x \cdot \epsilon/n$ of the maximum budget
%   per voter.
  
  Let $m = \left|C\right|$ and $n = \left|V\right|$.  The algorithm works as
  follows. First, we guess a voter $v_i \in V$ and a number $j \in
  [m]$.  We set $\pi_{\max} = \pi_i(j)$. We interpret $\pi_{\max}$ as
  the cost of the most expensive shift within the optimal
  solution. (Note that there are only $n\cdot m$ choices of $v_i$ and $j$).
  We set $K = \varepsilon \cdot \pi_{\max}/n$ and define a list~$\Pi'$
  of $n$ price functions as follows. For each $v_i \in V$ and $j \in
  \{0, \ldots, m\}$, we set
  \[
     \pi'_i(j) = \begin{cases} 
        \left\lceil \frac{\pi_i(j)}{K}\right \rceil & \text{if } \pi_i(j) \leq \pi_{\max} \\
        \frac{n\cdot (n+1)}{\varepsilon}+1                  & \text{otherwise}.
     \end{cases}
  \]
  We form an instance $I' = (C,V,\Pi',p)$. Note that for each $i$, $1
  \leq i \leq n$, and each $j$, $0 \leq j \leq m$, we have that if
  $\pi_i(j) \leq \pi_{\max}$ then $\pi'_i(j) \leq
  n/\varepsilon+1$.

  We compute (if there exists one) a lowest-cost shift action
  $\vv{s^{_{'}}}$ for $I'$ that ensures $p$'s victory and does not use
  shifts that cost $n\cdot (n+1)/\varepsilon+1$. We can do so by
  considering all vectors $\vv{b} = (b_1, \ldots, b_n)$ where 
  each~$b_i$ is in $\{0, \ldots, \left\lceil n/\varepsilon
  \right\rceil\}$: For each $\vv{b}$ we check if $p$ is an
  $\calR$-winner of $\shift(E,\costprimeshift(\vv{b}))$ and, if so,
  we store $\costprimeshift(\vv{b})$. We take $s'$ to denote the
  cheapest stored shift action (we make an arbitrary choice if there
  are several stored shift actions of the same cost; if there is no
  stored shift action then it means that our guess of $\pi_{\max}$ was
  incorrect). This process requires considering
  $O(\lceil n/\varepsilon+1\rceil^n)$ shift actions. 

  After trying each guess for $\pi_{\max}$, we return the cheapest
  successful shift action $s'$ that we obtained. We claim that this
  shift action has cost at most $(1+\varepsilon)\cdot \OPT(I)$.

  Consider an iteration where the guess of $\pi_{\max}$ is
  correct. Let $K$, $\Pi'$, and~$s'$ be as in that iteration and let $s$
  be a lowest-cost successful shift action for~$I$ (i.e., $p$ is
  an $\calR$-winner of $\shift(E,s)$ and $\Pi(s) = \OPT(I)$).
  We have the following inequalities:
  \[
  \Pi(s') \leq K\cdot \Pi'(s') \leq K\cdot \Pi'(s) \leq \Pi(s) + K\cdot n \leq \Pi(s) +
  \varepsilon\cdot \pi_{\max}.
  \]
  The first inequality follows because of the rounding in $\Pi'$, the
  second inequality follows because $s'$ is optimal for $I'$, the
  third inequality follows because for each $i$, $1 \leq i \leq n$, and
  each $j$, $0 \leq j \leq m$, $K\cdot\left\lceil \pi_i(j)/K \right\rceil \leq \pi_i(j) + K$,
  and the final equality follows by the definition of $K$.
  Since $\pi_{\max} \leq \OPT(I)$ and by the above inequalities, we have:
  \[ 
    \Pi(s')  \leq \Pi(s) + \varepsilon\cdot \pi_{\max} \leq
    (1+\varepsilon)\cdot \text{OPT}(I) .
  \]
  Thus the algorithm returns a $(1+\varepsilon)$-approximate solution.
  Our estimates of the running time given throughout the analysis
  justify that the algorithm runs in $\classFPT$ time.
\end{proof}

%\cref{thm:fpt-as-num-voters} follows by combining a brute-force search with price scaling.

% For general price functions, finding an $\classFPT$ algorithm parameterized
% by the number of voters seems challenging, but it might be easier if
% we restrict the class of price functions.  Indeed, for the case of
% all-or-nothing price functions we can obtain a very simple $\classFPT$
% algorithm.

Is it possible to obtain full-fledged $\classFPT$ algorithms for the
parameterization by the number of voters?  For the case of
all-or-nothing price functions, we do provide a very simple
$\classFPT$ algorithm, but for the other price-function classes this
is impossible (under the assumption that $\classFPT \neq \wone$).

\newcommand{\propfptnumvotersallornothingprices}{
  Let $\calR$ be a voting rule for which winner determination
  parameterized by the number of voters is in $\classFPT$. 
  $\calR$ \ShiftBribery parameterized by the number of voters is
  in $\classFPT$ for all-or-nothing prices.
}

\begin{proposition}\label{prop:fpt-num_voters-0/1-prices}
 \propfptnumvotersallornothingprices
\end{proposition}

\begin{proof}
  Note that with all-or-nothing prices, it suffices to consider
  shift actions where for each voter's preference order
 we either shift the preferred
  candidate to the top or do not shift him or her at all. Thus, given an
  election with $n$ voters it suffices to try each of the $2^n$
  shift actions of this form.
\end{proof}

% Unfortunately, this result does not extend to the other of our price
% function classes. Indeed, 
For Copeland we show $\wone$-hardness already for the case of unit
prices, via a somewhat involved reduction from a variant of the
$\probClique$ problem.  For Borda the same result holds due to a very
recent argument of \citet{BreFalNieTal2016b} (and one can verify that
their technique generalizes to the case of Maximin).

\newcommand{\thmwonecopelandnumvoters}{
  % For each rational value~$\alpha$, $0 \leq \alpha \leq 1$, 
  \probShiftBCopeland parameterized by the number of voters is $\wone$-hard for unit prices.
}
\begin{theorem}\label{thm:W[1]h-num_voters-copeland}
  \thmwonecopelandnumvoters
\end{theorem}

\newcommand{\numVertices}{\ensuremath{n_G}}
\newcommand{\numEdges}{\ensuremath{m_G}}
\newcommand{\numEdgesS}{\ensuremath{\frac{\Delta \cdot \numVertices}{2}}}
\newcommand{\numVerticesI}[1]{\ensuremath{\ell_{#1}}}
\newcommand{\selectColorCand}[2]{\ensuremath{S^{(#1)}_{#2}}}
\newcommand*{\LargerCdot}{\raisebox{-0.25ex}{\scalebox{1.4}{$\cdot$}}}
\newcommand{\selectSet}{\ensuremath{S}}
\newcommand{\fillerColorCand}[2]{\ensuremath{F^{(#1)}_{#2}}}
\newcommand{\fillerSet}{\ensuremath{F}}
\newcommand{\edgeCand}[1]{\ensuremath{c_{#1}}}
\newcommand{\edgeCandSet}{\ensuremath{R_{\edges}}}
\newcommand{\dummyCandSet}[1]{\ensuremath{D_{#1}}}
\newcommand{\dummySet}{\ensuremath{D}}

\newcommand{\incidentCands}[2]{\ensuremath{R^{(#1)}_{#2}}}
\newcommand{\arbincidentCands}{\ensuremath{R}}

\newcommand{\nameselectCand}{selection candidate\xspace}
\newcommand{\nameselectCands}{selection candidates\xspace}
\newcommand{\nameSelectCand}{Selection candidate\xspace}
\newcommand{\nameSelectCands}{Selection candidates\xspace}

\newcommand{\namefillerCand}{filler candidates\xspace}
\newcommand{\namefillerCands}{filler candidates\xspace}
\newcommand{\nameFillerCand}{Filler candidates\xspace}
\newcommand{\nameFillerCands}{Filler candidates\xspace}
\newcommand{\interestingPart}[1]{\ensuremath{\seq{P_{#1}}}}

\newcommand{\tmpdeg}{\numVerticesI{i}}

\begin{proof}
  	We give a reduction from \probMCClique---a variant of \probClique 
        %where the vertices in the input graph are colored with~$k$ colors and the task is to find a multicolored clique of size~$k$.
	which, given an undirected graph $G = (V(G), E(G))$, a non-negative integer $k\ge 0$, 
        and a coloring~$\col\colon V(G) \rightarrow [k]$,
        asks whether $G$~contains a size-$k$ \emph{multicolored clique}~$\clique\subseteq V(G)$,
        that is size-$k$ vertex subset~$\clique \subseteq V(G)$ 
        such that the vertices in $\clique$ are pairwise adjacent and 
        have pairwise distinct colors. 
	We assume in the following that~$G$ is regular, 
        that is, all vertices have the same degree, 
        and the coloring is proper, 
        that is, there are no edges between vertices of the same color.
	Note that even this restricted variant of \probMCClique parameterized by the clique size~$k$ 
        is $\wone$-hard~\cite{MS12}.

	Observe that in the given \probMCClique instance $(G,k,\col)$, 
        a \emph{multicolored clique} contains exactly one vertex of color~$i$ for each $i\in [k]$.
	We exploit this fact in our reduction which we now describe at high level: 
	For each color~$i \in [k]$ there is a color-gadget consisting of two voters~$v_{2i-1}$ and~$v_{2i}$.
	Each possible combination of shifting~$p$ 
        in the preference orders of these two voters will refer to a selection of one vertex of color~$i$.
	Each edge in~$E(G)$ is represented by one candidate;
        we call all candidates representing edges \emph{edge candidates}.
	In the constructed election, 
        our preferred candidate~$p$ loses all head-to-head contests against the edge candidates
        and has a score difference of $\binom{k}{2}$ with the original winner~$d$.
        With some technical gadget, 
        in order to make~$p$ a winner, 
        $p$ has to win head-to-head contests against at least~$\binom{k}{2}$ edge candidates after the shifting.
	Furthermore, 
        to win in a head-to-head contest against the edge candidate representing edge~$\{x,y\}$, 
        one has to ``select''~$x$ and~$y$ in the respective color-gadgets.
	Thus, if~$p$ can be made a winner, 
        then the corresponding selected vertices in the color-gadgets form a multicolored clique.
	We now describe the construction in detail.
	
	First, set~$\numVertices = |V(G)|$ and~$\numEdges = |E(G)|$,
        and let~$\Delta$ be the vertex degree in the regular graph~$G$, 
        implying~$\numEdges = \Delta \cdot \numVertices / 2$.
	%For technical reasons, 
        %we assume in the following without loss of generality that~$k \cdot \numEdges + \Delta \le \numVertices^3$.
        Furthermore, for each~$i \in [k]$, 
        let $\tmpdeg$ denote the number of vertices with color~$i$.
        %let~$X_i = \{x^i_1, \ldots, x^i_{\ell_i}\}$ be the set of vertices in~$G$ colored with~$i$.
	We start with initially empty sets~$V$ of voters and~$C$ of candidates.
	First, add to $C$ the preferred candidate~$p$ and the candidate~$d$.
        Candidate $d$ will be the unique winner in the original election.
        Then, we extend the set~$C$ of candidates by adding
        \begin{enumerate}[i)]
        \item for each color~$i \in [k]$ and each index~$j \in [\tmpdeg+1]$,
          add a set~$\selectColorCand{i}{j}$ of $\numVertices^3$~\emph{\nameselectCands}
          and a set~$\fillerColorCand{i}{j}$ of $2\numVertices^3$~\emph{\namefillerCands}; 
          we say that $\selectColorCand{i}{j}$ corresponds to the $j$th vertex with colored $i$ if $j\le \tmpdeg$.\\
          Let $\selectSet$ be the set of all \nameselectCands and let $\fillerSet$ be the set of all \namefillerCands.
          For each color~$i\in [k]$,
          define~$\selectColorCand{i}{\LargerCdot}\coloneqq\bigcup_{j=1}^{\tmpdeg}\selectColorCand{i}{j}$ and
          define~$\selectColorCand{-i}{\LargerCdot}\coloneqq\selectSet\setminus\selectColorCand{i}{\LargerCdot}$; 
          equivalently, define~$\fillerColorCand{i}{\LargerCdot}\coloneqq\bigcup_{j=1}^{\tmpdeg}\fillerColorCand{i}{j}$ and 
          define~$\fillerColorCand{-i}{\LargerCdot}\coloneqq\fillerSet\setminus\fillerColorCand{i}{\LargerCdot}$; 
        \item for each edge~$e\in E(G)$, 
          add an \emph{edge candidate}~$\edgeCand{e}$;
          we say that $\edgeCand{e}$~corresponds to edge~$e$.
          Denote by $\edgeCandSet$ the set of all edge candidates. 
          For each color~$i \in [k]$ and each index~$j \in [\tmpdeg]$, 
          we denote by $\incidentCands{i}{j}$ 
          the set of edge candidates corresponding to the edges that are incident to the $j$th vertex with color~$i$.
          Accordingly, we define $\incidentCands{i}{\LargerCdot}\coloneqq\bigcup_{j\in[\tmpdeg]}{\incidentCands{i}{j}}$ 
          and define $\incidentCands{-i}{\LargerCdot}\coloneqq\bigcup_{i' \in [k]\setminus\{i\}}{\incidentCands{i'}{\LargerCdot}}$;
        \item add three sets of dummy candidates, denoted by $\dummyCandSet{1}, \dummyCandSet{2}, \dummyCandSet{3}$, each of size~$\numVertices^5$.\\
          Define~$\dummySet\coloneqq\dummyCandSet{1}\cup \dummyCandSet{2}\cup
          \dummyCandSet{3}$.
          Accordingly, for each $i \in [3]$, we define $\dummyCandSet{-i} \coloneqq \dummySet
          \setminus \dummyCandSet{i}$.
        \end{enumerate}
        
        Note that $|C|=3\numVertices^3\cdot\sum_{i=1}^{k}(\tmpdeg+1)+\numEdges+3\numVertices^5 = 3(k + \numVertices + \numVertices^2)\cdot\numVertices^3 + \numEdges$.

	Next, we specify the voters.
	The voter set~$V$ consists of $2k + 3$ voters:~$v_{2i-1},v_{2i}$ for all~$i \in [k]$ and~$v_{2k+1}, v_{2k+2}, v_{2k+3}$.
	The first~$2k$ voters form the~$k$ color-gadgets as mentioned above. 
	The last three voters are necessary to obtain the desired scores.
	We begin with the color-gadgets. 
	For each color~$i \in [k]$, we first specify the ``interesting'' part
        of the preference orders of voters~$v_{2i-1}$ and~$v_{2i}$ by setting
        \noindent
        \begin{align*}
          \interestingPart{2i-1} & \coloneqq \seq{\selectColorCand{i}{2}}
          \pref \seq{\incidentCands{i}{1}} \pref 
          \seq{\fillerColorCand{i}{2}} \pref \ldots \pref
          \seq{\selectColorCand{i}{\tmpdeg+1}} \pref
          \seq{\incidentCands{i}{\tmpdeg}} \pref
          \seq{\fillerColorCand{i}{\tmpdeg+1}} \pref p,\\
          \interestingPart{2i} & \coloneqq \seq{\selectColorCand{i}{\tmpdeg}}
          \pref \seq{\incidentCands{i}{\tmpdeg}} \pref 
          \seq{\fillerColorCand{i}{\tmpdeg}} \pref \ldots \pref 
          \seq{\selectColorCand{i}{1}} \pref
          \seq{\incidentCands{i}{1}} \pref
          \seq{\fillerColorCand{i}{1}} \pref p.
	\end{align*}
        \noindent
        Note that both $\interestingPart{2i-1}$ as well as
        $\interestingPart{2i}$ have exactly $\tmpdeg \cdot
        (3\numVertices^3 + \Delta) + 1$~candidates.
	Having these ``interesting'' parts, we specify the preference
        orders of voters~$v_{2i-1}$ and $v_{2i}$ as
	% \begin{align*}
        %         v_{2i-1}\colon & d \pref \seq{\selectColorCand{i}{1}} \pref \seq{\fillerColorCand{i}{1}} \pref \seq{\dummyCandSet{1}} \pref \interestingPart{2i-1} \pref p
        %         \\ &\pref \seq{\dummySet \setminus \dummyCandSet{1}} 
        %         \pref \seq{\selectColorCand{-i}{\LargerCdot}\cup\incidentCands{-i}{\LargerCdot}} \pref \seq{\fillerColorCand{-i}{\LargerCdot}}\\
	% 	v_{2i}\colon & \revseq{\selectColorCand{-i}{\LargerCdot} \cup \incidentCands{-i}{\LargerCdot}} \pref \seq{\selectColorCand{i}{\tmpdeg+1}} \pref \seq{\fillerColorCand{i}{\tmpdeg+1}} \pref
        %          \seq{\dummySet \setminus \dummyCandSet{1}} \pref \interestingPart{2i} \pref p
        %         \\ & \pref \revseq{\dummyCandSet{1}} \pref \revseq{\fillerColorCand{-i}{\LargerCdot}}
	% \end{align*}
	\begin{align*}
                d \pref \seq{\selectColorCand{i}{1}} \pref
                \seq{\fillerColorCand{i}{1}} \pref
                \seq{\dummyCandSet{1}} \pref \interestingPart{2i-1}
                \pref \seq{\dummyCandSet{-1}} 
                \pref \seq{\selectColorCand{-i}{\LargerCdot}\cup\incidentCands{-i}{\LargerCdot}} \pref \seq{\fillerColorCand{-i}{\LargerCdot}},\\
		\revseq{\selectColorCand{{-i}}{\LargerCdot} \cup \incidentCands{{-i}}{\LargerCdot}} \pref \seq{\selectColorCand{i}{\tmpdeg+1}} \pref \seq{\fillerColorCand{i}{\tmpdeg+1}} \pref
                 \revseq{\dummyCandSet{-1}} \pref \interestingPart{2i} 
                \pref \revseq{\dummyCandSet{1}} \pref d \pref \revseq{\fillerColorCand{-i}{\LargerCdot}}.
	\end{align*}
	Let~$\arbincidentCands\subseteq \edgeCandSet$ be a set of $\binom{k}{2}$ arbitrary edge candidates and
        let $\arbincidentCands'\coloneqq\edgeCandSet\setminus \arbincidentCands$ be the set consisting of the remaining edge candidates.
        % and let $R\coloneqqC\setminus (\{d, p\} \cup \edgeCandSet \cup \dummySet)$.
	To complete~$V$, we set the preference orders of
        voters~$v_{2k+1}$, $v_{2k+2}$, and $v_{2k+3}$ as 
	\begin{align*}
		d \pref \seq{\edgeCandSet} \pref \seq{\dummyCandSet{1}} \pref p \pref \seq{\dummyCandSet{-1}} \pref \seq{\selectSet\cup\fillerSet}, \\
		\seq{\arbincidentCands'} \pref d \pref \seq{\arbincidentCands} \pref \seq{\dummyCandSet{2}} \pref p \pref \seq{\dummyCandSet{-2}} \pref \seq{\selectSet\cup\fillerSet}, \\
		\seq{\selectSet\cup\fillerSet} \pref \seq{\edgeCandSet} \pref \seq{\dummyCandSet{3}}  \pref p \pref  \seq{\dummyCandSet{-3}} \pref d.
	\end{align*}
	Finally, set the budget~$B \coloneqq (\numVertices+k) \cdot (3\numVertices^3 + \Delta)$ and the election~$E \coloneqq (C,V)$. 
	This completes the construction (recall that we consider unit prices).

	Before showing the correctness of the construction we determine the score of each candidate.
	Recall that the score of candidate~$d$ 
        in Copeland for odd number of voters is equal to the number of candidates against whom~$d$ wins the head-to-head contests.
	To this end, let $c_d \in D$, let $c_u \in S\cup F$, 
        and let $c_e \in \edgeCandSet$ with $e$ being an edge incident to the $r$th vertex of color $i$ and the $s$th vertex of color~$j$.
        The scores are:
	\begin{align*}
		\s{E}{d} 	= {} & {} |C \setminus
                \arbincidentCands'| - 1 = |C|-\numEdgesS+\binom{k}{2} - 1,\\
		\s{E}{p} 	= {} & {} |C\setminus
                (\edgeCandSet\cup\selectSet)| - 1 =
                |C|-\numEdgesS-\numVertices^3\cdot (\numVertices+k) - 1,\\
		\s{E}{c_d} 	< {} & {} |C\setminus \edgeCandSet| -
                1 = |C|-\numEdgesS - 1 < \s{E}{c},\\
		\s{E}{c_e}	\le {} & {} |C\setminus
                (\selectColorCand{i}{r} \cup
                \selectColorCand{i}{r+1} \cup
                \selectColorCand{j}{s} \cup \selectColorCand{j}{s+1})|  - 1 
                                \\= {} & {} |C|-4\numVertices^4  - 1 
				<  \s{E}{c},\\
		\s{E}{c_u}	\le {} & {} |C\setminus
                (\dummyCandSet{1}\cup \{d\})| - 1  < |C|-\numVertices^5 < \s{E}{c}.
	\end{align*}
	% \begin{align*}
	% 	\s{E}{d} 	= {} & {} |D \cup C([k]) \cup C(E_k) \cup\{p\}| \\= {}&{} 3\numVertices^5 + 3\numVertices^3(\numVertices + k) + \binom{k}{2} \\
	% 	\s{E}{p} 	= {} & {} |D \cup C^d([k])| = 3\numVertices^5 + 2\numVertices^3(\numVertices + k) \\
	% 	\s{E}{x} 	< {} & {} |D \cup C([k])| = 3\numVertices^5 + 3\numVertices^3(\numVertices + k) \\< {}&{} \s{E}{c}\\
	% 	\s{E}{c_e}	< {} & {} |D \cup C([k]) \cup C(E(G)) \cup \{c,p\}| - |C^s_{j_1}(i_1) \cup C^s_{j_1+1}(i_1) \cup C^s_{j_2}(i_2) \cup C^s_{j_2+1}(i_2)| \\
	% 					= {} & {} 3\numVertices^5 + 3\numVertices^3(\numVertices + k) + \numEdges + 2 - 4\numVertices^3 \\ < {}&{} \s{E}{c}\\
	% 	\s{E}{c_x}	< {} & {} |C([k]) \cup C(E(G)) \cup \{p\} \cup D_2 \cup D_3| \\
	% 					= {} & {} 3\numVertices^3(\numVertices + k) + \numEdges + 1 + 2\numVertices^5 \\ < {}&{}\s{E}{c}.
	% \end{align*}
	Hence, candidate~$d$ is the unique winner with score
        $|C|-\Delta\cdot \numVertices/2+\binom{k}{2} - 1$ 
        and~$p$ needs~$(\numVertices+k) \cdot n^3_G + \binom{k}{2}$ additional points to become a winner.
	Observe that in each preference order, 
        when $d$ beats $p$,
        it is at least~$\numVertices^5 > B$ positions ahead of~$p$.
	Hence, during the shifting of~$p$
        the score of~$d$ cannot be decreased, 
        implying that $p$ can only win by increasing its score.
        This can only happen if $p$ wins against candidates in~$\edgeCandSet \cup \selectSet$.
	Furthermore, note that for all~$c_u \in \selectSet$ and all~$\edgeCand{e} \in \edgeCandSet$,
        it holds that  
	\begin{align*}
		N_E(\edgeCand{e},p) & = N_E(p,\edgeCand{e}) + 7,\\
		N_E(c_u,p) & = N_E(p,c_u) + 1.
	\end{align*}
	In particular, 
        this means that~$p$ has to pass $c_u$ in just one voter's preference order in order to win against~$c_u$. 
	To win the contest against~$c_e$, however, 
        $p$ has to pass $c_e$ in at least four voters' preference orders.
	Furthermore, 
        the number of selection candidates is much higher than the number of edge candidates. 
	Hence, $p$ has to be shifted over basically all selection candidates to obtain the necessary points.
	Because of the distribution of the selection candidates, 
        it follows that the only way of winning against all selection candidates 
        is to spend~$(\tmpdeg+1) \cdot (3\numVertices^3 + \Delta)$~shifts in each pair of voters~$v_{2i-1}, v_{2i}$.
	This means that for each color~$i$, 
        we can select exactly one vertex of index~$j\in [\tmpdeg]$ 
        such that~$p$ passes in~$v_{2i-1}$'s and in~$v_{2i}$'s
        preference orders all candidates from~$\incidentCands{i}{j}$,
        implying that~$p$ wins against all edge candidates with the corresponding edges incident to the $j$th vertex of color~$i$.
	Thus, if~$p$ wins against~$\binom{k}{2}$ edge candidates, 
        then the corresponding edges describe a multicolored clique in~$G$.

	We now formalize this idea when proving the correctness of our construction: 
        $(G,k,\col)$ is a yes-instance of \probMCClique if and only if~$(E,p,B)$ is a yes-instance of \probShiftBCopeland with unit prices. 

        \medskip
	``$\Rightarrow:$''
	Let~$\clique \subseteq V(G)$ be a size-$k$ multicolored clique.
	We construct a shift-action making~$p$ a winner in~$\shift(E,\vv{s})$ with at most~$B$ shifts.
	To this end let~$\clique = \{u^{(1)}_{j_1}, \ldots, u^{(k)}_{j_k}\}$ be the multicolored clique
        such that for each $i \in [k]$, $u^i_{j_i}$ is the $j_i$th vertex of color~$i$.
	For each color~$i \in [k]$,
        perform the following shifts in $v_{2i-1}$'s and $v_{2i}$'s
        preference orders:
	Shift~$p$ in $v_{2i-1}$'s preference order directly in front
        of the first candidate in $\seq{\selectColorCand{i}{j_i+1}}$ 
        and shift~$p$ in $v_{2i}$'s preference order directly in front of the first candidate in $\seq{\selectColorCand{i}{j_i}}$.
%  and in~$v_{2i}$ in directly behind~$\vv{C(x_i)}$, that is, in~$v_{2i-1}$ \emph{and} in~$v_{2i}$ we shift~$p$ over~$\vv{C(I(x_i))}$.
	We perform no shifts in the last three preference orders.
	This completes the description of~$\vv{s}$.
	First, observe that altogether these are~$(\numVertices+k) \cdot (3\numVertices^3 + \Delta) = B$ shifts.
	Next, observe that $p$ passes in $v_{2i-1}$'s preference order all selection candidates from~$\bigcup_{r=j_i}^{\tmpdeg}\selectColorCand{i}{r+1}$.
	Furthermore, $p$ passes in $v_{2i}$'s preference order all selection candidates from~$\bigcup_{r \in [j_i]}\selectColorCand{i}{r}$.
	Recall that for all~$c_u \in \selectColorCand{i}{\LargerCdot}$,
        we have~$N_E(c_u,p) = N_E(p,c_u) + 1$.
        This implies that $p$ wins in~$\shift(E,\vv{s})$ against all selection candidates.
	Furthermore, note that for each edge~$e = \{u^i_{j_i}, u^{h}_{j_h}\}$ with~$u^i_{j_i}, u^{h}_{j_h} \in \clique$,
        $p$ also passes $\edgeCand{e}$ in the preference orders of
        voters~$v_{2i-1}, v_{2i}, v_{2j-1},$ and $v_{2j}$ as~$\edgeCand{e} \in \incidentCands{i}{j_i}$ and~$\edgeCand{e} \in \incidentCands{h}{j_h}$.
	As~$N_E(\edgeCand{e},p) = N_E(p,\edgeCand{e}) + 7$ for each~$\edgeCand{e} \in C(E(G))$, 
        $p$ wins against the~$\binom{k}{2}$ edge candidates corresponding to the edges in $G[\clique]$.
	Altogether $p$ gains~$(\numVertices+k) \cdot n^3_G + \binom{k}{2}$ additional points after the shifting and, 
        hence, $p$ is a winner in~$\shift(E,\vv{s})$.

        \medskip
	``$\Leftarrow:$'' Let $\vv{s} = (s_1, \ldots, s_{2k+3})$ be a shift action making~$p$ a winner in~$\shift(E,\vv{s})$. 
	This means that $p$ gains at least~$(\numVertices+k)\cdot n^3_G + \binom{k}{2}$ points through the shifting.
	% For convenience, for~$i \in [k]$ let~$s_{2i-1}$ and~$s_{2i}$ be the number of shifts in~$v_{2i-1}$ and~$v_{2i}$, respectively.
	% Furthermore, let $s_{2k+1}$, $s_{2k+2}$, and $s_{2k+3}$ denote the number of shifts in~$v_1$, $v_2$, and $v_3$.
	% Denote the number of vertices with color~$i$ by~$\ell_i$.
	As $p$ already wins in the original election~$E$ against all dummy candidates in~$D$ and 
        as $|D_1| = |D_2| = |D_3| > B$, by the construction of the preference orders,
        we can assume that~$p$ is shifted in no preference order from the last three voters, 
        and that $p$ is shifted at most~$\tmpdeg \cdot
        (3\numVertices^3 + \Delta)$ positions in~$v_{2i-1}$'s
        or~$v_{2i}$'s preference order, 
        that is, $s_{2k+1} = s_{2k+2} = s_{2k+3} = 0$, 
        $s_{2i-1} \le \tmpdeg \cdot (3\numVertices^3 + \Delta)$, 
        and~$s_{2i} \le \tmpdeg \cdot (3\numVertices^3 + \Delta)$.

	We next show that, for each color~$i$ there is at most one vertex of color $i$
        such that~$p$ is shifted in~$v_{2i-1}$'s and in~$v_{2i}$'s
        preference orders over some edge candidates corresponding to this vertex's incident edges.
	In other words, 
        if~$p$ wins in~$\shift(E, \vv{s})$ against the edge candidates~$\edgeCand{e_1}$ and~$\edgeCand{e_2}$ 
        with~$e_1 = \{x_1, y_1\}$ and~$e_2 = \{x_2,y_2\}$
        such that $x_1$ and~$x_2$ are both colored with~$i$, 
        then~$x_1 = x_2$.
	Suppose for the sake of contradiction that~$p$ wins in~$\shift(E, \vv{s})$ 
        against the edge candidates~$\edgeCand{e_1}$ and~$\edgeCand{e_2}$ 
        such that~$e_1 = \{x_1, y_1\}$, $e_2 = \{x_2,y_2\}$, $x_1$ and~$x_2$ are both colored with~$i$, and~$x_1 \neq x_2$.
	A lower bound for $s_{2i-1} + s_{2i}$ under this assumption can be seen as follows.
	Recall that in order to make~$p$ a winner one must shift~$p$ over each selection
	candidate from~$\selectColorCand{i}{\LargerCdot}$.
	This can only be done by shifting $p$ at least once (in $v_{2i-1}$ or $v_{2i}$) over
        each of altogether $\ell_i \cdot (3\numVertices^3 + \Delta)$ candidates from
        $X^{(i)}\coloneqq\selectColorCand{i}{\LargerCdot} \cup \incidentCands{i}{\LargerCdot}
        \cup \fillerColorCand{i}{\LargerCdot}$.
	Furthermore, at least two sets of filler candidates (each containing $2\numVertices^3$
        candidates and corresponding either to $x_1$ or to $x_2$) and the two edge candidates
        $\edgeCand{e_1}$ and~$\edgeCand{e_2}$ have to be passed by~$p$ a second time.
	Hence, $s_{2i-1} + s_{2i} \ge \ell_i \cdot (3\numVertices^3 +
        \Delta) + 4\numVertices^3 + 2 = (\tmpdeg+1) \cdot (3 \numVertices^3 + \Delta) + \numVertices^3 - \Delta + 2$.
	Thus, the remaining budget~$B'$ is:
	\begin{align*} 
		B' = {} & {}\sum_{j \in [k]\setminus \{i\}} s_{2j-1} + s_{2j} \\
			= {} & {} B - (s_{2i-1} + s_{2i}) \\ 
			\le {} & {} (\numVertices + k) \cdot
                        (3\numVertices^3+\Delta) - (\ell_i+1) \cdot (3 \numVertices^3 + \Delta) - \numVertices^3 + \Delta - 2 \\
% 			= {} & {} (\numVertices - \ell_i + 1 - k)(2\numVertices^3+\Delta) - \numVertices^3 - 1\\
			= {} & {}  -(\numVertices^3 - \Delta + 2) +
                        (3\numVertices^3+\Delta) \cdot \sum_{j \in [k]\setminus\{i\}} (\ell_j + 1)
	\end{align*}
	The first equation holds since $s_{2k+1} = s_{2k+2} = s_{2k+3} = 0$,
	To see the second equation and the inequality apply the definitions of $B'$ and $B$
        and the above discussed lower bound for $s_{2i-1} + s_{2i}$.

	Next, observe that in order to win in~$\shift(E,\vv{s})$, 
        $p$ has to win head-to-head contests against at least~$|\selectSet| - \numEdges$~selection candidates.
	This implies that for each~$j \in [k]$, 
        at least $(\ell_j +1) \cdot (3\numVertices^3 + \Delta) - \numEdges$ unit shifts in the color-gadget for color~$j$ are needed.
	The number of required unit shifts sum up to:
	\begin{align*}
		& \sum_{j \in [k]\setminus\{i\}} ((\numVerticesI{j}+1)
                \cdot (3 \numVertices^3 + \Delta) - \numEdges) \\
		= {} & {}  - (k-1) \cdot \numEdges +
                (3\numVertices^3+\Delta)\cdot \sum_{j \in [k]\setminus\{i\}} (\numVerticesI{j} + 1) \\
		> {} & {} - (\numVertices^3 - \Delta + 2) + (3\numVertices^3+\Delta)\cdot \sum_{j \in [k]\setminus\{i\}} (\numVerticesI{j} + 1).
	\end{align*}
	To see the last inequality recall that $\numEdges=\Delta \cdot \numVertices/2$
	and $k<\numVertices$ and observe that $\Delta \cdot \numVertices^2/2< \Delta \cdot \numVertices^3/2-\Delta+2$.
	Finally, it follows that the number of additional shifts needed to make $p$ a winner
        exceeds the remaining budget~$B'$; a contradiction.
% 	Hence, the remaining budget~$B'$ is insufficient to make~$p$ winning against $\numVertices^3 - \numEdges$~candidates from each set~$C^s_r(j)$ for all~$j \in [k]$ and~$r \in [\ell_i]$.
% 	This implies, that~$p$ does not win in~$\shift(E,\vv{s})$, a contradiction.

	Next, we show how to construct a multicolored clique given the shift-action $\vv{s}$ that makes~$p$ a winner. 
	Denote by~$X$ the set of edge candidates against whom $p$ wins already in~$\shift(E,\vv{s})$, 
        let~$E(X)$ be the corresponding set of edges, 
        and let~$V(X)$ be the vertices incident to these edges, 
        that is,~$V(X) = \bigcup_{e \in E(X)}e$.
	We now prove that~$V(X)$ forms a multicolored clique of size~$k$.
	We already showed above that for each pair~$e_1 = \{x_1, y_1\} \in E(X)$ and~$e_2 = \{x_2,y_2\} \in E(X)$ such that~$x_1$ and~$x_2$ are both colored with~$i$ it follows that~$x_1 = x_2$.
	Hence, $V(X)$ contains at most one vertex of each color, 
        implying that~$|V(X)| \le k$.
	Even though we assume that~$p$ wins against all selection candidates in~$\shift(E, \vv{s})$, 
        $p$ has to further win against at least~$\binom{k}{2}$ edge candidates in~$\shift(E, \vv{s})$ in order to have the same score as~$d$ and win the election.
	Hence, $|E(X)| \ge \binom{k}{2}$.
	Thus, $V(X)$ forms a clique of size exactly~$k$ in~$G$. 
	Since for each color there is at most one vertex in~$V(X)$ and there are~$k$ colors and thus~$V(X)$ is multicolored.
\end{proof}

% Since the Copeland rule is polynomial-time computable, the above
% $\wone$-hardness result shows that unless unlikely complexity class
% collapses occur, it is impossible to improve
% \autoref{thm:fpt-as-num-voters} to speak of an \emph{exact} $\classFPT$
% algorithm. Indeed, the $\classFPT$ approximation scheme from
% \autoref{thm:fpt-as-num-voters} seems to be the best we can hope for,
% for the parameterization by the number of voters.

Since determining the winner under the Copeland rule can be done in polynomial-time, 
the above $\wone$-hardness result shows that unless unlikely complexity class collapses occur, 
it is impossible to improve \cref{thm:fpt-as-num-voters} to provide an \emph{exact} $\classFPT$ algorithm. 
%Indeed, the $\classFPT$ approximation scheme from \cref{sb-thm:fpt-as-num-voters} 
%seems to be the best we can hope for when parameterizing by the number~$n$ of voters.
Since the publication of our conference paper~\cite{BreCheFalNicNie2014}, 
\citet{BreFalNieTal2016b} complemented our $\wone$-hardness result by showing that for the Borda and maximin rules, 
\ShiftBribery parameterized by the number of voters is $\wone$-hard.

\subsection{Number of Candidates}

As opposed to almost all other (unweighted) election problems ever
studied in computational social choice, for \ShiftBribery (and for
bribery problems in general) the parameterization by the number of
candidates is one of the most notorious ones. In other election
problems, the natural, standard attack is to give integer linear program~(ILP) formulations and
use Lenstra's algorithm~\cite{len:j:integer-fixed}.  For instance,
this has been applied to winner
determination~\cite{bar-tov-tri:j:who-won}, %for Dodgson
control~\cite{fal-hem-hem-rot:j:llull}, 
%\cite{fal-hem-hem-rot:c:llull}, %Copeland$^\alpha$~
possible winner~\cite{bet-hem-nie:c:parameterized-possible-winner,BreCheNieWal2015},
and lobbying~\cite{BCHKNSW14} problems.  This works because with
$m$~candidates there are at most $m!$~different preference orders and
we can have a variable for each of them in the ILP.  However, in our
setting this approach fails.  The reason is that in bribery problems
the voters are not only described by their preference orders, but also
by their prices.  This means that we cannot easily lump together a
group of voters with the same preference orders anymore, and we have
to treat each of them individually (however, for the case of sortable
and all-or-noting price function families, \citet{BreFalNieTal2015b}
found a novel way of employing the ILP approach to obtain $\classFPT$~algorithms).

Dorn and Schlotter \citep{dor-sch:j:parameterized-swap-bribery}
have already considered the complexity
of \textsc{Swap Bribery} parameterized by the number of candidates.
However, their proof implicitly assumes that each voter has the same
price function and, thus, it implies that \ShiftBribery (parameterized
by the number~$m$ of candidates) is in $\classFPT$ for unit prices, but not
necessarily for the other families of price functions. 
Whenever the number of different prices or different
price functions is upper-bounded by some constant or, at least, by some function
only depending on $m$, Dorn and Schlotter's approach
can be adapted.
% Some applications of this are captured by the following corollary.

% \newcommand{\corilpnumcandidatescases}{
%  IN IS NOT CLEAR WHETHER THE FOLLOWING IS CORRECT
%  Let $\calR$ be a voting rule for which winner-determination
%  parameterized by the number of candidates is in $\classFPT$.
%  HOWEVER CORRECT SHOULD BE AT LEAST
%  Let $\calR$ be a voting rule for which winner-determination
%  can be expressed via linear inequalities.
% \begin{itemize}
%  \item 
% (1) $\calR$ \ShiftBribery parameterized by the number of candidates
%   is in $\classFPT$ for unit prices.
%  \item 
% (2) $\calR$ \ShiftBribery parameterized by combined parameter
%    number of candidates and budget is in $\classFPT$.
%  \item 
% (3) $\calR$ \ShiftBribery parameterized by combined parameter
%   number of candidates and maximum value in the price functions is in $\classFPT$.
% \end{itemize}
% }
% 
% \begin{corollary}
%   \label{cor:candidate-ILP-cases}
%   \corilpnumcandidatescases
% \end{corollary}

For the more general price function families (i.e., for convex and
arbitrary price functions), we were neither able to find alternative
$\classFPT$ attacks nor to find hardness proofs (which, due to the
limited number of candidates one can use and the fact that the voters
are unweighted, seem particularly difficult to design).  However, as
in the case of parameterization by the number of voters
(\cref{thm:fpt-as-num-voters}), we can show that there is an
$\classFPT$-approximation scheme for the number of candidates when the
prices are sortable.

\newcommand{\classFPTasnumcandidates}{
  Let $\calR$ be a voting rule for which winner determination
  parameterized by the number~$m$ of candidates is in $\classFPT$.
%   There is an $\classFPT$-approximation scheme for $\calR$~\ShiftBribery
%   parameterized by the number of candidates for sortable prices.
% %
%   To obtain approximation ratio $(1+\varepsilon)$, the algorithm runs
%   in time $O^*(M^{M\lceil\ln (M/\varepsilon)\rceil+1})$ (where $M = m\cdot m!$) times the
%   cost of $\calR$'s winner determination.  
  There is a factor-$(1+2\varepsilon + \varepsilon^2)$ approximation algorithm solving $\calR$ \ShiftBribery for sortable prices
  in time $O^*(M^{M\cdot \lceil\ln (M/\varepsilon)\rceil+1})$ (where $M = m\cdot m!$) times the
  cost of $\calR$'s winner determination.  
}
\begin{theorem}\label{thm:fpt-as-num-candidates}
  \classFPTasnumcandidates
\end{theorem}

\begin{proof} 
  Let $I=(C,V,\Pi,p)$ be an instance of (the optimization variant
  of) $\calR$~\ShiftBribery with $\Pi$ being sortable,
  and let $\varepsilon > 0$ be a rational number.
  We show how to compute a successful  shift action~$\vv{s}$
  for $I$ such that $\Pi(s) \leq (1+\varepsilon)\cdot\OPT$.
  Our algorithm will have $\classFPT$ running time for the combined parameter
  number of candidates and~$\varepsilon$.

  Before we start with the algorithm, we describe a helpful interpretation
  of solutions for $\calR$~\ShiftBribery with sortable prices.
  We call a set of voters all having the same preference order a \emph{voter block}.
  Let $m'\le m!$ denote the number of voter blocks in our election.
  We assume that the voters from each voter block appear consecutively
  and are ordered according to the price of shifting~$p$ to the top position 
  (if they did not, then we could sort them in this way). 
  Using the fact that the price function list~$\Pi$ in $I$ is sortable, 
  simple exchange arguments show that there is a successful
  minimum-cost shift action $\vv{s}$ such that for each two voters $v_i$ and $v_j$
  from the same block, it holds that $i<j$ implies $s_i\ge s_j$.
  Informally, this means that for each two voters with the same preference order, one
  never shifts $p$ farther in the preference order of the more expensive voter.

  Based on the above observation, we can conveniently express shift actions for
  our election in an alternative way. 
%  
%
  % Each single shift operation of some successful shift action can now be characterized
  % by the voter block it belongs to and by the position to which it shifts.
  % There are $m'm$ such \emph{shift types}.
  % As a consequence, one can express shift actions for sorted prices in an alternative way,
  % where we ask for the number of shift for each shift type instead of the number of
  % shift for each voters.
%
  We say that a \emph{stepwise shift action}
  $\vv{\mu}=(\mu_1^1,\dots,\mu_{m-1}^1,\mu_1^2,$
  $\dots,\mu_{m-1}^2,$ $\dots,\mu_1^{m'},$ $\dots,\mu_{m-1}^{m'})$
  is a vector of natural numbers, such that for each $x \in [m']$, $\mu_0^x,\dots,\mu_{m-1}^x$ describe the shifts
  $p$~makes in the block~$x$ in the following way: for each $y \in [m-1]$, $\mu_y^x$ is the number of voters in the $x$'th 
  block for whom we shift $p$ by at least $y$ positions.
  More precisely, we define $\sortedshift(E,\vv{\mu})$ to be the election $E'=(C,V')$ identical
  to $E$, except that $p$ has been shifted forward in $\mu_y^x$ voters of block~$x$
  by at least $y$~positions.
  Moreover, a stepwise shift action $\vv{\mu}$ is \emph{valid} if for each
  $x \in [m']$, $y \in [m-2]$ it holds that $\mu_{y+1}^x \le \mu_{y}^x$.
  Informally, a stepwise shift action is valid if one \emph{never} shifts $p$ by at least $y+1$~positions without first having it shifted by at least $y$~positions.
  A stepwise shift action is \emph{successful} if $p$ is an $\calR$-winner
  in $\sortedshift(E,\vv{\mu})$.
  Every valid stepwise shift action can be translated into a shift action.

  %The relation between stepwise shift actions and shift actions is illustrated
  %through an example in \cref{fig:stepwise-shift}.
  
%  Analogously to $\costshift$, 
  % We now define a convenient way of describing stepwise shift actions through
  % appropriate budget partitions.  % $\costsortedshift$.

  The above notion is crucial for understanding our algorithm.
  Thus, let us give a concrete example (see
  \cref{fig:stepwise-shift}). 
  Assume that we have four voters, all in the same
  voter block. The first one has price function $\pi_1(j) = j$, the
  second one has price function $\pi_2(j) = 2^j-1$, the third one
  has price function $\pi_3(j) = 3j$, and
  the fourth one has price function $\pi_4(j) = j^2 + j  + 4$.
  For $\vv{b} = (6,4,9,0)$, $\costsortedshift(\vv{b})$ is a stepwise shift
  action $(3,2,2,0, \ldots)$: We can spend the $6$~units of budget on
  shifting $p$ forward by at least one position and this suffices to
  shift $p$ forward for the first, second, and third voter.
  Note that, the cost is actually $5$ units.
  For these three voters, we have $4$ units of budget to shift $p$ further. This
  is enough to shift $p$ one position up in the preference orders of
  the first and the second voter (the actual cost of this operation is
  $3$).  Finally, for these two voters, we have $9$ units of budget to
  shift $p$ forward by one more position. This costs $3$ for both
  voters and so we, again, shift $p$ forward by one more position in
  the preference orders of the first and second voter. Even though we
  have $6$ units of budget left, we do not use them on anything (in
  particular, we do not try to spend them on the third voter).
  The \emph{stepwise budget distribution} corresponding to the
  stepwise shift action is $(5, 3, 5)$. 
  % 1,2,3     1,1,1
  % 2,4,6     2,2,2
  % 3,6,9     3,3,3

\begin{figure}
  \centering
\begin{tikzpicture}[scale=.5,pile/.style={thick, ->}]
  \path[use as bounding box] (12,-2) rectangle (4,3);
\node at (0,0){
\newcommand{\wc}[1]{\multicolumn{1}{c|}{#1}}
\newcommand{\gc}[1]{\multicolumn{1}{c|}{\cellcolor{darkgray!25}#1}}
\newcommand\Tstrut{\rule{0pt}{2.2ex}}
\tabcolsep=4pt
\begin{tabular}{lccccl}
 \hhline{~|-|-|-|-|~}
 \wc{$v_1$} &\gc{1}&\gc{1}&\gc{1}&\wc{$p$}& 3\Tstrut\\
 \hhline{~|-|-|-|-|~}
 \wc{$v_2$} &\gc{4}&\gc{2}&\gc{1}&\wc{$p$}& 7\Tstrut\\
 \hhline{~|-|-|-|-|~}
 \wc{$v_3$} &\wc{3}&\wc{3}&\gc{3}&\wc{$p$}& 3\Tstrut\\
 \hhline{~|-|-|-|-|~}
 \wc{$v_4$} &\wc{6}&\wc{4}&\wc{6}&\wc{$p$}& 0\Tstrut\\
 \hhline{~|-|-|-|-|~}
            &5&3&5\Tstrut
\end{tabular}};
%\end{tikzpicture}
%\begin{tikzpicture}[scale=.5,pile/.style={thick, ->}]
%  \node at(12, -2) (ssa) {stepwise shift action: $(3, 2, 2)$};
%  \node[above = 0pt of bd] {shift action: $\vv{s}=(3,3,1,0)$};

  \node[text width = 40ex] at (12, -2) (sbd) {stepwise budget distribution: $(5,3,5)$};
  \node[above = 0pt of sbd,text width = 40ex] (bd) {budget
    distribution: $(3,7,3,0)$};
  
  \node[above = 2pt of bd, text width = 40ex] (ssa) {stepwise shift action: $(3, 2, 2)$};
  \node[above = 0pt of ssa, text width = 40ex] {shift action: $(3,3,1,0)$};
 % \node at (3,5.2) {$\vv{s}=(3,3,1,0)$};
  % \node at (3,3.2) {$(3,7,3,0)$};

   % \node at (3,1.2) {$(5,3,5)$};
\end{tikzpicture}
\caption{Illustration of the shift action $(3,3,1,0)$ restricted to a specific voter block,
 where $p$ ranks fourth in the common preference order.
 The price functions are
 $\pi_1(j)=j$, $\pi_2(j)=2^{j}-1$, $\pi_3(j)=3j$, and $\pi_4(j)=j^2 + j + 4$
 with $j<4$.
 The voters are sorted descending with respect to their price functions, that is,
 for any $1 \le j \le 4$ and $1 \le i < i' \le 4$ it holds that $\pi_i(j) \le \pi_{i'}(j)$.
 Each row in the matrix represents one voter.
 The entry $j$~cells left of~$p$ contains the price for moving~$p$ from
 position~$4-j+1$ to position~$4-j$.
 The cell in row~$i$ and column $4-j$ is marked gray
 if we shift~$p$ by $j$~positions in voter~$i$'s preference order.
 Counting the number of marked cells row-wise gives the shift action.
 Summing up the costs row-wise gives the budget distribution.
 Summing up the costs column-wise gives the stepwise budget distribution.
} 
\label{fig:stepwise-shift}
\end{figure}

  Let $\vv{b}=(b_1^1, \dots,b_{m-1}^1, b_1^2, \dots,b_{m-1}^2, \dots, b_1^{m'}, \dots, b_{m-1}^{m'})$ 
  be a vector (intuitively, the entries in this vector should sum up to the amount of money that we wish to spend).
  We define the stepwise shift action $\costsortedshift(\vv{b})=(\costsortedshift(b_1^1)$, $\dots$,
  $\costsortedshift(b_{m-1}^1)$, $\costsortedshift(b_1^2)$, $\dots$, $\costsortedshift(b_{m-1}^2)$,
  $\dots$, $\costsortedshift(b_{1}^{m'})$, $\dots$, $\costsortedshift(b_{m-1}^{m'})$) such that for each $x \in [m']$, $y \in [m-1]$ 
  it holds that $\costsortedshift(b_y^x)$ is the largest number $t_y^x$ such
  that the cost of shifting $p$ in the preference orders of $t_y^x$ voters of block $x$ by $y$ positions (provided that they
  have already been shifted by $y-1$ positions)
  is at most~$b_y^x$ (and $t_{y'+1}^x \le t_{y'}^x$, for all $y' \in [m-2]$).
  Note that the last condition ensures that $\costsortedshift(\vv{b})$ is valid.
  It is not hard to verify that $\costsortedshift(\vv{b})$ can be computed in polynomial time.

\begin{algorithm}[t!]
   \footnotesize
   \SetKwInput{KwParameters}{Parameters}
   \SetKwFunction{RecursiveSearch}{Search}
   \SetKwFunction{Main}{Main}
   \SetKwBlock{Block}
   \SetAlCapFnt{\footnotesize}
   \KwParameters{\\
          $\hspace{3pt}$ $(C,V,\Pi,p,B)$ --- input \ShiftBribery instance\\
          $\hspace{3pt}$ $\calR$ --- voting rule\\
          $\hspace{3pt}$ $m'$ --- the number of voter blocks\\
          $\hspace{3pt}$ $m$ --- the number of candidates\\
          $\hspace{3pt}$ $\varepsilon$ --- approximation parameter}
    \vspace{3mm}	
    \RecursiveSearch{$\vv{q} = (q_1, \ldots, q_{M}), Q, t$}:
        \Block{
          \If{$t = 0$}
          {
            \Return $\{\vv{q}\}$ \;
          }
          $R \leftarrow \{\}$ \;
          \ForEach{$i \in [M]$}
          {
            $\vv{q'} \leftarrow (q_1,\ldots, q_{i-1},q_i+\frac{Q}{M},q_{i+1},\ldots, q_{M})$ \;
            $R \leftarrow R \cup \RecursiveSearch( \vv{q'}, Q - \frac{Q}{M}, t-1 )$ \;
          }
          \Return $R$ \;
        }
    \vspace{3mm}	
              $M \leftarrow m'\cdot (m-1)$\;
              $R \leftarrow \RecursiveSearch( (0, \ldots, 0), B, M\cdot \lceil\ln(\frac{M}{\varepsilon})\rceil+1)$ \;
              \ForEach{ $\vv{b'} \in R$}
              {
                $\vv{b''} \leftarrow \vv{b'} + (\frac{\varepsilon}{M} \cdot B, \ldots, \frac{\varepsilon}{M} \cdot B)$ \;               
                \If{ $p$ is $\calR$-winner in $\sortedshift(E,\costsortedshift(\vv{b''}))$}
                {
                   \Return shift action corresponding to $\costsortedshift(\vv{b''})$\;
                }
              }
           \Return ``no''\;
           \caption{\small $\classFPT$ approximation scheme for \ShiftBribery parameterized by the number of candidates.}
   \label{alg:candidates-approx}
\end{algorithm}

  We are now ready to sketch the idea behind the main component of our
  \ifShowFormalprfFptAsCandidates
  algorithm (we present the formal proof later),
  \else
  algorithm,
  \fi
  namely behind \cref{alg:candidates-approx}.
  Given an $\epsilon>0$ and a budget
  $B$, if it is possible to ensure $p$'s victory by spending at most
  $B$ units of budget, the algorithm computes a shift-action with cost
  at most $\leq (1+\varepsilon) \cdot B$.  The algorithm consists of
  two parts: an $\classFPT$ time exhaustive search part (realized by the
  function \RecursiveSearch), and an approximation and verification part
  (realized by the bottom loop of the algorithm).  
%
%Recall that there are
%  at most $m'm$ shift types. 
  If there is a successful stepwise shift action of cost $B = \OPT(I)$,
  then there is a stepwise budget distribution $\vv{b}$ whose entries sum up to $B$ and such
  that $\costsortedshift(\vv{b})$ is a successful stepwise shift action.  The
  algorithm tries to find this vector $\vv{b}$.  Observe that at least
  one entry in $\vv{b}$ must be greater than or equal to $B/(m'\cdot (m-1))$
  (that is, to the average value of the entries).  \RecursiveSearch
  simply guesses the correct entry and repeats this procedure
  recursively with the remaining budget.  Doing this, it guesses the
  entries to which it allocates smaller and smaller chunks of the
  budget (perhaps adding funds to the entries that already received
  some part of the budget), until it distributes a ``major part'' of
  the budget (never in any entry exceeding the corresponding value
  from $\vv{b}$).  Then, the algorithm turns to the approximation
  part: It simply increases \emph{every} entry in the vector by
  whatever was left of the budget. In effect, we find a vector such that
  every entry matches or exceeds that of $\vv{b}$ but the total
  price is not much greater than $B$. The only problem with this
  approach is that we do not know $\OPT(I)$. However, we will see that
  trying a certain small number of different budget values $B$
  suffices.

  We now move on to the formal proof.  We will first show that for
  each integer $B$, if $I(B)=(C,V,\Pi,p,B)$ is a yes-instance of
  $\calR$ \ShiftBribery with sortable prices, $m$ candidates, and $m'$
  different voter blocks, then \cref{alg:candidates-approx} computes
  in $O^*((m'\cdot m)^{m'\cdot m\cdot \lceil\ln (m'\cdot m/\varepsilon)\rceil+1} \cdot
  w(m))$ time (where $w(m)$ is the exponential part of the running
  time needed for winner determination) a successful shift action
  $\vv{s}$ for $I$ such that
  $\Pi(\vv{s}) \leq (1+\varepsilon) B$.
%  $\Pi(\vv{s}) \leq (1+2\varepsilon +
  % \varepsilon^2)B$.

  Let $M = m'\cdot (m-1)$.
  It is easy to verify the running time of \cref{alg:candidates-approx} which comes
  from the exhaustive search of depth $M\cdot \lceil\ln(M/\varepsilon)\rceil+1$
  in \RecursiveSearch.
  To see that it indeed gives the desired answer, assume  that $I(B)$ is a
  yes-instance and let $\vv{b}=(b_1^1,\dots,b_{m-1}^1,b_1^2,\dots,b_{m-1}^2,\dots,b_{m-1}^{m'})$
  be a vector such that $\sum_{x \in [m'], y \in [m-1]}b_y^x = B$,
  and $\costsortedshift(b)$ is a successful stepwise shift action. 
  We observe that, by the pigeonhole principle, there are
  $x \in [m']$, $y \in [m-1]$, such that $b_y^x \geq B/M$.
  \RecursiveSearch tries each choice of $x$ and $y$ and then repeats this process
  recursively for $M\cdot \lceil\ln(M/\varepsilon)\rceil+1$ steps.
  This means that there is a vector
  $\vv{{b'}}=({b'}_1^1,\dots,{b'}_{m-1}^1,{b'}_1^2,\dots,{b'}_{m-1}^2,\dots,{b'}_1^{m'},\dots,{b'}_{m-1}^{m'})$
  among those in $R$ such that:
  \begin{enumerate}[(a)]
   \item for each $x \in [m']$, $y \in [m-1]$ it holds that
         ${b'}_y^x \leq b_y^x$, and
   \item $B - \sum_{x \in [m'], y \in [m-1]}{b'}_y^x <
         B(1-\frac{1}{M})^{M\cdot \ln\frac{M}{\varepsilon}} < B \cdot 
         e^{-\ln\frac{M}{\varepsilon}} = \frac{\varepsilon}{M}\cdot B$.
  \end{enumerate}
  %   (a) for each $x \in [m']$, $y \in [m]$ it holds
  %   ${b'}_y^x \leq b_y^x$, and
  %   (b) $B - \sum_{x \in [m'], y \in [m]}{b'}_y^x <
  %   B(1-1/M)^{M\ln (M/\varepsilon)} < B
  %   e^{-\ln (M/\varepsilon)} = \varepsilon/M B$.
  Thus, for each $x \in [m']$, $y \in [m-1]$ it holds that ${b''}_y^x =
  {b'}_y^x + B \cdot \varepsilon/M \geq b_y^x$.
  Recall that $\vv{b''}$ is a vector computed by \RecursiveSearch.
  Now, $\costsortedshift(\vv{{b''}})$ is a successful stepwise shift action.
  Even ensuring that $\costsortedshift(\vv{{b''}})$ is valid will
  not decrease any component below the corresponding value in $\costsortedshift(\vv{b})$.
%  because we gave the total difference of all components to each single
%  component as extra budget.
%  Naturally, 
  The cost of $\costsortedshift(\vv{{b''}})$ is at most $(1+\varepsilon)\cdot B$.

  Now we can compute an approximate solution for $I$ as follows.
  We  start by considering shift action $\costshift(0,\ldots,0)$ and
  output it if it is successful.
  Otherwise, we run \cref{alg:candidates-approx} for each budget
  $(1+\varepsilon)^i$, for $i = 0, 1, 2, \ldots$ and output the first
  successful shift action that it outputs.

  Note that if the cost of some cheapest successful shift action for $I$ is
  $B$, then after $O(\ln B)$ iterations the procedure above will find a
  budget $B'$ such that $B \leq B' < (1+\varepsilon)\cdot B$. For this $B'$,
  \cref{alg:candidates-approx} will output a successful
  shift action with cost at most $(1+\varepsilon)\cdot B' \leq
  (1+\varepsilon)^2\cdot B = (1+2\varepsilon + \varepsilon^2)\cdot B$. This shows
  that our algorithm is indeed an approximation scheme. It remains to
  show that it runs in $\classFPT$ time (parameterized by the number of
  candidates). To see this, it suffices to note that $O(\ln B) = O(\ln
  (\Pi(m, \ldots, m))) = O(|I|)$ and that
  \cref{alg:candidates-approx} runs in $\classFPT$ time.
\end{proof}

Since the publication of our original
findings~\cite{BreCheFalNicNie2014}, \citet{BreFalNieTal2015b} have
improved upon our $\classFPTAS$ result stated in
\cref{thm:fpt-as-num-candidates}, using mixed integer linear programs
to show membership in $\classFPT$. 
%  their result
% also applies to \ShiftBriberyO because the optimization variant of
% %\todo{What should probILPshort expand to?}
% \textsc{ILP}probILPshort  is also fixed-parameter
% tractable.  
Their approach relies on the fact that the fixed-parameter
tractability result of \probILPshort parameterized by the number of
variables~\cite{Len83} also holds for mixed integer linear programs,
where the number of integer-valued variables is upper-bounded by a function in the
parameter but the number of \emph{real-valued} variables is not
restricted (i.e., the number of real-valued variables is polynomial in
the size of the input). Nonetheless, they also use a number of
insights from the above algorithm.

% (adapted to our setting) does not only hold for
% integer linear programs with $g_1(m)$ integer-valued variables and
% $g_2(m)$ constraints, but also holds when the program has polynomially
% many (in the input size) additional \emph{real-valued} variables.

Adopting some of the ideas of the above proof, we obtain the following
$\classXP$ algorithm for \ShiftBribery with arbitrary price functions.
The corresponding algorithm relies on the fact that
for a set of voters with the same preference orders,
finding a subset of $q$~voters who shift $p$ higher by exactly the same number~$j$ of positions
and who have the minimum price
is polynomial-time solvable.

\newcommand{\propxpnumcands}{
  Let $\calR$ be a voting rule whose winner determination procedure is
  in $\classXP$ when parameterized by the number~$m$ of candidates. $\calR$
  \ShiftBribery parameterized by~$m$ is in $\classXP$
  for each price function family that we consider.  The algorithm runs in
  time $O^*((n^m)^{m!})$ times the cost of $\calR$'s winner
  determination.
}
\begin{theorem}\label{thm:xp-num_cands-0/1-prices}
  \propxpnumcands
\end{theorem}

\begin{proof}
  Let $\calR$ be a voting rule whose winner determination procedure is
  in $\classXP$ when parameterized by the number of candidates.  We give an
  $\classXP$ algorithm for $\calR$ \ShiftBribery.

  Let $I = (C,V,\Pi,p)$ be an input instance of $\calR$
  \ShiftBribery. Let $n$ be the number of voters and let $m$ be the
  number of candidates in $E = (C,V)$. Each voter has one of $m!$
  possible preference orders and we partition $V$ into vectors $V^1,
  \ldots, V^{m!}$ of voters with the same preference orders.  Our
  algorithm works as follows:
  \begin{enumerate}
  \item For each $\ell$, $1 \leq \ell \leq m!$, guess a vector
    $\vv{q^\ell} = (q^\ell_0, \ldots, q^\ell_m)$ of nonnegative
    integers that sum up to $\left|V^\ell\right|$. For each guessed list
    of vectors do the following:
    \begin{enumerate}
    \item Compute the cheapest shift action $\vv{s}$ such that for
      each $\ell$, $1 \leq \ell \leq m!$, and for each $j$, $0 \leq j
      \leq m$, there are exactly $q^\ell_j$ voters in $V^\ell$ for
      whom we shift $p$ by exactly $j$ positions. If such a
      shift action does not exist then move on to another guess in the
      preceding step of the algorithm. (Below we describe precisely
      how to implement this step.)
    \item If $p$ is an $\calR$-winner of $\shift(E,\vv{s})$ then store
      $\vv{s}$ as a successful shift action.
    \end{enumerate}
  \item Return the cheapest stored successful shift action.
  \end{enumerate}

  Since, in essence, the algorithm tries all interesting
  shift actions, it is clear that it is correct. We now show that it
  is possible to implement it to run in $\classXP$ time with respect to the
  number of voters. 

  In the first step, the algorithm has $O((n^m)^{m!})$ possible guesses so we
  can try them all. However, it is not clear how to implement given vectors
  $\vv{q^1}, \ldots, \vv{q^{m!}}$ in step (1a). Let $\ell$ be some integer,
  $1 \leq \ell \leq m!$. We will show how to compute the part of $\vv{s}$
  that implements $\vv{q^\ell}$.

  We solve the problem by reducing it to an instance of a variant of
  \textsc{Min Cost/Max Flow} problem where, in addition to the
  standard formulation, we are also allowed to put lower bounds on the
  amount of flow traveling through some edges (we call these values
  \emph{demands}; this problem is well-known to be polynomial-time
  solvable~\cite{ahu-mag-orl:b:flows}).

  Let us rename the voters in $V^\ell$ so that $V^\ell = (v_1, \ldots,
  v_r)$ for some $r \in \naturals$.  We form the following flow
  network for \textsc{Min Cost/Max Flow}. We let the set $N$ of vertices
  be $\{s,t\} \cup \{v_1, \ldots, v_r\} \cup \{0, \ldots, m\}$. Vertices
  $s$ and $t$ are, respectively, the source and the sink. Vertices
  $v_1, \ldots, v_r$ form the first layer in the flow network and
  correspond to the voters in $V^\ell$. Vertices $0, \ldots, m$ form
  the second layer in the network and correspond to the extent to
  which $p$ should be shifted (in a given group of voters). We have the
  following edges:
  \begin{enumerate}
  \item For each $i$, $1 \leq i \leq r$, there is an edge with unit
    capacity and zero cost from $s$ to $v_i$.
  \item For each $i$, $1 \leq i \leq r$ and each $j$, $0 \leq j \leq
    m$, there is an edge with unit capacity and cost $\pi_{i}(j)$ from
    $v_i$ to $j$ ($\pi_i$ is $v_i$'s price function).
  \item For each $j$, $0 \leq j \leq m$, there is an edge from $j$ to
    $t$ with capacity and demand $q^\ell_j$ and zero cost (the setting
    of capacity and demand means that we are interested in flows where
    there are exactly $q^\ell_j$ units of flow traveling from $j$ to
    $t$).
  \end{enumerate}
  There is a polynomial-time algorithm that finds a lowest-cost flow
  from $s$ to $t$ that satisfies the capacity and demand constraints
  of all the edges~\cite{ahu-mag-orl:b:flows}. The constructed flow
  network is of polynomial size with respect to our input instance, so
  our algorithm indeed runs in $\classXP$ time.  Further, the flow from
  vertices $\{v_1, \ldots, v_r\}$ to vertices $\{0, \ldots, j\}$ defines the
  desired shift action in a natural way: If a unit of flow travels
  from some $v_i$, $1 \leq i \leq r$, to some~$j$, $0 \leq j \leq m$,
  then we shift $p$ by $j$ positions in $v_i$. This completes the
  argument.
\end{proof}

\section{Conclusion}\label{sec:conclusions}
We have studied the parameterized complexity of \ShiftBribery under
the voting rules Borda, Maximin, and Copeland, for several natural
parameters (either describing the nature of the solution or the size
of the election) and for several families of price functions
(arbitrary, convex, unit, sortable, and all-or-nothing). Our results
confirmed the belief that the computational 
complexity depends on all three factors:
the voting rule, the parameter, and the type of price function used.

Our work leads to some natural follow-up questions. First, it would be
interesting to resolve the complexity of \ShiftBribery parameterized
by the number of candidates for arbitrary price functions (see \Cref{tab:results}); 
at the moment we only know that it is $\wone$-hard and in $\classXP$.
For both Borda and Maximin,
and for the parameter ``number~$t$ of unit shifts'',
we obtain a partial problem kernel, that is, using polynomial time, 
we can shrink a given instance to one whose numbers of alternatives and voters are upper-bounded by a function in~$t$~(\cref{thm:Borda-partialkernel-num-shifts,obs:Maximin-partialkernel-num-shifts}).
This function, however, is an exponential function.
Thus, it would be interesting to know whether it is possible to design an efficient kernelization algorithm~\cite{BGKN11}
that replaces this exponential function with a polynomial.
A similar concept to the kernelization lower bound
would be to show a lower bound on the size of \emph{only} part of the problem kernel

Further, one could study the
problem for other voting rules. Going in the direction of the
\textsc{Margin of Victory} type of problems, it would be interesting
to study \textsc{Destructive} \ShiftBribery, where we can push back
the despised candidate to prevent him or her from being a winner
(initial results in this direction are due to
\citet{kac-fal:c:destructive-shift-bribery}). 

Finally, recently \citet{BulCheFalNieTal2015} introduced
a model of affecting election outcomes by adding groups of voters to
the election (this is an extension of a well-studied problem known as
\textsc{Control by Adding
  Voters}~\cite{bar-tov-tri:j:control,hem-hem-rot:j:destructive-control}).
\citet{BreFalNieTal2015a} considered \ShiftBribery in a similar
setting where one can affect multiple voters at the same time (for
example, airing an advertisement on TV could affect several voters at
the same time). Unfortunately, most of their results are quite
negative and it would be interesting to investigate the combinatorial variant of \ShiftBribery from
the point of view of heuristic algorithms.

\section*{Acknowledgments}
We are grateful to the anonymous referees of AAAI~2014 and of \emph{Information and Computation} that helped us to significantly improve the paper.

Robert Bredereck was supported by the DFG project PAWS (NI 369/10).
Jiehua Chen was partially supported by the Studienstiftung des
Deutschen Volkes.  Piotr Faliszewski was partially supported by the
DFG project PAWS (NI 369/10) while staying at TU~Berlin, and by AGH
University grant 11.11.230.124 (statutory research) during the rest of
the project.

% \section{Bibliography styles}

% There are various bibliography styles available. You can select the style of your choice in the preamble of this document. These styles are Elsevier styles based on standard styles like Harvard and Vancouver. Please use Bib\TeX\ to generate your bibliography and include DOIs whenever available.

% Here are two sample references: \cite{Feynman1963118,Dirac1953888}.
%\bibliographystyle{abbrvnat}

%\bibliography{grypiotr2006}

\appendix
\section{Proof of \texorpdfstring{\cref{thm:w[2]-h-num_voters_affected}}{Theorem~\ref{thm:w[2]-h-num_voters_affected}}}\label{sec:omitted}

In this appendix we present the omitted proof of \cref{thm:w[2]-h-num_voters_affected}. % from the main text. 
To this end, we split the proof in three propositions, each proposition covering one of the considered voting rules.
In most cases the proofs of the propositions base on an idea already present in the literature (mostly in the paper of Elkind et al.~\cite{elk-fal-sli:c:swap-bribery}), but they do include some additional insights.

% \newtheorem*{corborwtwo}{Corollary~\ref{cor:Borda-w[2]-h-num_voters_affected}}
% 
% \begin{corborwtwo}
%   \corBordawtwonumberaffectedvoters
% \end{corborwtwo}

\begin{proposition}\label{prop:Borda-w[2]-h-num_voters_affected}
	\corBordawtwonumberaffectedvoters
\end{proposition}

\begin{proof}
  % We consider the voting rules mentioned in the corollary one by one.
  % We note that we view this result as a corollary because the main
  % ideas of our proofs were already present in the literature (see the
  % work of Elkind et al.~\shortcite{elk-fal-sli:c:swap-bribery}). Our
  % contribution is in providing technical details to obtain the exact
  % form of our theorem (in particular, we need to ensure the proofs
  % work for unit prices, whereas original constructions strongly relied
  % on using all-or-nothing prices).\medskip
  We consider \probShiftBBorda for unit prices first. 
  We give a parameterized reduction from \probSC parameterized by the set cover size~$k$. 
  Let $I = (\setUniverse,\setFamily,k)$ be our input
  instance, where $\setUniverse = \{u_1, \ldots, u_n\}$ and $\setFamily = \{S_1,
  \ldots, S_m\}$ is a collection of subsets of $\setUniverse$, and $k$ is a positive
  integer. 
	We set the budget $B = k\cdot (n+1)$ and form a set of candidates $C = \{p,d,g\} \cup \setUniverse \cup F$, where $F$ is a
  set of ``filler'' candidates, $F = \{f_1, \ldots, f_{2k\cdot (n+1)+2}\}$.
  We form a set~$V$ of voters as follows:
  \begin{enumerate}[i)]
  \item We form a subset~$V_{\sets}$ of voters as follows (note that this
    subset will be useful in the proof for the case of Copeland). 
    For each set $S_i$ in $\setFamily$ we create two voters, $v_{2i-1}$
    and $v_{2i}$.  Voter $v_{2i-1}$ has preference order
    \[
    d \pref \seq{S_i} \pref \seq{F_i} \pref p \pref \seq{C \setminus (S_i \cup F_i \cup
    \{p,d\})},
    \]
    where $F_i$ is an arbitrarily chosen set of $n-|S_i|$ candidates
    from $F$.  
	Voter~$v_{2i}$ has the reversed preference order of~$v_{2i-1}$. 
	Note that each pair of voters $v_{2i-1}$, $v_{2i}$, gives each candidate $|C|-1$ points.

  \item % Let $c$ be a candidate in $C \setminus \{p, g, d\}$. We define
  % $\point(c)$ to be two voters with the following preference orders:
  % the first voter has preference order $p \pref c \pref g \pref C
  % \setminus \{p,d,c,g\} \pref d$ and we obtain the second voter's
  % preference order from reversing the one of the first voter and then
  % swapping the positions of $c$ and $g$ (so that in the second voter's
  % preference order $c$ receives $2$ points and $g$ receives $1$
  % point). Note that, together, the two voters give $|C|$ points to $c$,
  % $|C|-2$ points to $g$, and $|C|-1$ points to each candidate in $C -
  % \{c,g\}$.

  For each $i$, $1 \leq i \leq n$, we add $B+1$ pairs of voters
  with preference orders from~$\point(p,\emptyset, u_i, g, C\setminus \{u_i,g,p,d\}, d)$ (see \cref{subsec:conventions} for the definition) to the voter set~$V$.
  Note that, with two preference orders from $\point(p,\emptyset, u_i, g, C\setminus \{u_i,g,p,d\}, d)$, 
  $u_i$ receives $|C|$~points, $g$~receives~$|C|-2$ points, and each remaining candidate in $C\setminus\{u_i,g\}$ receives $|C|-1$ points.
  In both preference orders, there are $|C|-2$ candidates between $p$ and $d$.

  \item %Let $f$ be some arbitrary candidate in $F$. We define $\point(d)$ to be two voters
    %that are obtained from the voters $\point(f)$ by swapping the positions of $f$ and $d$
    %in both of them.
    Let $F'$ be an arbitrary size-$(B+1)$ subset of $F$ 
    and let $z\in F\setminus F'$ be an arbitrary candidate.       
    We add $B+k$ pairs of voters with preference orders of
    $\point(p,F',d,g,C\setminus (\{p,g,d,z\}\cup F'), z)$ to the voter set~$V$.

    % We define
    % $\point(d)$ to be two voters with the following preference orders:
    % the first voter has preference order $p \pref F' \pref d \pref g \pref C
    % \setminus (\{p,d,g\} \cup F')$ where $F'$ is an arbitrary size-$(k(n+1)+1)$ subset of $F$
    % and we obtain the second voter's
    % preference order from reversing the one of the first voter and then
    % swapping the positions of $d$ and $g$.
    % Note that together, voters $\point(d)$ give $|C|$ points to $d$, 
    % $|C|-2$ points to $g$, and $|C|-1$ points to all candidates in $C - \{d,g\}$.
  \end{enumerate}

  Since we are dealing with unit price functions, each voter has
  price function $\pi(j) = j$.
  By routine calculation, it is easy to verify that there is a value $L$ such that
  the candidates in election $E=(C,V)$ have the following scores:
  \begin{itemize}
    \item $\s{E}{p} = L$,
    \item $\score_E(d) = L + B+k$,
    \item for each $i$, $1 \leq i \leq n$, $\score_E(u_i) = L + B+1$, and
    \item for each $f \in F \cup \{g\}$, $\score_E(f) \leq L$.
  \end{itemize}
	We set the maximal number~$n_a$ of affected voters to~$k$.
	We claim that $I$ is a yes-instance of $\probSC$ if and only if the formed instance of \probShiftBBorda is a yes-instance. 

  ``$\Rightarrow$'':
  Assume that there is a collection~$\setcover$ of $k$ sets from $\setFamily$
  such that their union gives $\setUniverse$. Without loss generality, let $\setcover\coloneqq\{S_{1}, \ldots, S_{k}\}$.  
  Bribing the voters $v_{2i-1}, \ldots,
  v_{2k-1}$ to shift $p$ to the top position has price $k\cdot (n+1) = B$
  and ensures $p$'s victory: Such shifts ensure that $p$ gains
  $k\cdot (n+1)$ points, each $u_i \in \setUniverse$ loses at least one point (because
  $S_{1}, \ldots S_{k}$ is a cover of $\setUniverse$), and $d$ loses $k$
  points (because $p$ passes $d$ in each of these $k$ voters' preference orders). In
  effect, $p$, $d$, and each $u_i \in \setUniverse$, have $L+k\cdot (n+1)$ points each
  and tie for victory.

  ``$\Leftarrow$'':
  For the reverse direction, assume that there is a successful shift
  action $\vv{s}$ that ensures $p$'s victory and affects $n_a = k$ voters
  only.  Since the budget is $k\cdot (n+1)$ and we use unit price functions,
  the score of $p$ can increase to at most $L + k\cdot (n+1)$. Since $d$ has
  score $L + k\cdot (n+1)+k$ and we can affect $k$ voters only, for each
  voter where we shift $p$, $p$ must pass $d$.

  Thus $\vv{s}$ involves only the voters from the set $V' = \{v_{2i+1}
  \mid 1 \leq i \leq m\}$. The reason for this is that each voter in
  $V \setminus V'$ either already prefers $p$ to~$d$ or ranks $d$ more
  than $B$ positions ahead of $p$ (so $p$ would not be able to pass~$d$ within the budget). Consequently, in the preference order of
  each voter from $V'$ affected by
  $\vv{s}$, $p$ must be shifted to the top position (because in these
  preference orders, $d$ is ranked first).

  Further, since the score of $p$ can increase to at most $L+k\cdot (n+1)$
  and each candidate $u_i \in \setUniverse$, prior the shifts, has score $L +
  k\cdot (n+1)+1$, it must be the case that $p$ passes each $u_i$ at least
  once. This means that the voters affected under $\vv{s}$ correspond
  to a cover of $\setUniverse$ by at most $k$ sets from $\setFamily$. 

  To complete the proof, it is easy to verify that the reduction works
  in polynomial time.

  The above argument applies to unit prices but, by
  \cref{prop:inclusion-price-functions}, it also
  immediately covers the case of convex price functions, sortable
  price functions, and of arbitrary price functions.  To deal with
  all-or-nothing price functions, % by
  % Proposition~\ref{prop:inclusion-price-functions}, it suffices to
  % consider all-or-nothing prices. To do so,
  it suffices to use the
  same construction as above and set the budget~$B$ to $k$,
  but with the following prices:
  \begin{enumerate}[i)]
  \item For each $i$, $1 \leq i \leq n$, we set the price function of
    voter $v_{2i+1}$ to be $\pi(j) = 1$ for each $j > 0$ (and, of
    course, $\pi(0)=0$).
  \item  For all the remaining voters, we use price
    function $\pi'$, such that $\pi'(j) = B+1$ for each $j > 0$ (and
    $\pi'(0)=0$).
  \end{enumerate}
  It is easy to see that the construction remains correct with these
  price functions.
\end{proof}

% \newtheorem*{cormaxwtwo}{Corollary~\ref{cor:Maximin-w[2]-h-num_voters_affected}}
% 
% \begin{cormaxwtwo}
%   \corMaximinwtwonumberaffectedvoters
% \end{cormaxwtwo}

\begin{proposition}\label{prop:Maximin-w[2]-h-num_voters_affected}
	\corMaximinwtwonumberaffectedvoters
\end{proposition}

\begin{proof}
  We reduce from \probSC parameterized by the set cover size to
  \probShiftBMaximin parameterized by the number of affected voters.  
  Let $I = (\setUniverse,\setFamily,k)$ be our input instance
  of \probSC, where $\setUniverse = \{u_1, \ldots, u_n\}$, $\setFamily = \{S_1,
  \ldots, S_m\}$ is a collection of subsets of $\setUniverse$, and $k$ is a
  positive integer (and the value of the parameter). We form a
  candidate set $C = \{p,d,g\} \cup \setUniverse \cup F$, where 
  $d$ is the unique winner in the original election and $F = \{f_1,
  \ldots, f_{2n\cdot k+2k}\}$ is a set of ``filler'' candidates (we
  partition $F$ into two disjoint sets~$X_1$ and $X_2$ of cardinality
  $n\cdot k+k$ each). We
  form a set~$V$ of voters as follows: % (note that when we list
  % sets in voters' preference orders, we mean listing the candidates
  % from these sets in lexicographic order of their names):
  \begin{enumerate}[i)]
  \item We create a set~$V_\sets$ of voters as follows.  For each
    $S_i \in \setFamily$, we include in $V_\sets$ two voters, $v_{2i-1}$ and
    $v_{2i}$, such that $v_{2i-1}$ has preference order 
    \begin{align*}
      g \pref \seq{S_i} \pref \seq{F_i} \pref p \pref \seq{F \setminus F_i} \pref \seq{\setUniverse \setminus S_i},
    \pref d
  \end{align*}
  where each $F_i$ contains an arbitrarily fixed subset of
  $n-|S_i|$ candidates from $F$, 
  and $v_{2i}$ has preference order which is the revers of $v_{2i-1}$'s.

  \item We create a set~$V_\structure$ of voters by including the
    following group of voters:
    \begin{enumerate}[a)]
    \item we include $k$ voters with preference order
      \begin{align*}
        p \pref d \pref g \pref \seq{\setUniverse} \pref \seq{X_1} \pref \seq{X_2}\text{,}
      \end{align*}
    \item we include $k-1$ voters with the same preference order 
      \begin{align*}
        g \pref \seq{X_1} \pref p \pref \seq{X_2} \pref d \pref \seq{\setUniverse}\text{,}
      \end{align*}
    \item we include a single voter with preference order 
      \begin{align*}
        g \pref \seq{\setUniverse} \pref \seq{X_1} \pref p \pref \seq{X_2} \pref d\text{, and}
      \end{align*}
    \item we include $2k$ voters with preference order 
      \begin{align*}
        d \pref g \pref \seq{\setUniverse} \pref \seq{X_2} \pref p \pref \seq{X_1}\text{.}
      \end{align*}
    \end{enumerate}
  \end{enumerate}
  We set $V = V_\sets \cup V_\structure$, and $E = (C,V)$. Each voter has a unit price
  function and we set the budget~$B$ to be $k \cdot (n+1)$. We claim that the
  constructed \probShiftBMaximin instance is a yes-instance if and
  only if $I$ is a yes-instance of \probSC. However, before we prove
  this fact, let us calculate the scores of the candidates in the election~$E$.

  \begin{table}
  \begin{center}
  \tabcolsep=0.12cm
  \begin{tabular}{c|ccccc}
        & $p$     & $g$    & $d$    & $\setUniverse$      & $F$ \\[+.3em]
   \hline
   $p$  & --      & $m+k$  & $m+2k$ & $m+2k-1$ & $\ge m+2k$   \\[+.3em]
   $g$  & $m+3k$  &  --    & $m+k$  & $m+4k$   & $m+4k$ \\[+.3em]
   $d$  & $m+2k$  & $m+3k$ & --     & $m+4k-1$ & $m+3k$ \\[+.3em]
   $\setUniverse$  & $m+2k+1$& $m$    & $m+1$  & $\leq m+4k$ & $m+3k+1$ \\[+.3em]
   $F$  & $\le m+2k$  & $m$    & $m+k$  & $m+k-1$ & $\leq m+4k$
  \end{tabular}
  \end{center}
  \caption{\label{tab:maximin}Table of the values of the $N_E(\cdot,\cdot)$ function for
    the election constructed for the Maximin part of the proof of 
    \cref{prop:Maximin-w[2]-h-num_voters_affected}.}
  \end{table}

  In \cref{tab:maximin} we give the values of the function
  $N_E(\cdot,\cdot)$ for the just constructed election (for the values
  of the function among candidates from the set $\setUniverse \cup F$ we use the
  trivial upper bound, $m+4k$). Note that for each $x,y \in C$, the
  voters from $V_{\sets}$ contribute value $m$ to $N_E(x,y)$ and the
  remaining value comes from the voters in $V_2$. Based on the function
  $N_E$, we calculate the scores of the candidates in our election:
  \begin{itemize}
  \item $\s{E}{p} = m+k$, 
  \item $\score_E(g) = m+k$,
  \item $\s{E}{d} = m+2k$,
  \item for each $u_i \in \setUniverse$, $\score_E(u_i) = m$, and
  \item for each $f_j \in F$, $\score_E(f_j) = m$.
  \end{itemize}
  Thus, prior to any bribery, $d$ is the unique winner.

  We show that $(\setUniverse, \setFamily, k)$ is a yes-instance for \probSC if and only if the constructed \probShiftBMaximin instance is a yes-instance.

  ``$\Rightarrow$'': 
  Suppose without loss of generality that $S_1, \ldots, S_k$ cover the universe~$\setUniverse$. 
  Then, 
  % If there exists a collection $\setcover\coloneqq\{S_{i_1}, \ldots, S_{i_{k}}\}$, of $k$
  % sets from $\setFamily$, whose union is $\setUniverse$, then our 
  % \probShiftBMaximin instance is a yes-instance. It
  it suffices for each voter
  in the set $V' = \{v_{2i-1} \mid 1 \leq i\leq k\}$ to shift $p$
  to the top position. This leads to $p$ passing $g$ for $k$ times and
  to $p$ passing each candidate from $\setUniverse$ at least once. 
  In effect, the score of $p$ increases to $m+2k$, 
  and $p$ and $d$ are tied winners.
  Further, doing so costs at most $n\cdot k+k$: 
  In each of the at most $k$ preference orders where we shift~$p$, we shift $p$ by $n+1$ positions.

  ``$\Leftarrow$'': 
  Assume that there is a shift-action $\vv{s}$ such
  that applying $\vv{s}$ ensures that $p$ is a (co)winner. We can assume
  that $\vv{s}$ shifts $p$ only in some preference orders of the
  voters in the set $V' = \{v_{2i-1} \mid 1 \leq i \leq n\}$. This is
  so because in all the other preference orders, $p$ is ranked just below a group
  of at least $B$ candidates from $F$ and for each candidate $f_i \in
  F$, we have $N_E(p,f_i) = m + 2k$; $f_i$ is not blocking $p$ from
  gaining additional $k$ points and, since we can affect at most $k$
  voters, $p$ can get at most $k$ additional points. Further, voters
  in the set $V'$ rank $p$ ahead of $d$, and so by applying~$\vv{s}$
  we certainly cannot lower the score of $d$. This means that $\vv{s}$
  must ensure that $p$'s score increases to at least $m+2k$. This is
  possible only if $\vv{s}$ affects exactly $k$ voters, in each vote
  affected by $\vv{s}$ candidate $p$ passes $g$, and $p$~passes each
  candidate from $\setUniverse$ at least once. This means that $\vv{s}$ shifts
  $p$ to the first position in the preference orders of $k$ voters from $V'$, this costs exactly
  $B$ (shifting $p$ to the top position in the preference order of a single voter from $V'$ has
  price $n+1$), and the voters for which $\vv{s}$ shifts $p$
  correspond to the sets from $\setFamily$ that cover $\setUniverse$.

  The above construction covers the case of unit price functions.  By
  \cref{prop:inclusion-price-functions}, this
  also applies to convex price functions, sortable price functions,
  and arbitrary price functions. For the all-or-nothing price
  functions, we can adapt the above proof in the same way in which we
  have adapted the proof for the case of Borda, see proof of \cref{prop:Borda-w[2]-h-num_voters_affected}.
\end{proof}

% \newtheorem*{corcopwtwo}{Corollary~\ref{cor:Copeland-w[2]-h-num_voters_affected}}
% 
% \begin{corcopwtwo}
%   \corCopelandwtwonumberaffectedvoters
% \end{corcopwtwo}

\begin{proposition}
	\corCopelandwtwonumberaffectedvoters
\end{proposition}

\begin{proof}
  We reduce from \probSC parameterized by the set cover size. 
  Our reduction is similar to that
  for the case of Borda and, in particular, we use the same notation for
  the input instance $I$ of
  \probSC, %(however, without loss of generality, we assume $k > 2$),
  we form the same candidate set $C=\{p,d,g\}\cup \setUniverse$ (see the proof for \cref{prop:Borda-w[2]-h-num_voters_affected}), 
  and we use the subset
  $V_{\sets}$ of voters. However, we extend $V_{\sets}$ to a complete set of
  voters in a different way, described below.  Let $r = 2k\cdot (n+1)+3$ be
  the number of candidates in the set $F'\coloneqq F \cup \{g\}$.
  \begin{enumerate}[i)]
  \item We introduce a single voter $p \pref d \pref \seq{U} \pref \seq{F'}$ (the reason for doing so is to have an odd number of voters, so the proof
    will work for all values of $\alpha$).
  \item Using a modified McGarvey's construction (see below), we
    introduce voters to ensure the following results of head-to-head
    contests among the candidates:
    \begin{enumerate}[a)]
    \item Candidate $d$ defeats each candidate in $\setUniverse$ by one preference order;
      $d$~defeats $p$ by $2k-1$ preference orders, 
      and $d$ defeats by one preference order $r-n$
      (arbitrarily chosen) candidates from $F'$ (the
      remaining candidates from $F'$ defeat $d$ by one preference order).
      
    \item Candidate $p$ defeats by one preference order $r-n$ arbitrarily chosen candidates
      from $F'$ and loses to the remaining $n$ of them by $2k+1$
      preference orders. Candidate $p$ loses by one preference order to each candidate from $\setUniverse$.

    \item For each $i$, $1 \leq i \leq n$, $u_i$ defeats by one preference order
      all candidates $u_j$ such that $i > j$. 
      Candidate $u_i$ defeats by one preference order % candidate~$p$ and
      $r-(i-1)$ arbitrarily
      chosen candidates from $F'$, 
      and loses to the remaining $i-1$ candidates from $F'$.

    \item Each candidate in $F'$ defeats by one preference order at most $\lfloor
      |F'|/2\rfloor$ candidates from $F'$ (an easy way of
      constructing a set of results of head-to-head contests that
      achieves this effect and a proof that this way we can set the
      results of all head-to-head contests among the candidates from
      $F'$ is given in the work of Faliszewski et
      al.~\citep[discussion after their Lemma~2.3]{fal-hem-sch:c:copeland-ties-matter}).
    \end{enumerate}
  \end{enumerate}

  Before we explain what we mean by modified McGarvey's construction,
  let us first calculate the scores that the candidates would have if
  we formed election $E = (C,V)$ according to the above description.
  We would have:
  \begin{itemize}
   \item $\s{E}{p} = r-n$, 
   \item $\score_E(d) = n + r-n +1 = r+1$,
   \item each candidate $u_i$ in $\setUniverse$ would have $\score_E(u_i) = i-1 + r - (i-1) +1 = r+1$, and
   \item each candidate $f$ in $F$ would have $\score_E(f) \leq 1+1+n + r/2 \leq r$ (because $k > 2$).
  \end{itemize}

  Let us now describe what we mean by modified McGarvey's
  construction.  
	McGarvey's theorem~\cite{mcg:j:election-graph} is as follows. 
	Assume that we are given given a set of candidates % 	(and, even, possibly some preference orders over this candidate set) 
	and a set of results of head-to-head contests (specified for each pair $x, y$ of candidates by the number $M(x,y)$ of voters that prefer $x$ over $y$ minus the number of voters that prefer $y$ over $x$; either all values $M(x,y)$ are odd or all are even).
	Then, we can compute in polynomial time (with respect to the number of candidates, the number of given voters, and the sum of the values $M(x,y)$) an election $E$ % (that includes the already given preference orders) 
	such that for each two candidates $x,y$ we have $N_E(x,y)-N_E(y,x) = M(x,y)$.

  It is standard to prove McGarvey's theorem using a construction that
  introduces pairs of voters (for candidate set $C$) with preference
  orders of the form $x \pref y \pref \seq{C - \{x,y\}}$, and
  $\revseq{C-\{x,y\}} \pref x \pref y$. % (where in the latter
  % preference order we rank all the candidates from $C-\{c,d\}$ in
  % reverse order compared to the former preference order.
  We use this
  construction with the modification that we ensure that in each preference order
  $p$ and $d$ are always separated from each other by at least
  $k\cdot (n+1)$ candidates. (For example, it suffices to ensure that $p$
  and $d$ are always ranked first, second, last, or second to last,
  and that if $p$ is ranked first or second then $d$ is ranked last or
  second to last, and the other way round).

  We can now continue our proof. We set budget $B \coloneqq k\cdot (n+1)$ and we use
  unit price functions for all the voters. We claim that there is a
  cover of~$\setUniverse$ by at most $k$ sets from $\setFamily$ if and only if our instance of
  \probShiftBCopeland is a yes-instance. 

  ``$\Rightarrow$'':
  Assume without loss of generality that $S_1, \ldots, S_k$ cover the universe~$\setUniverse$.
  Indeed, if we shift $p$ to the top position in those preference orders from $\{v_{2i-1}\mid 1\le i \le k\}$, 
  then $p$ will pass $d$ $k$ times (in effect winning the head-to-head contest against $d$), 
  $p$ will pass each candidate from $\setUniverse$ at least once (in effect
  winning head-to-heat contests with all of them), and this shift
  action will have price exactly $B$. In effect, the score of $p$ will
  increase to $r+1$, the scores of candidate in $\setUniverse \cup \{d\}$ will
  decrease to $r$, and $p$ will be the unique winner.

  ``$\Leftarrow$'':
  Now assume that our instance of \probShiftBCopeland is a yes-instance
  and let $\vv{s}$ be a shift-action that involves at most $k$ voters,
  has price at most $B$, and ensures $p$'s victory. Note that by
  shifting $p$ in at most $k$ preference orders, $p$ can win at most $n+1$
  additional head-to-head contests~(only those with the candidates in $\setUniverse
  \cup \{d\}$). If, after applying $\vv{s}$, $p$~does not win the
  contest against $d$, then $p$'s score is at most $r$ and $d$'s score
  remains $r+1$ (so $p$ is not a winner). Thus we know that $p$ must
  win the head-to-head contest against $d$ and, so, $p$ must pass $d$
  in at least $k$ preference orders. Since all the voters other than those in $V' = \{ v_{2i-1} \mid 1 \leq i \leq n\}$ either already prefer $p$ to
  $d$, or rank $d$ ahead of more than $B$ other candidates that
  themselves are ranked ahead of $p$, $\vv{s}$ must involve exactly
  $k$ voters from~$V'$, and for each preference order of the voters from $V'$ 
  where $p$ is shifted, 
  $p$ must be shifted to the top (because otherwise $p$ would
  not pass~$d$). Since, in addition, $p$ must pass each candidate
  $u_i$ from $\setUniverse$ at least once (to win the head-to-head contest with
  $u_i$, to obtain the final score of $r+1$), it follows that the
  voters for which $\vv{s}$ shifts $p$ to the top position correspond
  to the sets from $\setFamily$ that form a cover of $\setUniverse$.

  The above proof covers the case of unit price functions, convex
  price functions, sortable price functions, and arbitrary price
  functions. It is easy to adapt it to work for all-or-nothing prices
  in the same way in which the proof for Borda was adapted.
\end{proof}

\end{document}